\newif\ifdraft \drafttrue \draftfalse
\documentclass[final]{dmtcs-episciences}
\usepackage[utf8]{inputenc}
\usepackage{amsthm,amsmath,amssymb}
\usepackage{xargs,calc}
\usepackage{stmaryrd}
\usepackage{xcolor}
\usepackage{subfigure}

\usepackage{pstricks}
\usepackage{xspace}
\usepackage[shortlabels]{enumitem}
\setlist{parsep=0pt, itemsep={0.2\thmspace}, topsep={0.3\thmspace}}
\setlist[enumerate]{label=\enumstyle{\alph{*}}}
\setlist[trivlist]{parsep=0pt,partopsep=0pt,topsep=\thmspace}
\newcommand{\enumstyle}[1]{{\rm(#1)}}

  \newcommand*\showkeyslabelformat[1]{}
\newcommand{\TreeScale}{0.5}

\usepackage[round]{natbib}

\title{On subtrees of the representation tree in rational base numeration systems}

\author
{
  Shigeki Akiyama\affiliationmark{1}%
  \and
  Victor Marsault\affiliationmark{2,3}%
  \and
  Jacques Sakarovitch\affiliationmark{4}%
}

\affiliation{
  University of Tsukuba, Ibaraki, Japan.\\
  University of Edinburgh, United Kingdom.\\
  University of Li\`ege, Belgium.\\
  IRIF, CNRS/Paris Diderot University, and LTCI, Telecom-ParisTech, France.
}
\received{2017-6-27}
\revised{2018-1-27}
\accepted{2018-2-23}


\makeatletter
\hypersetup{
  pdftitle={\@title{}},
  pdfauthor={Shigeki Akiyama, Victor Marsault and Jacques Sakarovitch},
  bookmarks=false,
}
\makeatother

\date{\isotoday}

\newlength{\thmspace}
\theoremstyle{plain}

\newcommand{\fa}{\forall}
\newcommand{\ext}{\exists}
\newcommand{\e}{\text{\quad}}                 % un moins petit espace
\newcommand{\ee}{\text{\qquad}}               % un espace
\newcommand{\eee}{\text{\qquad \qquad}} % et un grand
\newsavebox{\InterSymbolSpace}
\savebox{\InterSymbolSpace}{\hspace{0.125em}}
\newsavebox{\SideFormulaSpace}
\savebox{\SideFormulaSpace}{\hspace{0.2em}}
\newcommand{\msp}{\usebox{\SideFormulaSpace}} % espace pour faire ressortir
\newcommand{\xmd}{\usebox{\InterSymbolSpace}} % espace entre les symboles
\newcommand{\eqpnt}{\makebox[0pt][l]{\: .}}
\newcommand{\eqpntvrg}{\makebox[0pt][l]{\: ;}}
\newcommand{\eqvrg}{\makebox[0pt][l]{\: ,}}

\newcommand{\quantvrg}{\, , \;}
\newcommand{\quantsp}{\ee }
\newcommand{\quantsmsp}{\e }
\newcommand{\LatinLocution}[1]{{\itshape #1}\xspace}

\newcommand{\cf}{\LatinLocution{cf.}}

\newcommand{\UNmbb}{{\mathchoice
{\hbox{$\textstyle\rm 1\kern-0.2em I$}}%
{\hbox{$\textstyle\rm 1\kern-0.2em I$}}%
{\hbox{$\scriptstyle\rm 1\kern-0.15em I$}}%
{\hbox{$\scriptscriptstyle\rm 1\kern-0.1em I$}}%
}}
\newcommand{\Ac}{\mathcal{A}}

\newcommand{\Dc}{\mathcal{D}}

\newcommand{\Tc}{\mathcal{T}}

\newlength{\ArrowDiagSize}
\newlength{\ArrowDiagWidth}
\newpsstyle{SLDiagStyle}%
   {colsep=6ex,rowsep=5ex,nodesep=1ex,npos=.45,%
    arrows=->,linewidth=\ArrowDiagWidth,arrowsize=\ArrowDiagSize,%
        linestyle=solid,linecolor=\ArrowDiagColor}

\newenvironment{SLDiag}%
   {\psset{style=SLDiagStyle}\begin{psmatrix}}%
   {\end{psmatrix}}%
\newcommand{\CDSL}{\begin{SLDiag}}
\newcommand{\CDSLF}{\end{SLDiag}}
\newenvironment{DiagraBig}%
{\psmatrix[colsep=7ex,rowsep=6ex,arrows=->,nodesep=1ex,npos=.45]}%
{\endpsmatrix}
\newcommand{\CDB}{\begin{DiagraBig}}
\newcommand{\CDBF}{\end{DiagraBig}}
\newenvironment{DiagraSmall}%
{\psmatrix[colsep=3ex,rowsep=3ex,arrows=->,nodesep=1ex,npos=.45]}%
{\endpsmatrix}
\newcommand{\CDS}{\begin{DiagraSmall}}
\newcommand{\CDSF}{\end{DiagraSmall}}
\newcommand{\matriceuu}[1]%
    {\begin{pmatrix} #1 \end{pmatrix}}
\newcommand{\matricedd}[4]%
    {\begin{pmatrix} #1 & #2 \\ #3 & #4 \end{pmatrix}}
\newcommand{\vecteurd}[2]%
    {\begin{pmatrix} #1 \\ #2 \end{pmatrix}}
\newcommand{\ligned}[2]%
    {\begin{pmatrix} #1 & #2 \end{pmatrix}}
\newcommand{\matricett}[9]%
    {\begin{pmatrix}  #1 & #2 & #3 \\
                      #4 & #5 & #6 \\
                      #7 & #8 & #9 \end{pmatrix}}
\newcommand{\vecteurt}[3]%
    {\begin{pmatrix} #1 \\ #2 \\ #3 \end{pmatrix}}
\newcommand{\lignet}[3]%
    {\begin{pmatrix} #1 & #2 & #3 \end{pmatrix}}
\newlength{\jsWidthCol}
\newlength{\blocinterligne}
\newlength{\blocinterligned}
\newlength{\temparraycolsep}
\newlength{\longueurbloc}
\newlength{\hauteurbloc}
\newlength{\centragebloc}
\newlength{\longueurblc}
\newlength{\hauteurblc}
\newlength{\centrageblc}
\newcommand{\blocligne}[1]%
    {\framebox[\longueurbloc]{$#1$}}
\newcommand{\blocmatrice}[1]%
    {\framebox[\longueurbloc]{\rule[\centragebloc]{0mm}{\hauteurbloc}$#1$}}
\newcommand{\blocvecteur}[1]%
    {\framebox{\rule[\centragebloc]{0mm}{\hauteurbloc}$#1$}}
\newcommand{\blcligne}[1]%
    {\framebox[\longueurblc]{$#1$}}
\newcommand{\blcmatrice}[1]%
    {\framebox[\longueurblc]{\rule[\centrageblc]{0mm}{\hauteurblc}$#1$}}
\newcommand{\blcvecteur}[1]%
    {\framebox{\rule[\centrageblc]{0mm}{\hauteurblc}$#1$}}
\newcommand{\matriceddblvs}[4]%%
   {\setlength{\temparraycolsep}{\arraycolsep}%
    \setlength{\arraycolsep}{1.3pt}%
        \left (%
    \begin{array}{cc}%
                #1  & \blcligne{#2} \\
            \blcvecteur{#3} & \blcmatrice{#4}
        \end{array}%
        \right )%
    \setlength{\arraycolsep}{\temparraycolsep}%
   }%
\newcommand{\vecteurdblvs}[2]%
   {\setlength{\temparraycolsep}{\arraycolsep}%
    \setlength{\arraycolsep}{1.5pt}%
        \left (%
    \begin{array}{c}%
                #1  \\
            \blcvecteur{#2}
        \end{array}%
        \right )%
    \setlength{\arraycolsep}{\temparraycolsep}%
   }%
\newcommand{\lignedblvs}[2]%
   {\setlength{\temparraycolsep}{\arraycolsep}%
    \setlength{\arraycolsep}{1.5pt}%
        \left (%
    \begin{array}{cc}%
                #1  & \blcligne{#2}
        \end{array}%
        \right )%
    \setlength{\arraycolsep}{\temparraycolsep}%
   }%
\newcommand{\matricettblvs}[9]%%
   {\setlength{\temparraycolsep}{\arraycolsep}%
    \setlength{\arraycolsep}{1.5pt}%
        \left (%
    \begin{array}{ccc}%
                #1  & \blcligne{#2} & #3\\
            \blcvecteur{#4} & \blcmatrice{#5} & \blcvecteur{#6}\\
                #7  & \blcligne{#8} & #9\\
        \end{array}%
        \right )%
    \setlength{\arraycolsep}{\temparraycolsep}%
   }%
\newcommand{\vecteurtblvs}[3]%
   {\setlength{\temparraycolsep}{\arraycolsep}%
    \setlength{\arraycolsep}{1.5pt}%
        \left (%
    \begin{array}{c}%
                #1  \\
            \blcvecteur{#2}\\
                #3
        \end{array}%
        \right )%
    \setlength{\arraycolsep}{\temparraycolsep}%
   }%
\newcommand{\lignetblvs}[3]%
   {\setlength{\temparraycolsep}{\arraycolsep}%
    \setlength{\arraycolsep}{1.5pt}%
        \left (%
    \begin{array}{ccc}%
                #1  & \blcligne{#2} & #3
        \end{array}%
        \right )%
    \setlength{\arraycolsep}{\temparraycolsep}%
   }%
\newcommand{\matricettblblvs}[9]%%
   {\setlength{\temparraycolsep}{\arraycolsep}%
    \setlength{\arraycolsep}{1.5pt}%
        \left (%
    \begin{array}{ccc}%
                #1  & \blcligne{#2} & \blcligne{#3}\\
            \blcvecteur{#4} & \blcmatrice{#5} & \blcmatrice{#6}\\
                \blcvecteur{#7}  & \blcmatrice{#8} & \blcmatrice{#9}\\
        \end{array}%
        \right )%
    \setlength{\arraycolsep}{\temparraycolsep}%
   }%
\newcommand{\vecteurtblblvs}[3]%
   {\setlength{\temparraycolsep}{\arraycolsep}%
    \setlength{\arraycolsep}{1.5pt}%
        \left (%
    \begin{array}{c}%
                #1  \\
            \blcvecteur{#2}\\
                \blcvecteur{#3}
        \end{array}%
        \right )%
    \setlength{\arraycolsep}{\temparraycolsep}%
   }%
\newcommand{\lignetblblvs}[3]%
   {\setlength{\temparraycolsep}{\arraycolsep}%
    \setlength{\arraycolsep}{1.5pt}%
        \left (%
    \begin{array}{ccc}%
                #1  & \blcligne{#2} & \blcligne{#3}
        \end{array}%
        \right )%
    \setlength{\arraycolsep}{\temparraycolsep}%
   }%
\newlength{\DefiTest}
\newlength{\DefiHeightu}\newlength{\DefiHeightd}%
\newlength{\DefiDepthu}\newlength{\DefiDepthd}%
\newcommand{\Defi}[2]%
    {%
     \settoheight{\DefiHeightu}{${\displaystyle #1}$}%
     \settodepth{\DefiDepthu}{${\displaystyle #1}$}%
     \addtolength{\DefiHeightu}{\DefiDepthu}%
     \settoheight{\DefiHeightd}{${\displaystyle #2}$}%
     \settodepth{\DefiDepthd}{${\displaystyle #2}$}%
     \addtolength{\DefiHeightd}{\DefiDepthd}%
     \left\{#1%
     \rule[-\DefiDepthd]{\DefiTest}{\DefiHeightd}%
     \xmd\right|%
     \left.%
     \rule[-\DefiDepthu]{\DefiTest}{\DefiHeightu}%
      #2\right\}%
     }
\newlength{\ColoText}% largeur de la colonne "de texte"
\newlength{\ColoFigu}% largeur de la colonne "de figure"
\newlength{\TextFiguSpace}% intervalle entre les deux colonnes
\newlength{\parindenttemp} % for indentation in minipage
\newlength{\parskiptemp} % for alinea spacing in minioage
\newlength{\fboxseptemp} % pour memoriser \fboxsep
\newcommand{\TFBoxing}{}
\newcommand{\TFVertAlig}{}
\newcommand{\LeftLarg}{}
\renewcommand{\LeftLarg}{.66}
\ifdraft\renewcommand{\TFBoxing}{\fbox}\fi
\newcommand{\TxtFg}[3]%
   {%
    \setlength{\ColoText}{#1\textwidth}%
    \setlength{\ColoFigu}{\textwidth}%
    \addtolength{\ColoFigu}{-\ColoText}%
    \addtolength{\ColoText}{-.5\TextFiguSpace}%
    \addtolength{\ColoFigu}{-.5\TextFiguSpace}%
    \ifdraft\setlength{\fboxsep}{0pt}\fi% % mod 000912, 050822
    \noi
    \TFBoxing{%
       \begin{minipage}[\TFVertAlig]{\ColoText}%
          \setlength{\parindent}{\parindenttemp}%
          \setlength{\parskip}{\parskiptemp}%
          \par\vspace*{0mm}% pour l'alignement sur le haut
             #2
       \end{minipage}%
             }%
    \hspace*{\TextFiguSpace}%
    \TFBoxing{%
       \begin{minipage}[\TFVertAlig]{\ColoFigu}%
          \par\vspace*{0mm}%
             #3%
       \end{minipage}%
             }%
    \ifdraft\setlength{\fboxsep}{\fboxseptemp}\fi%050822
   }%
\newcommand{\TextFigu}[3][\LeftLarg]%
   {\renewcommand{\TFVertAlig}{t}\TxtFg{#1}{#2}{#3}}
\newcommand{\TextFiguC}[3][\LeftLarg]%
   {\renewcommand{\TFVertAlig}{c}\TxtFg{#1}{#2}{#3}}
\newcommand{\TextFiguX}[3][\LeftLarg]
   {%
    \setlength{\ColoText}{#1\textwidth}%
    \setlength{\ColoFigu}{\textwidth}%
    \addtolength{\ColoFigu}{-\ColoText}%
    \addtolength{\ColoText}{-.5\TextFiguSpace}%
    \addtolength{\ColoFigu}{-.5\TextFiguSpace}%
    \addtolength{\ColoFigu}{\ETAExtendedLineWidth}% mod 000912,050822
    \ifdraft\setlength{\fboxsep}{0pt}\fi% % mod 000912, 050822
    \noi
    \ifodd\value{page}%
       \TFBoxing{%
          \begin{minipage}[t]{\ColoText}%
             \RstBLS% 050822
             \setlength{\parindent}{\parindenttemp}%
             \setlength{\parskip}{\parskiptemp}%
             \par\vspace*{0mm}% pour l'alignement sur le haut
                #2
          \end{minipage}%
                }%
       \hspace*{\TextFiguSpace}%
       \TFBoxing{%
          \begin{minipage}[t]{\ColoFigu}%
             \par\vspace*{0mm}%
                #3%
          \end{minipage}%
                }%
    \else
       \hspace*{-\ETAExtendedLineWidth}% mod 000912
       \TFBoxing{%
          \begin{minipage}[t]{\ColoFigu}%
             \par\vspace*{0mm}%
                #3%
          \end{minipage}%
                }%
       \hspace*{\TextFiguSpace}%
       \TFBoxing{%
          \begin{minipage}[t]{\ColoText}%
             \RstBLS% 050822
             \setlength{\parindent}{\parindenttemp}%
             \setlength{\parskip}{\parskiptemp}%
             \par\vspace*{0mm}% pour l'alignement sur le haut
                #2
          \end{minipage}%
                }%
    \fi%
    \ifdraft\setlength{\fboxsep}{\fboxseptemp}\fi%050822
   }
\newcommand{\NoteEnMarge}[1]%
   {%
    \marginpar[\begin{flushright}%
               {\sl {\scriptsize #1}}%
               \end{flushright}]%
              {\begin{flushleft}%
               {\sl {\scriptsize #1}}%
               \end{flushleft}}%
	}%
\newcommand{\Axio}[1]%
   {\pointn #1\hspace*{.1em}\jspointtiret\hspace*{.4em}\ignorespaces}
\usepackage{ifthen}
\usepackage{etoolbox}
\usepackage{xargs}
\undef\clipbox
\usepackage{adjustbox}
\usepackage{chemarrow}

\newcommandx{\newtheoremy}[3][2={}]{
  \ifthenelse{\equal{#2}{}}{
    \ifcsmacro{#1}{}{\newtheorem{#1}{#3}}
  }{
    \ifcsmacro{#1}{}{\newtheorem{#1}[#2]{#3}}
  }
}

\newcommand{\thmBlockFont}[1]{#1}

\newtheoremy{Theorem}{\thmBlockFont{Theorem}}
\newtheoremy{Corollary}[Theorem]{\thmBlockFont{Corollary}}
\renewcommand{\theTheorem}{\Roman{Theorem}}
\makeatletter
\newcommand\ifcounter[3]{\@ifundefined{c@#1}{#3}{#2}}
\makeatother
\ifcounter{thm}{}{\newcounter{thm}}
\newtheoremy{algorithm}[thm]{\thmBlockFont{Algorithm}}
\newtheoremy{corollary}[thm]{\thmBlockFont{Corollary}}
\newtheoremy{conjecture}[thm]{\thmBlockFont{Conjecture}}
\newtheoremy{definition}[thm]{\thmBlockFont{Definition}}
\newtheoremy{example}[thm]{\thmBlockFont{Example}}
\newtheoremy{lemma}[thm]{\thmBlockFont{Lemma}}
\newtheoremy{proposition}[thm]{\thmBlockFont{Proposition}}
\newtheoremy{property}[thm]{\thmBlockFont{Property}}
\newtheoremy{question}[thm]{\thmBlockFont{Question}}
\newtheoremy{remark}[thm]{\thmBlockFont{Remark}}
\newtheoremy{notation}[thm]{\thmBlockFont{Notation}}
\newtheoremy{theorem}[thm]{\thmBlockFont{Theorem}}

\newtheoremy{falsepropositionX}{\thmBlockFont{Proposition}}

\newtheoremy{falsetheoremX}{\thmBlockFont{Theorem}}
\newenvironment{falsetheorem}[1]{\renewcommand{\thefalsetheoremX}{\protect#1}\begin{falsetheoremX}}{\end{falsetheoremX}}

\newtheoremy{falsecorollaryX}{\thmBlockFont{Corollary}}

\newtheoremy{falselemmaX}{\thmBlockFont{Lemma}}

\newtheorem*{falsestatementX}{\thmBlockFont{\thestatement}}

\newcommand{\lcorollary}[1]{\label{c.#1}}

\newcommand{\ldefinition}[1]{\label{d.#1}}

\newcommand{\llemma}[1]{\label{l.#1}}
\newcommand{\lproposition}[1]{\label{p.#1}}
\newcommand{\lproperty}[1]{\label{pp.#1}}

\newcommand{\lnotation}[1]{\label{n.#1}}
\newcommand{\lsection}[1]{\label{s.#1}}

\newcommand{\lfigure}[1]{\label{f.#1}}
\newcommand{\ltheorem}[1]{\label{t.#1}}
\newcommand{\lequation}[1]{\label{eq.#1}}

\newcommand{\preprocgenref}[2]{}
\newcommand{\generalref}[2]{%
  \preprocgenref{#1}{#2}%
  \ifthenelse{\equal{#1}{eq}}%
  {(\ref{#1.#2})}%
  {\ref{#1.#2}}%
}
\newcommand{\generalpageref}[2]{\pageref{#1.#2}}

\makeatletter
\newcommand*{\ralgorithm}{\@ifstar{\generalref{a}}{Algorithm~\ralgorithm*}}
\newcommand*{\palgorithm}{\@ifstar{\generalpageref{a}}{page~\palgorithm*}}

\newcommand*{\rcorollary}{\@ifstar{\generalref{c}}{Corollary~\rcorollary*}}
\newcommand*{\pcorollary}{\@ifstar{\generalpageref{c}}{page~\pcorollary*}}

\newcommand*{\rconjecture}{\@ifstar{\generalref{cj}}{Conjecture~\rconjecture*}}
\newcommand*{\pconjecture}{\@ifstar{\generalpageref{cj}}{page~\pconjecture*}}

\newcommand*{\rdefinition}{\@ifstar{\generalref{d}}{Definition~\rdefinition*}}
\newcommand*{\pdefinition}{\@ifstar{\generalpageref{d}}{page~\pdefinition*}}

\newcommand*{\rexample}{\@ifstar{\generalref{e}}{Example~\rexample*}}
\newcommand*{\pexample}{\@ifstar{\generalpageref{e}}{page~\pexample*}}

\newcommand*{\rlemma}{\@ifstar{\generalref{l}}{Lemma~\rlemma*}}
\newcommand*{\plemma}{\@ifstar{\generalpageref{l}}{page~\plemma*}}

\newcommand*{\rproposition}{\@ifstar{\generalref{p}}{Proposition~\rproposition*}}
\newcommand*{\pproposition}{\@ifstar{\generalpageref{p}}{page~\pproposition*}}

\newcommand*{\rproperty}{\@ifstar{\generalref{pp}}{Property~\rproperty*}}
\newcommand*{\pproperty}{\@ifstar{\generalpageref{pp}}{page~\pproperty*}}

\newcommand*{\rprocedure}{\@ifstar{\generalref{pc}}{Procedure~\rprocedure*}}
\newcommand*{\pprocedure}{\@ifstar{\generalpageref{pc}}{page~\pprocedure*}}

\newcommand*{\rremark}{\@ifstar{\generalref{r}}{Remark~\rremark*}}
\newcommand*{\premark}{\@ifstar{\generalpageref{r}}{page~\premark*}}

\newcommand*{\rnotation}{\@ifstar{\generalref{n}}{Notation~\rnotation*}}
\newcommand*{\pnotation}{\@ifstar{\generalpageref{n}}{page~\pnotation*}}

\newcommand*{\rsection}{\@ifstar{\generalref{s}}{Section~\rsection*}}
\newcommand*{\psection}{\@ifstar{\generalpageref{s}}{page~\psection*}}

\newcommand*{\rtable}{\@ifstar{\generalref{t}}{Table~\rtable*}}
\newcommand*{\ptable}{\@ifstar{\generalpageref{t}}{page~\ptable*}}

\newcommand*{\rfigure}{\@ifstar{\generalref{f}}{Figure~\rfigure*}}
\newcommand*{\pfigure}{\@ifstar{\generalpageref{f}}{page~\pfigure*}}

\newcommand*{\requation}{\@ifstar{\generalref{eq}}{Equation~\requation*}}
\newcommand*{\pequation}{\@ifstar{\generalpageref{eq}}{page~\pequation*}}

\newcommand*{\rtheorem}{\@ifstar{\generalref{t}}{Theorem~\rtheorem*}}
\newcommand*{\ptheorem}{\@ifstar{\generalpageref{t}}{page~\ptheorem*}}

\newcommand*{\rclaim}{\@ifstar{\generalref{cl}}{Claim~\rclaim*}}
\newcommand*{\pclaim}{\@ifstar{\generalpageref{cl}}{page~\pclaim*}}
\makeatother

\usepackage{boites}
\makeatletter
 \newdimen\bk@hauteurcourrante
  \newdimen\bk@hauteursuivante
  \newdimen\bk@tempdim
\newenvironment{leftbar}{%
  \def\bk@espace{ }%
  \def\pt@to@bp##1{##1=.99627393548##1}% 1bp=1.00374pt
  \def\bkvz@before@breakbox{\ifhmode\par\fi\bk@hauteurcourrante=1200bp}%
  \def\bkvz@set@linewidth{\advance\linewidth-0.5\parindent}%
  \def\bkvz@left{\hskip 1pt\vrule\@width 0.5pt\hskip0.5\parindent\hskip -1.5pt}%
  \let\bkvz@right\relax
  \let\bkvz@top\relax
  \let\bkvz@bottom\relax
  \breakbox}{\endbreakbox}

\makeatother

\newcommandx{\wlen}[1]{|#1|}

\newcommandx{\cod}[2][2={}]{\ifthenelse{\equal{#2}{}}{\langle #1 \rangle}{\langle #1 \rangle_{#2}}}
\newcommandx{\floor}[1]{\lfloor #1 \rfloor}
\newcommandx{\bfloor}[1]{\left\lfloor #1 \right\rfloor}
\newcommandx{\bceil}[1]{\left\lceil #1 \right\rceil}
\newcommandx{\ceil}[1]{\lceil #1 \rceil}

\newcommandx{\newcommandy}[5][1=i,3=0,4={}]{%
  \ifthenelse{\isundefined{#2}}{\newcommandx{#2}[#3][#4]{#5}}{%
      \ifthenelse{\equal{#1}{i}}{}{}%
      \ifthenelse{\equal{#1}{o}}{\renewcommandx{#2}[#3][#4]{#5}}{}%
    }%
}

\newcommandx{\yrightarrow}[4][1=\empty, 2=\empty, 4=\empty, usedefault=@]{%
  \ifthenelse{\equal{#1}{\empty}}%
  {% there's no text below
    \xrightarrow{~\adjustbox{raise={-#4}{\height-#4}{0pt},trim=0pt 0pt 0pt 1pt}{\ensuremath{\scriptstyle#3}}~}%
  }{%\adjustbox{margin=2pt 0pt 1.5pt 0pt,trim=0pt 3pt 0pt 0pt}{
    \adjustbox{trim=0pt 2pt 0pt 0pt}{\ensuremath{\xrightarrow%there is text below
    [\,~\adjustbox{scale=0.9,raise={#2}{\height}}{\ensuremath{\scriptstyle#1}}~\,]%2
    {~\adjustbox{raise={-#4}{\height-#4}{0pt},trim=0pt 0pt 0pt 1pt}{\ensuremath{\scriptstyle#3}}~}}}%
  }%
}
\newcommandx{\ylefttarrow}[4][1=\empty, 2=\empty, 4=\empty, usedefault=@]{%
  \ifthenelse{\equal{#1}{\empty}}%
  {% there's no text below
    \xleftarrow{~\adjustbox{raise={-#4}{\height-#4}{0pt},trim=0pt 0pt 0pt 1pt}{\ensuremath{\scriptstyle#3}}~}%
  }{%\adjustbox{margin=2pt 0pt 1.5pt 0pt,trim=0pt 3pt 0pt 0pt}{
    \adjustbox{trim=0pt #2 0pt 0pt}{\ensuremath{\xleftarrow%there is text below
    [~\adjustbox{scale=0.9,raise={#2}{\height}}{\ensuremath{\scriptstyle#1}}~]%2
    {~\adjustbox{raise={-#4}{\height-#4}{0pt},trim=0pt 0pt 0pt 1pt}{\scriptstyle#3}~}}}%
  }%
}
\newcommandy[o]{\pathx}[2][2=\empty]{{\let\rightarrow\chemarrow%
  \nlb%
  \hspace*{1mm}\yrightarrow[#2][3pt]{\minwidthbox{$\scriptstyle#1$}{5mm}\hspace*{2pt}}[3.5pt]\hspace*{1mm}%
  \nlb%
}}
\newcommandy[o]{\pathy}[2][2=\empty]{{%
  \nlb%
  \hspace*{1mm}\ylefttarrow[#2][3pt]{\hspace*{2pt}\minwidthbox{$\scriptstyle#1$}{5mm}}[3.5pt]\hspace*{1mm}%
  \nlb%
}}
\newcommand*{\minwidthbox}[2]{%
  \makebox[{\ifdim#2<\width\width\else#2\fi}]{#1}%
}

\makeatletter

\newlength{\vm@xmd@d}
\newlength{\vm@xmd@n}
\newlength{\vm@xmd@s}
\newlength{\vm@xmd@ss}
\newcommandy[o]{\xmd}{%
  \ifmmode%
    \mathchoice%
    {\hspace*{\vm@xmd@d}}% display style
    {\hspace*{\vm@xmd@n}}% text style
    {\hspace*{\vm@xmd@s}}% script style
    {\hspace*{\vm@xmd@ss}}% scriptscript style
  \else%
    \hspace*{\vm@xmd@n}%
  \fi%
}
\makeatother
\newcommand{\transpair}[2]{ #1 \xmd | \xmd #2 }

\newcommand{\val}[1]{\widebar{#1}}
\newcommand{\card}[1]{{\tt Card}(#1)}

\newcommand{\pref}[1]{\text{Pref}\,(#1)}

\newcommand{\set}[1]{%
  \left\{\mathchoice%
  {\halfspace #1 \halfspace}%
  {\thirdspace #1 \thirdspace}%
  {#1}%
  {#1}\right\}%
}

\newcommandy[o]{\Z}{\mathbb{Z}}
\newcommandy[o]{\N}{\mathbb{N}}
\newcommandy[o]{\Q}{\mathbb{Q}}
\newcommandy[o]{\R}{\mathbb{R}}
\newcommand{\widebar}{\overline}

\newcommandy[o]{\Cup}{\bigcup}
\newcommand{\nlb}{\nolinebreak}

\renewcommand{\thmBlockFont}[1]{\ssc{#1}}

\newcommand{\vmfiguretodo}[2][]{%
  \begin{figure}[ht!]
    \frame{%
      \begin{minipage}{\linewidth}
        ~\hfill~
        \vspace*{#2cm}
      \end{minipage}
    }
    \ifthenelse{\equal{#1}{}}{}{\caption{#1}}
  \end{figure}
}

\makeatletter
\newcommandx{\newcommandWithStar}[3][1=i]{%
  \newcommandy[#1]{#2}{\protect\@ifstar{\leavevmode\protect\nlb$\protect#3$}{#3}}
}
\makeatother

\makeatletter
\newcommand{\vm@date@separator}{\hspace*{0.15ex}\rule[0.4\vm@date@height]{1ex}{0.07\vm@date@height}\hspace*{0.15ex}}
\newcommand{\vmdatefont}[1]{#1}
\newcommand{\isotoday}{%
  \vmdatefont{
    \newdimen\vm@date@height%
    \setbox0=\hbox{0123456789}%
    \vm@date@height=\ht0 \advance\vm@date@height by -\dp0
    \the\year\vm@date@separator\two@digits{\month}\vm@date@separator\two@digits{\day}%
  }%
}

\makeatother

\newcommand{\vmfbox}[1]{{%
  \fboxsep=0pt%
  \ifmmode%
    \mathchoice%
      {\fbox{$\displaystyle#1$}}%
      {\fbox{$\textstyle#1$}}%
      {\fbox{$\scriptstyle#1$}}%
      {\fbox{$\scriptscriptstyle#1$}}%
  \else%
    \fbox{#1}%
  \fi%
}}

\newcommand{\halfspace}{\hspace{0.5\fontdimen2\font plus 0.5\fontdimen3\font
minus 0.5\fontdimen4\font}}
\newcommand{\thirdspace}{\hspace{0.33\fontdimen2\font plus 0.33\fontdimen3\font
minus 0.33\fontdimen4\font}}

\renewcommand{\leq}{\leqslant}
\renewcommand{\geq}{\geqslant}
\renewcommand{\phi}{\varphi}
\renewcommand{\epsilon}{\varepsilon}
\renewcommand{\mod}{\text{~mod~}}

\newtheoremy{Theorem}{\thmBlockFont{Theorem}}
\renewcommand{\thmBlockFont}[1]{\textbf{#1}}
\renewcommand{\theTheorem}{\Roman{Theorem}}

\newcommand{\Der}{\xi}

\newcommand{\pq}{\frac{p}{q}}
\newcommand{\Lpq}{L_\base}
\newcommand{\base}{z}
\newcommand{\Tpq}{\mathcal{T}_\base}
\newcommand{\Tpqp}{{\mathcal{S}_\base}}
\newcommand{\Dpq}{\mathcal{D}_{\hspace*{-0.1ex}\base}}
\newcommand{\Dpqi}{\mathcal{D}_{\hspace*{-0.1ex}\base,i}}
\newcommand{\Spq}{\mathrm{Span}_\base}
\newcommand{\Wpq}{\ows{W}_{\hspace*{-1pt}\base}}
\newcommand{\Wpqp}{\ibehav{\Tpqp}}
\newcommand{\taupq}{\mathsf{\tau}_\base}

\newcommand{\td}{\frac{3}{2}}
\newcommand{\Ltd}{L_\td}
\newcommand{\Ttd}{\mathcal{T}_\td}
\newcommand{\Ttdp}{\mathcal{S}_\td}
\newcommand{\Dtd}{\mathcal{D}_\td}
\newcommand{\Wtd}{\ows{W}_\td}

\newcommand{\qt}{\frac{4}{3}}

\newcommand{\Tqt}{\mathcal{T}_\qt}
\newcommand{\Tqtp}{\mathcal{S}_\qt}
\newcommand{\Dqt}{\mathcal{D}_\qt}

\newcommand{\st}{\frac{7}{3}}
\newcommand{\Lst}{L_\st}
\newcommand{\Tst}{\mathcal{T}_\st}
\newcommand{\Tstp}{\mathcal{S}_\st}
\newcommand{\Dst}{\mathcal{D}_\st}
\newcommand{\Wst}{\ows{W}_\st}

\newcommand{\Btd}{D_{\td}}
\newcommand{\Bqt}{D_{\qt}}
\newcommand{\Bst}{D_{\st}}
\newcommand{\Amax}{C_z}
\newcommand{\Amaxs}{{C_z}^{\!*}}
\newcommand{\Amaxo}{{C_z}^{\!\omega}}

\newcounter{subthm}[thm]
\newcounter{cond}[thm]
\newcounter{outthm}
\renewcommand{\theoutthm}{\thethm}

\newcommand{\sublabel}[1]{\refstepcounter{subthm}\label{#1}{\renewcommand\showkeyslabelformat[1]{}\refstepcounter{cond}\label{#1*}}}

\newcommand{\Ap}{A_p}
\newcommand{\Aq}{B_\base}
\newcommand{\Bpq}{D_\base}
\newcommand{\Aps}{{\Ap}^{\!*}}
\newcommand{\Aqs}{{\Aq}^{\!*}}
\newcommand{\Bpqs}{{\Bpq}^{\!*}}
\newcommand{\Apo}{{\Ap}^{\!\omega}}
\newcommand{\Aqo}{{\Aq}^{\!\omega}}
\newcommand{\Bpqo}{{\Bpq}^{\!\omega}}

\newcommand{\mathbsf}[1]{\boldsymbol{\mathsf{#1}}}
\newcommand{\minword}[1]{\ow{w}_{#1}^{-}}%{\mathop{\ow{b}(#1)}}
\newcommand{\maxword}[1]{\ow{w}_{#1}^{+}}%{\mathop{\ow{t}(#1)}}
\newcommand{\spanword}[1]{\mathop{\ow{s}(#1)}}
\newcommand{\minwords}{\ows{Bot}_{\base}}
\newcommand{\maxwords}{\ows{Top}_{\base}}
\newcommand{\spanwords}{\ows{Spw}_{\base}}
\newcommand{\Suf}[1]{\ows{Suf}\ifthenelse{\equal{#1}{\empty}}{}{(#1)}}

\newcommand{\aut}[1]{\left\langle #1 \right\rangle}

\newcommand{\rempfun}{\psi}

\newcommand{\powerset}{\mathfrak{P}}

\renewcommandx{\val}[3][1=s,3={\base}]{%
  \pi^{~}_{#3}%
  \ifthenelse{\equal{#2}{}}%
  {}%
  {%
    \ifthenelse{\equal{#1}{b}}%
      {\hspace*{-1mm}\left(#2\right)}%
      {(#2)}%
  }
}
\newcommandx{\realval}[3][1=s,3={\base}]{%
  \rho^{~}_{#3}%
  \ifthenelse{\equal{#2}{}}%
  {}%
  {%
    \ifthenelse{\equal{#1}{b}}%
      {\hspace*{-1mm}\left(#2\right)}%
      {(#2)}%
  }
}

\newcommand{\abs}[1]{\mathsf{Abs}\left(#1\right)}

\newcommand{\rad}[1]{\mathrel{{#1}_{\text{rad}}}}
\newcommand{\lex}[1]{\mathrel{{#1}_{\text{lex}}}}

\newcommand{\behav}[1]{L(#1)}
\newcommand{\ibehav}[1]{\Lambda(#1)}

\newcommand{\Vn}{\ows{V}_n}

\newcommand{\spann}[1][n]{\sigma(#1)}

\newcommand{\adh}[1]{c\ell(#1)}

\newcommand{\Zu}[1][u]{\ows{Z}_{#1}}
\newcommand{\Iu}[1][u]{I_{#1}}
\newcommand{\II}{\mathfrak{I}}
\newcommand{\pII}{\powerset(\II)}
\newcommand{\realvalue}{a.r.p.\@ value\xspace}
\newcommand{\realvalues}{a.r.p.\@ values\xspace}

\newcommand{\itrans}{\textsf{refine}}

\newcommand{\pS}{\mathbb{S}}

\newcommand{\oword}{{$\omega$-}word\xspace}
\newcommand{\owords}{{$\omega$-}words\xspace}
\newcommand{\olanguage}{{$\omega$-}language\xspace}
\newcommand{\olanguages}{{$\omega$-}languages\xspace}
\newcommand{\orun}{{$\omega$-}run\xspace}

\newcommand{\opath}{branch\xspace}

\newcommand{\omegapq}{\omega_\base}
\newcommand{\gammapq}{\gamma_\base}

\newcommand{\dwminus}{\mathbin{\hspace*{1pt}\ominus\hspace*{1pt}}}
\newcommand{\dwplus}{\mathbin{\hspace*{1pt}\oplus\hspace*{1pt}}}

\DeclareFontFamily{U}{arr}{\hyphenchar\font=-1}
\DeclareFontShape{U}{arr}{m}{n}{ <-> arrow}{}

\renewcommand{\cod}[1]{\mathord{\left\langle#1\right\rangle_{\!\base}}}

\newcommand{\ow}[1]{{\mathbsf{#1}}}
\newcommand{\ows}[1]{\mathsf{#1}}

\newcommand{\mes}{\hspace{1.5pt plus 1pt minus 1.5pt}}

\renewcommand{\pref}[1]{\textsc{Pre}\ifthenelse{\equal{#1}{\empty}}{}{(#1)}}

\newcommand{\hausmeas}[2][d]{\mathcal{H}^{#1}(#2)}
\newcommand{\hausdim}[1]{\textsf{dim}_{\mathcal{H}}(#1)}

\DeclareMathSymbol{\in}{\mathbin}{symbols}{"32}
\newcommand{\noneqspacing}{
  \thickmuskip=5mu plus 3mu minus 2mu
  \medmuskip=3mu plus 2mu minus 2mu
}
\noneqspacing
\newcommand{\eqspacing}{
  \thickmuskip=12mu plus 3mu minus 3mu
  \medmuskip=5mu plus 3mu minus 2mu
}%
\AtBeginEnvironment{gather*}{\eqspacing}%
\AtBeginEnvironment{gather}{\eqspacing}%
\AtBeginEnvironment{equation}{\eqspacing}%
\AtBeginEnvironment{equation*}{\eqspacing}%
\AtBeginEnvironment{align}{\eqspacing}%
\AtBeginEnvironment{align*}{\eqspacing}%
\AtBeginEnvironment{multline}{\eqspacing}%
\AtBeginEnvironment{multline*}{\eqspacing}%

\begin{document}

\publicationdetails{20}{2018}{1}{10}{3742}

\maketitle

\begin{abstract}
Every rational number~$\pq$ defines a rational base numeration system
in which every integer has a unique finite representation, up to leading zeroes.
This work is a contribution to the study of the set of the representations
of integers.
This prefix-closed subset of the free monoid is naturally represented as
a highly non-regular tree.
Its nodes are the integers, its edges bear labels taken in~$\{0,1,\ldots,p-1\}$,
and its subtrees are all distinct.
%

% This work is a contribution to the study of the representations of integers in a rational base numeration system~$\pq$.

%
We associate with each subtree (or with its root~$n$) three infinite words.
The bottom word of~$n$ is the lexicographically smallest word that is the label
of a branch of the subtree.
The top word of~$n$ is defined similarly.
The span-word of~$n$ is the digitwise difference between the latter and the former.
First, we show that the set of all the span-words is accepted by an infinite
automaton whose underlying graph is essentially the same as the tree itself.
Second, we study the function that computes for
all~$n$ the bottom word associated with~$n+1$ from the one associated
with~$n$, and show that it is realised by an infinite sequential transducer whose
underlying graph is once again essentially the same as the tree itself.
An infinite word may be interpreted as an expansion in base~$\pq$ after the radix point,
hence evaluated to a real number.
If~$T$ is a subtree whose root is~$n$, then the evaluations
of the labels of the branches of~$T$ form an interval of~$\R$.
The length of this interval is called the span of~$n$ and is equal to the
evaluation of the span-word of~$n$.
The set of all spans is then a subset of~$\R$ and we use the preceding
construction to study its topological closure.
We show that it is an interval when~$ p \leq 2\xmd q -1$, and
a Cantor set of measure zero otherwise.
\end{abstract}

\clearpage

%%%%%%%%%%%%%%%%%%%%%%%%%%%%%%%%%%%%%%%%%%%%%%%%%%%%%%%%%%%%%%%%%%%%%%%%%%%%%%%%%%%%%%%%%%%%%%%%%%%
%%%%%%%%%%%%%%%%%%%%%%%%%%%%%%%%%%%%%%%%%%%%%%%%%%%%%%%%%%%%%%%%%%%%%%%%%%%%%%%%%%%%%%%%%%%%%%%%%%%
%%%%%%%%%%%%%%%%%%%%%%%%%%%%%%%%%%%%%%%%%%%%%%%%%%%%%%%%%%%%%%%%%%%%%%%%%%%%%%%%%%%%%%%%%%%%%%%%%%%

%%% some additional macros for introduction AND conclusion
%%% and some typographical modifications to macros
%
\newcommand{\dex}[1]{\xmd(#1)}
\newtheoremy{Problem}[Theorem]{\thmBlockFont{Problem}}
\newtheoremy{problem}[thm]{\thmBlockFont{Problem}}
\newcommand{\lproblem}[1]{\label{pb.#1}}
\newcommand*{\rproblem}{\@ifstar{\generalref{pb}}{Problem~\rproblem*}}
\newcommand{\En}{{\mathcal{E}_n}}
\renewcommand{\Tpq}{\mathcal{T}_{\textstyle{\base}}}
\renewcommand{\Tpqp}{{\mathcal{S}_{\textstyle{\base}}}}
\renewcommand{\Dpq}{\mathcal{D}_{\hspace*{-0.1ex}{\textstyle{\base}}}}
\renewcommand{\Dpqi}{\mathcal{D}_{\hspace*{-0.1ex}{\textstyle{\base,i}}}}
\renewcommand{\Spq}{\mathbf{Span}_{\textstyle{\base}}}
\renewcommand{\Wpq}{\ows{W}_{\hspace*{-1pt}{\textstyle{\base}}}}
\renewcommand{\Lpq}{L_{{\textstyle{\base}}}}

\section{Introduction}

The purpose of this work is a further exploration and a better
understanding of the set of \emph{infinite words} that appear in the
definition of rational base numeration systems.
These numeration systems have been introduced and studied
by~\citet*{AkiyEtAl08}, leading to some progress and results in a
number theoretic problem related to the distribution modulo~$1$ of
the powers of rational numbers and usually known as Mahler's problem
\citep{Mahl68}.
Besides these results, these systems raise many new and fascinating
problems.

We give later the precise definition of rational base numeration systems
and of the representation of numbers (integers and reals) in such
systems.
But one can hint at the results established in this paper by just
looking at the figure showing the `representation tree' in a rational
base numeration system (\rfigure{rep-tre-sys}\dex{b} for the base~$\td$)
and by comparison with the representation tree in a integer base
numeration system (\rfigure{rep-tre-sys}\dex{a} for the base~$3$).
In these trees, nodes are the natural integers, and the label of the
path from the root to an integer~$n$ is the \emph{representation}
of~$n$ in the system, whereas the label of an infinite branch gives
the representation in the system of a real number, indeed, and
because the trees are drawn in a fractal way, of \emph{the} real
number which is the ordinate of the point where the branch ends.

% In most cases, a numeration system is associated with the set of finite
% words that represent the integers, the \emph{representation language}
% and often also with the set of infinite words that represent the real
% numbers.

%
\begin{figure}[h]
  \hfill%
  \subfigure[Integer base~$3$]{%
    \lfigure{rep-tre-sys-t}%
    \includegraphics[width=0.35\linewidth]{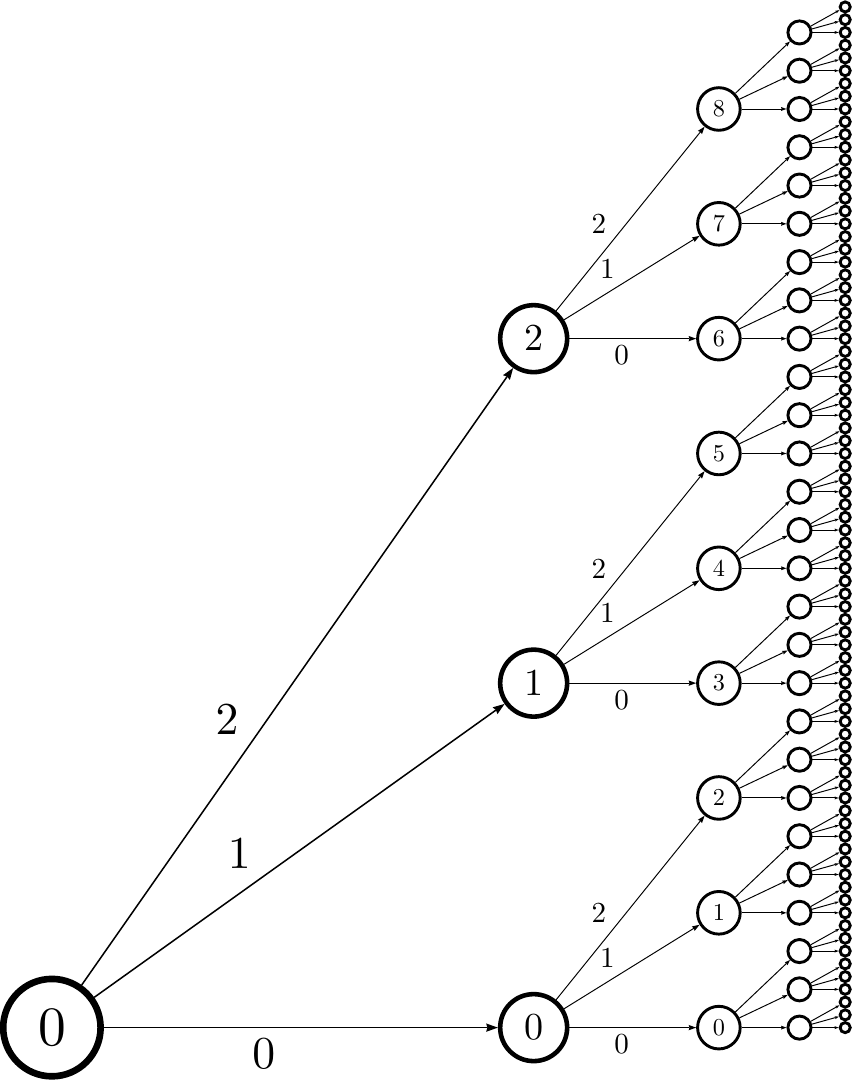}}%
  \hfill\hfill%
  \subfigure[Rational base~$\td$]{%
    \lfigure{rep-tre-sys-td}%
    \includegraphics[width=0.35\linewidth]{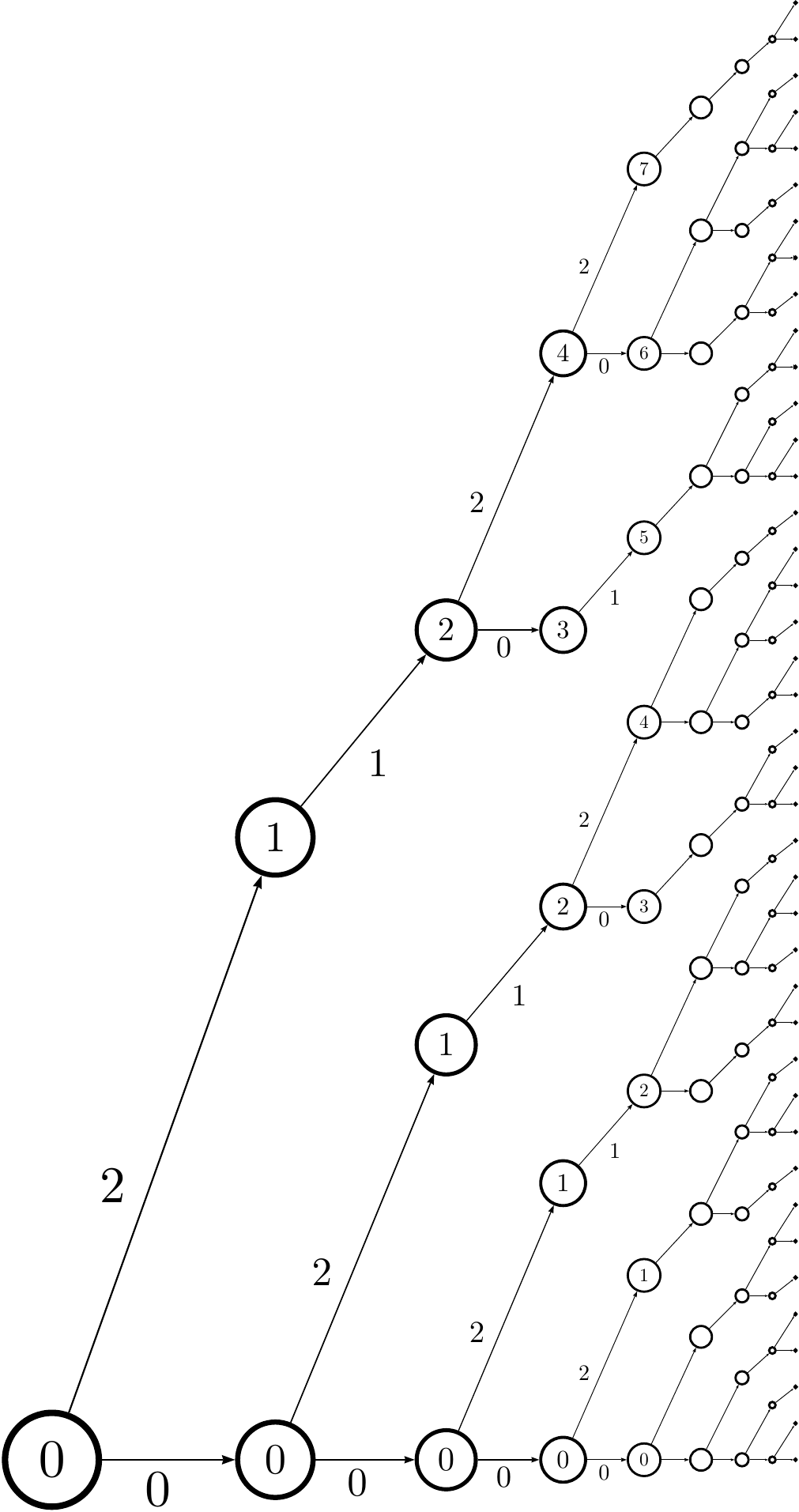}}%
  \hfill\hspace*{0cm}
  \caption{Representation trees in two number systems}
  \lfigure{rep-tre-sys}
\end{figure}
The first striking fact is that the \emph{representation language},
that is, the set of representations of integers, in a rational base
numeration system does not fit at all in the usual classifications of
formal language theory.
It looks very chaotic and defeats any kind of iteration lemma.
Nevertheless, these representation languages hide a certain kind of regularity
and we have shown \citep{MarsSaka17a} that they are so to speak characterized by
their \emph{periodic signatures}, that is, if one of these languages
is drawn as a tree and traversed breadth-first, the degrees of the nodes
are periodic.
%
% the signature of a prefix-closed language
% being the sequence of the degrees of nodes in a breadth-first traversal of the tree that represents the language.
%

If we now turn to the infinite branches of the trees, we first
find that every subtree in the tree of \rfigure{rep-tre-sys}\dex{a}
is the full ternary tree, whereas every subtree in the tree of
\rfigure{rep-tre-sys}\dex{b} is different from all other subtrees.
%
% As a result, the language of the representations of the integers is
%   not only a non regular language, but the situation is even worse as
%   this language indeed satisfies no iteration
%   lemma of any kind (\cite{MarsSaka13b}).
With the hope of finding some order or regularity within what seems to
be close to complete randomness (which, on the other hand, is not
established either and would be a very interesting result) we consider
the \emph{minimal words}, that we rather call \emph{bottom words},
originating from every node of the tree.

In the case of an integer base, this is perfectly uninteresting: all
these bottom words are equal to~$0^{\omega}$.
In the case of a rational base these words are on the contrary all
distinct, none are even ultimately periodic (as the other infinite
words in the representation tree).
In order to find some invariant of all these distinct words, or at
least a relationship between them, we have studied the function~$\Der$
that maps the bottom word~$\minword{n}$ associated with~$n$
onto~$\minword{n+1}$, the one associated with~$n+1$.
This function~$\Der$ is easily seen to be \emph{online} and
\emph{realtime}, that is, the knowledge of the first~$i$ digits of the
input is enough to compute the first~$i$ digits of the output, and
hence~$\Der$ is computable by an infinite sequential
letter-to-letter transducer.

The computation of such a transducer in the case the base~$\td$, and
more generally in the case of a base~$z=\pq$ with~$p=2\xmd q-1$, leads to a
surprising and unexpected result.
The transducer, denoted by~$\Dpq$, is obtained by replacing in the
representation tree, denoted by~$\Tpq$, the label of every edge by a
set of pairs of letters that depends upon this label only.
In other words, the \emph{underlying graphs} of~$\Tpq$ and~$\Dpq$
\emph{coincide}, and~$\Dpq$ is obtained from~$\Tpq$ by a
\emph{substitution} from the alphabet of digits into the alphabet of
pairs of digits, in this special and remarkable case.

The general case is hardly more difficult to describe, once it has
  been understood.
In the special case, the canonical digit alphabet has $p=2\xmd q-1$
  elements; in the general case, we still consider a digit alphabet
  with~$2\xmd q-1$ elements denoted by~$\Bpq$, either by keeping the
  larger~$2\xmd q-1$ elements of
  the canonical digit alphabet, when~$p$ is is greater than~$2\xmd q-1$,
  or by enlarging the canonical alphabet with enough negative digits,
  when~$p$ is smaller than~$2\xmd q-1$; in both cases, $p-1$ is the
  largest digit.

From~$\Tpq$ and with the digit alphabet~$\Bpq$, we then define
  another `representation graph' denoted by~$\Tpqp$: either by
  \emph{deleting the edges}
  of~$\Tpq$ labelled by digits that do not belong to~$\Bpq$ in the
  case where~$p>2\xmd q-1$ or, in the case where~$p<2\xmd q-1$ by \emph{adding
  edges} labelled with the new negative digits.
Then,~$\Dpq$ is obtained from~$\Tpqp$ exactly as above, by a
  \emph{substitution} from the alphabet of digits into the
  alphabet of pairs of digits.
This construction of~$\Dpq$, and the proof of its correctness yields:

\begin{Theorem}
\ltheorem{dpq=xi}~ %linebreak to have the statement on one line

Let~$p,q$ be two coprime integers such that~$p>q>1$ and~$z=\pq$.
Then~$\Dpq$ realises~$\Der$.
\end{Theorem}

In the original article \citep{AkiyEtAl08}, the tree~$\Tpq$, which is built from the representations
  of integers, is used to \emph{define} the representations of real
  numbers: the label of an infinite branch of the tree is the
  development `after the radix point' of a real number and the
  drawing of the tree as a fractal object --- like in
  \rfigure{rep-tre-sys} --- is fully justified by this point of view.
The same idea leads to the definition of the (normalised%
\footnote{%
   The classical definition of span of the node~$n$ is, in the fractal
   drawing, the width of the subtree rooted in~$n$.
   This value is obviously decreasing (exponentially) with the depth
   of the node~$n$, hence the span of two nodes cannot be easily
   compared.
   In this work, we only consider the \emph{normalised
   span} which is the span multiplied by~$(\pq)^k$, where~$k$ is the
   depth of the node~$n$.})
\emph{span} of a node~$n$ of the representation tree: it is the
difference between the real numbers represented respectively by the
top and the bottom words originating in the node~$n$
and let us denote by~$\Spq$ the set of spans for all integers and
by~$\adh{\Spq}$ its topological closure.

Again, this notion is totally uninteresting in the case of a
numeration system with an integer base~$p$: the span of every node~$n$
is always~$1$.
And again, the notion is far more richer and complex in the case of a
rational base~$\pq$ since we establish the following.

\begin{Theorem}
\ltheorem{span}%
Let~$p,q$ be two coprime integers such that~$p>q>1$ and~$z=\frac{p}{q}$.
\begin{enumerate}%[label={$(\alph{*})$}]
\item \renewcommand{\theoutthm}{\theTheorem}\sublabel{t.span-smal}
% If~$p\leq (2 q-1)$, then the topological closure of~$\Spq$ is an interval.
If~$p\leq 2\xmd q-1$, then~$\adh{\Spq}$ is an interval.
\item \sublabel{t.span-big}
% If $p > (2 q-1)$, then the topological closure of~$\Spq$ is a Cantor set
% (\ie $\Spq$  is closed, bounded, nowhere dense and has no isolated point).
If $p > 2\xmd q-1$, then~$\adh{\Spq}$ is a Cantor set of measure zero.
\end{enumerate}
\end{Theorem}

As different they may look, Theorems~I and~II have a common root
in the construction of the automaton~$\Tpqp$.
The trivial relationship between the bottom word originating at
node~$n+1$ and the top word originating at node~$n$ leads to the
connexion between the construction of the transducer~$\Dpq$ and the
description of the set of spans~$\Spq$.
The \emph{digitwise difference} between top and bottom words is
written on the alphabet~$\Bpq$, and all these `difference words' are
infinite branches in the automaton~$\Tpqp$.
This is explained in \rsection{spa-wor}.
\rtheorem{dpq=xi} is then established in \rsection{suc-fun} and
\rtheorem{span} in \rsection{on-span}.
The second case of \rtheorem{span} is completed with an upper
bound for the Hausdorff dimension of~$\adh{\Spq}$.
This paper is meant to be self-contained and starts, in particular,
with all necessary definitions concerning rational base number
systems in \rsection{rati-base}.
We conclude the paper with an open problem on minimal words which
indeed was the motivating force of all this work, and with a
conjecture on the Hausdorff dimension of~$\adh{\Spq}$.

The present article is a long version of a work \citep{AkiyEtAl13} presented
at the 9th International Conference on Words.
Most of the results are also part of the thesis of the second author~\citep{Mars16}.

%%%%%%%%%%%%%%%%%%%%%%%%%%%%%%%%%%%%%%%%%%%%%%%%%%%%%%%%%%%%%%%%%%%%%%%%%%%%%%%%%%%%%%%%%%%%%%%%%%%
%%%%%%%%%%%%%%%%%%%%%%%%%%%%%%%%%%%%%%%%%%%%%%%%%%%%%%%%%%%%%%%%%%%%%%%%%%%%%%%%%%%%%%%%%%%%%%%%%%%
%%%%%%%%%%%%%%%%%%%%%%%%%%%%%%%%%%%%%%%%%%%%%%%%%%%%%%%%%%%%%%%%%%%%%%%%%%%%%%%%%%%%%%%%%%%%%%%%%%%
\section{Preliminaries and notation}

%%%%%%%%%%%%%%%%%%%%%%%%%%%%%%%%%%%%%%%%%%%%%%%%%%%%%%%%%%%%%%%%%%%%%%%%%%%%%%%%%%%%%%%%%%%%%%%%%%%
%%%%%%%%%%%%%%%%%%%%%%%%%%%%%%%%%%%%%%%%%%%%%%%%%%%%%%%%%%%%%%%%%%%%%%%%%%%%%%%%%%%%%%%%%%%%%%%%%%%
\subsection{On words and numbers}

An \emph{alphabet} is a finite set of symbols, called \emph{letters}.
A \emph{word} (resp.\@ an \emph{\oword}) is a finite (resp.\@
infinite) sequence of letters and a language (resp.\@ an \olanguage)
is a set of words (resp.\@ \owords).
The set of the words (resp.\@ \owords) over an alphabet~$A$
is denoted by~$A^*$ (resp.~$A^\omega$).
Subsets of~$A^*$ are called \emph{languages} over~$A$ and those
of~$A^\omega$ are called \emph{\olanguages} over~$A$.
For the sake of clarity, we use the standard math font for letters
and words: $a,b,c,d,u,v,w$\ldots and
a bold sans-serif font for \owords: $\ow{u}, \ow{v},\ow{w}$\ldots %.
The \emph{length} of a word~$u$ is denoted by~$\wlen{u}$ and the
\emph{concatenation} of two words~$u$ and~$v$ is denoted simply
by~$uv$.
If~$w=uv$ (resp.~$\ow{w}=u\ow{v}$), then~$u$ is called a \emph{prefix}
of~$w$ (resp.~of~$\ow{w}$); note that the prefixes of word or of
\owords always are words.
We denote by~$\pref{}$ the
function~$A^*\cup A^\omega\rightarrow\powerset(A^*)$
that maps a word or an \oword to the set of all its
prefixes;~$\pref{}$ is naturally lifted to languages and \olanguages,
that is, to a
function~$\powerset(A^*)\cup\powerset(A^\omega)\rightarrow\powerset(A^*)$.
A language~$L$ is said \emph{prefix-closed} if~$\pref{L}=L$.
Words and \owords will later be evaluated using a rational base
numeration system (defined in \rsection{rati-base}).
It is then convenient to have a different index convention for words
and \owords: we index (finite) words \emph{from right to left} and
use~$0$ as the rightmost index (as in~$a_{k}\cdots a_1\xmd a_0$),
while \owords are indexed from left to right, starting with index~$1$
(as in~$a_1\xmd a_2\cdots$).
In this article, letters always are (relative) integers and we use
\emph{digit} as a synonym for letter.
Moreover, alphabets always are integer intervals, that is, sets of
consecutive integers.
In particular, our alphabets are totally ordered, which implies that
any set of words is equipped with two total orders: the \emph{radix
order} and the \emph{lexicographic order}:
\begin{definition}
Let~$u$ and~$v$ be words over~$A$ and~$w$ their longest common prefix.
\begin{enumerate}
\item $u\lex\leq v$ if
% A word in the \emph{lexicographic order}, denoted by~$,
%     if one of the two following conditions holds.
\begin{itemize}[nosep]
\item either~$u=w$, that is, $u$ is a prefix of~$v$,

\item or~$u=w\xmd a\xmd x$ and~$v=w\xmd b\xmd y$ with~$a,b$ in~$A$
and~$a<b$.
% where~$a,b$ are the unique digits and~$w$ the unique word such
% that~$a\neq b$,~$w\xmd a$ is a prefix of~$u$ and~$w\xmd b$ is a
% prefix of~$v$, if~$a<b$.
\end{itemize}
\item $u\rad\leq v$ if
% A word~$u$ is smaller than another word~$v$ in the \emph{radix
% order}, denoted by~$u\rad< v$
\begin{itemize}[nosep]
\item either~$\wlen{u}<\wlen{v}$ % if is~$u$ is shorter than~$v$
\item or~$\wlen{u}=\wlen{v}$ and~$u\lex\leq v$.
% if~$u$ and~$v$ are of the same length and~$u\lex<v$.
\end{itemize}
% \item The corresponding loose orders,~$\lex\leq$ and~$\rad\leq$, are
% defined as usual.
\end{enumerate}
Let~$\ow{u}$ and~$\ow{v}$ be \owords over~$A$.
\begin{enumerate}
    \setcounter{enumi}{2}%
\item $\ow{u}\lex\leq \ow{v}$ if
\begin{itemize}[nosep]
\item either~$\ow{u}=\ow{v}$,

\item or, if~$w$ (in~$A^*$) is their longest common prefix,
~$\ow{u}=w\xmd a\xmd\ow{x}$ and~$\ow{v}=w\xmd b\xmd\ow{y}$ with~$\ow{x},\ow{y}$ in~$A^\omega$ and~$a,b$ in~$A$
such that~$a<b$.
% where~$a,b$ are the unique digits and~$w$ the unique word such
% that~$a\neq b$,~$w\xmd a$ is a prefix of~$u$ and~$w\xmd b$ is a
% prefix of~$v$, if~$a<b$.
\end{itemize}
\end{enumerate}
\end{definition}
The set of \owords is classically equipped with the product topology
which can also be defined with a distance.
% Let us conclude Section~\thesubsection{} with a few words on the
% topology of infinite words.  Please see \cite[Chapter
% III]{PerrPin04} for more details.
% %
% This topology is defined by the distance~$d$ below.
% %

\begin{definition}
Let~$\ow{u},\ow{v}$ be two infinite words.
The distance between these two words is
\begin{equation*}
d(\ow{u},\ow{v}) = \left\{
\begin{array}{lp{8cm}}
   0 & if~~$\ow{u}=\ow{v}$ \\[2mm]
   2^{-\wlen{w}} & where~$w$ is the longest common prefix of~$\ow{u}$
   and~$\ow{v}$, otherwise.
\end{array}\right.
\end{equation*}
\end{definition}

% Note
%
% The distance between two \owords
% is always strictly smaller than 1, and~$A^\omega$ is dense for that
% topology.
% that every sequence with value in
% has a convergent subsequence.
%

%%%%%%%%%%%%%%%%%%%%%%%%%%%%%%%%%%%%%%%%%%%%%%%%%%%%
%%%%%%%%%%%%%%%%%%%%%%%%%%%%%%%%%%%%%%%%%%%%%%%%%%%%
\subsection{On trees, automata and transducers}

In this article we consider infinite, directed graphs of
a special form.
First, there is a special \emph{initial} vertex called the \emph{root}
and indicated by an incoming arrow in figures.
Second, the edges are labelled over a finite alphabet.
Third, they are \emph{deterministic}: there is never two different edges
originating from the same vertex and labelled by the same letter.
Such graphs are represented by quadruple~$\aut{A,V,i,\delta}$
where~$A$ is the finite alphabet,~$V$ is the (infinite)
vertex-set,~$\delta$ is a function $V\times A \rightarrow V$ is the
set of edges.
We call such graphs \emph{automata} and we use terminology of
automata theory;
in particular we use \emph{state} rather than \emph{vertex},
and \emph{transition} rather than \emph{edge}.
A transition is denoted by~$s\pathx{a}s'$,
where~$s,s'$ are states and~$a$ is a letter.
We will consider \emph{finite} and \emph{infinite} paths in these graphs.
We  refer to infinite paths as \emph{branches}
and refer to finite paths simply as \emph{paths}.
A branch is thus denoted by~$s\pathx{\ow{w}}\cdots$ and a path
by~$s\pathx{u}s'$, where~$s,s'$ denote states,~$\ow{w}$ an \oword
and~$u$ a word.
We call \emph{dead-end} a state with no outgoing transitions; in this
article, automata will have no dead-end.
A \emph{run} refers to a path starting from the root.
\emph{The run of~$u$} is the unique run labelled by~$u$ as a label, if
it exists; in which case~$u$ is said to be \emph{accepted} by the
automaton.
The language accepted by~$\Ac$, denoted by~$\behav{\Ac}$ is the set of
the words accepted by~$\Ac$.
The notions of \emph{\orun}  and \emph{accepted \olanguage}
(denoted by~$\ibehav{\Ac}$) are defined similarly.
If~$\Ac$ has no dead-end, then~$\behav{\Ac}=\pref{\ibehav{\Ac}}$.
We call \emph{tree} an automaton in which every state is reached by
exactly one run.
A \emph{transducer} is an automaton where the labels are taken in a
product alphabet~${A\times B}$;~$A$ is the \emph{input alphabet}
and~$B$ the \emph{output alphabet}.
All the transducers we consider are \emph{input-deterministic}:
if~$s\pathx{(a,b)}t$
and~$s\pathx{(a,b')}t'$ then~$b=b'$ and~$t=t'$.
They are interpreted as computing functions:
the first component is the input and the second is the output.
If~$(u,v)$ labels a run of a transducer~$\Tc$, then we say that~$v$ is
the \emph{image by~$\Tc$ of~$u$}; by abuse of language, this run will be called
 \emph{the run of~$u$}.
With the usual definition of automata and transducers (as for instance
in \citealp{Saka09}) what we call automaton is indeed an \emph{infinite
deterministic automaton with all states final} and what we call
transducer is indeed an \emph{infinite letter-to-letter pure-sequential
transducer}.

\medskip

Let us conclude this section with a statement linking the language and
the \olanguage accepted by an automaton (more details on the subject
in~\citealp{PerrPin04}).
\begin{lemma}\llemma{beha-ibeh}
  Let~$\Ac$ be an automaton with no dead-end and~$\ows{S}$
  an \olanguage.
  It holds~$\behav{\Ac} = \pref{\ows{S}}$ if and only if~$\ibehav{\Ac} = \adh{\ows{S}}$.
\end{lemma}
\begin{proof}
  Forward direction.
  Let~$\ow{w}$ be an \oword. The following sequence of equivalences
  holds.
\begin{align*}
    \ow{w} \in \ibehav{\Ac} &\e\Longleftrightarrow\e
    \pref{\ow{w}} \subseteq \behav{\Ac} \e\Longleftrightarrow\e
    \pref{\ow{w}} \subseteq \pref{\ows{S}}\\
%     \makebox[0cm][l]{~~~~\text{(forward-dir.\@ hypothesis)}} \\
  &\e\Longleftrightarrow\e
  \forall u \in \pref{\ow{w}}\quantvrg
  \exists \ow{s}_{u}\in\ow{S}\quantsp u\in\pref{\ow{s}_{u}}
  \e\Longleftrightarrow\e  \ow{w} \in \adh{\ows{S}}\eqpnt
\end{align*}
%   \begin{gather*}
%     \ow{w} \in \ibehav{\Ac} \\
%     \pref{\ow{w}} \subseteq \behav{\Ac} \\
%     \pref{\ow{w}} \subseteq \pref{\ows{S}}\makebox[0cm][l]{~~~~\text{(forward-dir.\@ hypothesis)}} \\
%     \forall u \in \pref{\ow{w}}\quantvrg \exists \ow{s}_{u} \in \ow{S} \quantsp u\in\pref{\ow{s}_{u}} \\
%     \ow{w} \in \adh{\ows{S}}
%   \end{gather*}

\medskip

\noindent
Backward direction. Let~$u$ be a word.
The following sequence of equivalences holds.
\begin{align*}
    u\in \behav{\Ac} &\e\Longleftrightarrow\e
    \exists \ow{w} \in\ibehav{\Ac}\quantsp u\in\pref{\ow{w}}
        \makebox[0cm][l]{~~~~\text{(no-dead-end hypothesis)}}\\
&\e\Longleftrightarrow\e
    \exists \ow{w} \in\adh{\ows{S}}\quantsp u\in\pref{\ow{w}}
        \makebox[0cm][l]{~~~~\text{(backward-dir.\@
	hypothesis)}}\eee\eee\\
&\e\Longleftrightarrow\e
\exists{\ow{w'}} \in \ows{S} \quantsp\e\msp u\in\pref{\ow{w'}}
        \makebox[0cm][l]{~~~~\text{(closure definition)}}\\
&\e\Longleftrightarrow\e    u\in \pref{\ows{S}}\eqpnt\e
\tag*{\qedhere}
\end{align*}
%   \begin{gather*}
%     u\in \behav{\Ac}\\
%     \exists \ow{w} \in\ibehav{\Ac}\quantsp u\in\pref{\ow{w}}
%         \makebox[0cm][l]{~~~~\text{(no-dead-end hypothesis)}}\\
%     \exists \ow{w} \in\adh{\ows{S}}\quantsp u\in\pref{\ow{w}}
%         \makebox[0cm][l]{~~~~\text{(backward-dir.\@ hypothesis)}}\\
%     \exists{\ow{w'}} \in \ows{S} \quantsp u\in\pref{\ow{w'}} \\
%     u\in \pref{\ows{S}}
%   \end{gather*}
\end{proof}
%

%%%%%%%%%%%%%%%%%%%%%%%%%%%%%%%%%%%%%%%%%%%%%%%%%%%%%%%%%%%%%%%%%%%%%%%%%%%%%%%%%%%%%%%%%%%%%%%%%%%
%%%%%%%%%%%%%%%%%%%%%%%%%%%%%%%%%%%%%%%%%%%%%%%%%%%%%%%%%%%%%%%%%%%%%%%%%%%%%%%%%%%%%%%%%%%%%%%%%%%
%%%%%%%%%%%%%%%%%%%%%%%%%%%%%%%%%%%%%%%%%%%%%%%%%%%%%%%%%%%%%%%%%%%%%%%%%%%%%%%%%%%%%%%%%%%%%%%%%%%
\section{Rational base numeration systems}
\lsection{rati-base}

In this section, we recall the definition of rational base numeration
systems that have been introduced by \citet*{AkiyEtAl08}, and the properties
of the representation trees %in rational base numeration systems
that were established in this paper.

% %
% Rational base numeration systems were introduced by two of the authors
% and Frougny \cite{AkiyEtAl08}.
% %
% This section~\thesection{} introduces no new concept and prove no new result
% on the matter, we
% simply give the definition of rational bases and restate results from \cite{AkiyEtAl08}.
% %

%
\begin{notation}
  We denote by~$p$ and~$q$ two co-prime integers such that~$p>q>1$,
  and by~$\base$ the rational number~$\base=\frac{p}{q}$.
  They will be fixed throughout the article.
\end{notation}
Note that the numeration system in base~$\pq$ we are about to describe is \emph{not}
the~$\beta$-numeration where~${\beta=\pq}$.
Indeed, in the latter, the representation of a number is
computed by a left-to-right algorithm
(called \emph{greedy}, \cf\citealp[Chapter~7]{Loth02}), the digit set is~$\set{0,1, \ldots,
\bfloor{\pq}}$ and the weight of the~$i$-th leftmost
digit is~$(\pq)^i$.
Meanwhile, in base~$\pq$, the representations are computed
by a right-to-left algorithm (\requation{MEA-alt}), digits are taken
in~$\set{0,1,\ldots,(p-1)}$ and the weight of the~$i$-th digits
is~$\frac{1}{q}(\pq)^i$.
\begin{figure}[ht!]
  \centering
  \includegraphics[scale=\TreeScale]{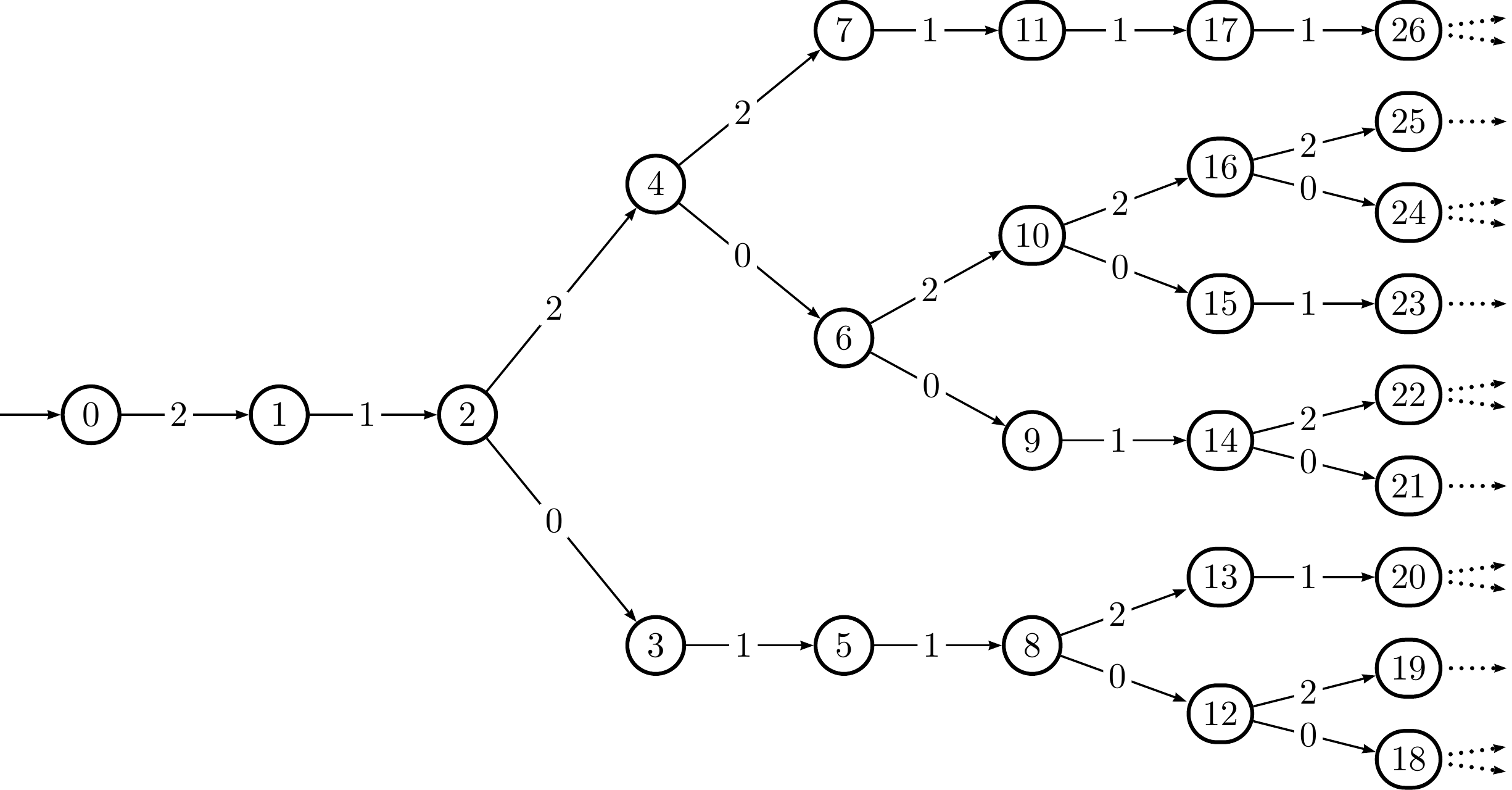}
  \caption{The language~$\Ltd$ represented as a tree}
  \lfigure{l32}
\end{figure}

%%%%%%%%%%%%%%%%%%%%%%%%%%%%%%%%%%%%%%%%%%%%%%%%%%%%%%%%%%%%%%%%%%%%%%%%%%%%%%%%%%%%%%%%%%%%%%%%%%%
%%%%%%%%%%%%%%%%%%%%%%%%%%%%%%%%%%%%%%%%%%%%%%%%%%%%%%%%%%%%%%%%%%%%%%%%%%%%%%%%%%%%%%%%%%%%%%%%%%%
\subsection{Representation of integers}

\medskip

Given a positive integer~$N$, let us define~$N_0=N$ and, for
all~$i>0$,
\begin{equation*}
  q\xmd N_i = p\xmd N_{(i+1)} + a_i
  \eqvrg
%   \notag
\end{equation*}
where~$a_i$ and~$N_{(i+1)}$ are the remainder and the quotient
of the Euclidean division of $q\xmd N_i$ by~$p$.
Hence~$a_i$ belongs to the alphabet~${\Ap=\{0,1,\ldots,p-1\}}$.
Since~$p>q$, the sequence~$(N_i)_{i\in\N}$ is first strictly decreasing
until it reaches 0:
there is an integer~$k$ such that~$N_0 > N_1 > \cdots > N_k > N_{k+1}=0$.
The word~$a_{k}\cdots a_1a_0$ of~$\Aps$ is denoted by~$\cod{N}$.
\requation{MEA-alt}, below, gives a compact definition of the same algorithm.
\begin{subequations}\lequation{MEA-alt}
  \begin{align}
  \lequation{MEA-alt-1}
    \cod{0} ={}&\epsilon \\
    \lequation{MEA-alt-2}
    \forall m\mathbin{>}0\quantsp
    \cod{m} ={}& \cod{n}\xmd a
      \qquad\text{where}\quad  n\in \N\quantvrg
                               a\in\Ap \quad\text{and}\quad q\xmd m = p\xmd n +a
  \end{align}
\end{subequations}

\medskip

\noindent
If~$\cod{N} = a_{k}a_{k-1}\cdots a_0$, then it holds
\begin{equation*}
  N = \sum^{k}_{i=0} \frac{a_i}{q} \left( \pq \right)^i
  \eqpnt
%   \notag
\end{equation*}
The \emph{evaluation function}~$\val{}$ is derived from this formula.
The \emph{value} of
any word~${a_{k}a_{k-1}\cdots a_0}$ over~$\Ap$, and indeed over any
alphabet of digits,   is defined by
\begin{equation}\label{eq.defi-pi}
  \val{a_{k}a_{k-1}\cdots a_0} =
  \sum^{k}_{i=0} \frac{a_i}{q} \left( \pq \right)^i
  \eqpnt
\end{equation}
A word~$u$ in~$\Aps$ is called a~$\pq$-\emph{expansion} of
an integer~$n$, if~${\val{u}=n}$.
Since $\pq$-expansions are unique up to leading 0's
(\cf~\citealt[Theorem~1]{AkiyEtAl08}),~$u$ is equal to~$0^i\cod{n}$ for some integer~$i$
and~$\cod{n}$ is called \emph{the~$\pq$-representation of~$n$}.
The set of the $\pq$-representations of integers is denoted by~$\Lpq$:
\begin{equation} \lequation{defi-lpq}
	\Lpq = \{\cod{n}~|~n\in\N \} \eqpnt
\end{equation}
It follows from \requation*{MEA-alt-2} that~$\Lpq$ is
prefix-closed and right-extendable.
% (for every representation~$\cod{n}$, there
% exists (at least) an~$a$ in~$\Ap$ such that~$q$ divides~$(np+a)$ and
% then~${\cod{\frac{np+a}{q}}=\cod{n}.a}$).
%
As a consequence,~$\Lpq$ can be represented as a tree with no dead-end
(cf.\@ Figures~\rfigure*{l32}, \rfigure*{l73} and later on \rfigure*{l43}).
The node set is~$\N$, the root is~$0$, and there is an arc~$n\pathx{a}m$
if~$\cod{n}a=\cod{m}$.
\begin{figure}[ht!]
  \centering
  \includegraphics[scale=\TreeScale]{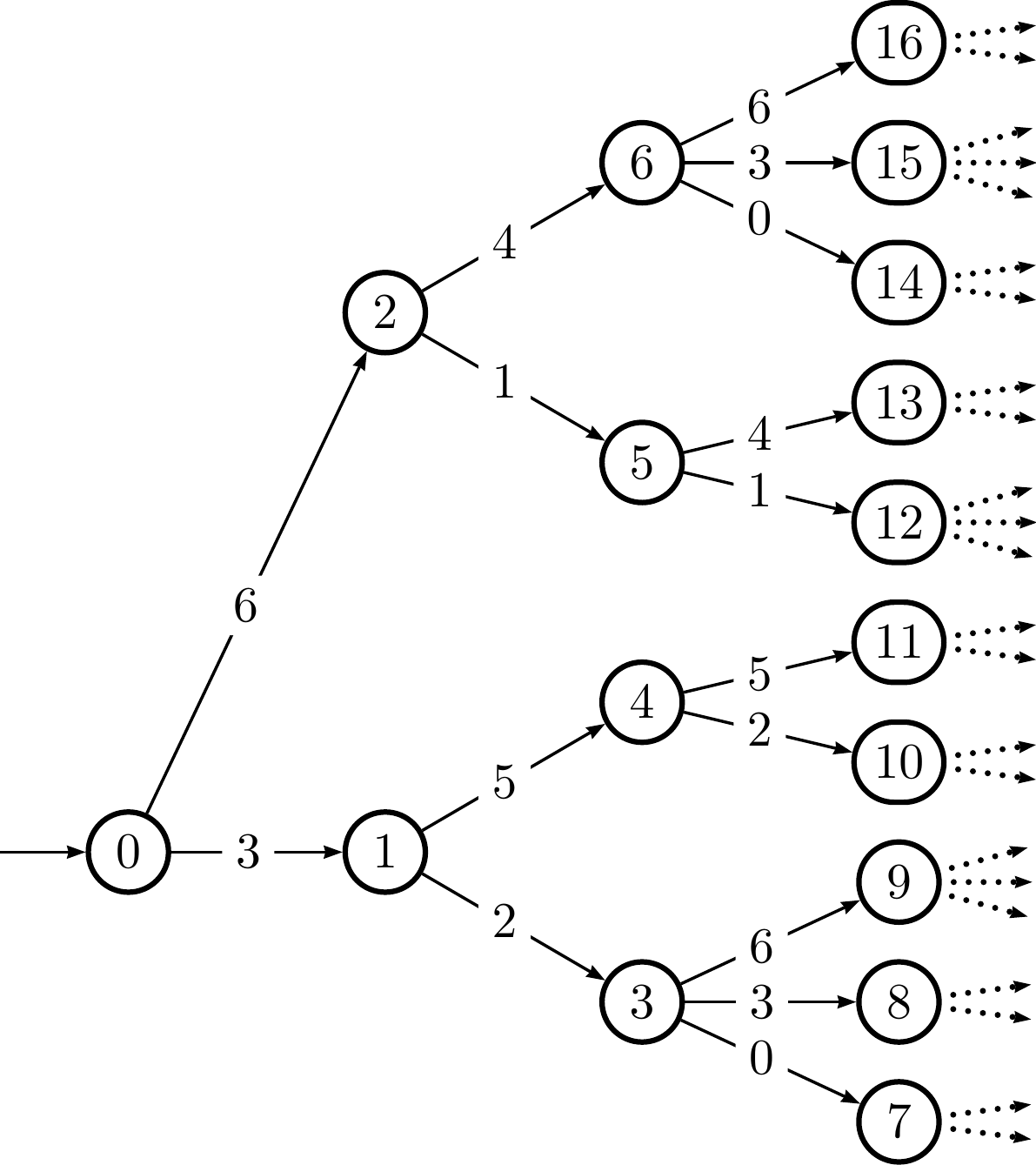}
  \caption{The language~$\Lst$ represented as a tree}
  \lfigure{l73}
\end{figure}
Moreover, the base~$\pq$ is the ``abstract numeration system'' (\cf \citealp{LecoRigo01,LecoRigo10-ib})
built from~$\Lpq$, a property that may be stated as follows:
%
% \begin{proposition}[{\cite[Proposition~11]{AkiyEtAl08}}]\lproposition{pq-ans}
%   The ordering of two integers~$n$ and~$m$ is the same as the one of~$\cod{n}$ and~$\cod{m}$
%   in the radix order.
% \end{proposition}
%
\begin{proposition}[{\citealp[Proposition~11]{AkiyEtAl08}}]
\lproposition{pq-ans}%
$\fa n,m\in\N\quantsp
n\leq m \e\Longleftrightarrow\e \cod{n}\rad\leq \cod{m}$.
\end{proposition}

\noindent
or, equivalently as:
\begin{equation}
\label{eq.pq-ans}%
\forall u,v \in\Lpq\quantsp
\val{u}\leq \val{v} \e\Longleftrightarrow\e u\rad\leq v
\eqpnt
% \notag
\end{equation}
It is known that~$\Lpq$ is not a regular language (not even a context-free language).
In fact, it even possesses a ``Finite Left Iteration Property''
which essentially says that~$\Lpq$ cannot satisfy any kind of pumping lemma.
\rlemma{tpq-flip}, later on, is a consequence of this fact.

\begin{definition}%\ldefinition{tpq}%
\begin{enumerate}
  \item \sublabel{d.taupq}
%   We denote by
  Let~$\taupq\colon\N\times\Z\rightarrow\N$ be the (partial) function
  defined by:
  \begin{equation}\label{eq.tau}
    \forall n\in\N\quantvrg
      \forall a \in \Z\quantsp
    \taupq(n,a)=\left(\frac{n\xmd p+a}{q}\right) \ee
      \text{if~$(n\xmd p+a)$ is divisible by~$q$.}
  \end{equation}

  \item \sublabel{d.tpq}
  We denote \/\footnote{%
      In \cite{AkiyEtAl08},~$\Tpq$ denotes an infinite directed
    tree.
      The labels of the (finite) paths starting from the root
    precisely formed the language~$\msp 0^*\Lpq$, as
    is~$\behav{\Tpq}$ in our case.}
  by~$\Tpq$ the \emph{infinite} automaton:~$\mes\eqspacing\Tpq=\aut{\,\Ap,\,\N,\,0,\,\taupq\,}$.
  \end{enumerate}
\end{definition}
\begin{remark}
  \begin{itemize}
    \item %In \requation{tau},
    The function~$\taupq$ is defined on~$\msp\N\times\Z\msp$
              instead of~$\msp\N\times\Ap\msp$ in anticipation of
              future developments.

    \item  The automaton~$\Tpq$ is not quite a tree.
        Indeed, the state~$0$ (that is, the root) holds a loop labelled by the digit~$0$
        since~$\taupq(0,0)=0$. % is obviously defined and equal to~$0$.
  \end{itemize}
\end{remark}
%

%
% For the sake of convenience, we give below an arithmetic characterisation
% of
%
The transitions of~$\Tpq$ are characterised by the following.
\begin{equation}\lequation{tpq-path}
\forall n,m \in \N \quantvrg \forall a\in\Ap \quantsp
n\pathx{a}[\Tpq]m \iff q\xmd m  = p\xmd n + a
\end{equation}
Comparing~\requation*{MEA-alt} and~\requation*{tpq-path}
shows how the difference between~$\Lpq$ and~$\Tpq$ is mostly a
question of formalism.
It holds $\behav{\Tpq}=\msp0^*\Lpq$ and next lemma gives a more precise
statement.
\begin{lemma} \llemma{tpq-val}
%   Let~$u\in\Aps$ be a word accepted by~$\Tpq$.
%   %
%   Its evaluation, $\val{u}$, is an integer and the run of~$u$ in~$\Tpq$
%   reaches the state~$\val{u}$.
  Let~$u$ be in~$\behav{\Tpq}$.
  Then, $\val{u}$ is in~$\N$ and
  $\msp0\pathx{u}[\Tpq]\val{u}\msp$.
\end{lemma}
\rlemma{tpq-val} implies that the tree representation of~$\Lpq$, as in Figures \rfigure*{l32},
\rfigure*{l73} and \rfigure*{l43}, augmented by an additional loop labelled
by~$0$ onto the root~$0$ becomes a representation of~$\Tpq$.
Moreover, since~$\Lpq$ is right-extendable, the next statement holds.
\begin{lemma}\llemma{tpq-righ-exte}
  $\Tpq$ has no dead-end.
\end{lemma}
We now state a few properties of~$\Tpq$.
They are the translations of results due to \cite{AkiyEtAl08} into the formalism
we use here.
\begin{lemma}[{\citealp[Lemma~6]{AkiyEtAl08}}]\llemma{tpq-futu}
  Let~$n$,~$n'$ be two integers. Let~$k$ be another integer.
  \begin{enumerate}%[label=$($\alph{*}$)$]
    \item \sublabel{l.tpq-futu->}
      If~$n$ and~$n'$ are congruent modulo~$q^k$, then for every word~$u$ of length~$k$, the following are equivalent.
      \begin{itemize}
        \item There exists an integer~$m$ such that~$n\pathx{u}m$.
        \item There exists an integer~$m'$ such that~$n'\pathx{u}m'$.
      \end{itemize}
    \item \sublabel{l.tpq-futu<-}
      If there exist two integers~$m,m'$ and a word~$u$ of length~$k$
      such that~$n\pathx{u}m$ and~$n'\pathx{u}m'$,
      then~$n$ et~$n'$ are congruent modulo~$q^k$.
  \end{enumerate}
\end{lemma}
\begin{lemma} \llemma{tpq-flip}
  Let~$n\pathx{\ow{w}}\cdots$ be a branch of~$\Tpq$.
  If~$\mes\ow{w}$ is periodic, then~${n=0}$ and~${\mes\ow{w}=0^\omega}$.
\end{lemma}
\begin{proof}
  The hypothesis implies that there is a word~$u$ such that~$n\pathx{u}m\pathx{u^\omega}\cdots$
  is a branch of~$\Tpq$.
  From \rlemma{tpq-futu<-},~$n$ and~$m$ are
  congruent modulo~$q^{\wlen{u} \times i}$ for every integer~$i$.
  Hence~$n=m$.
  The only circuit in~$\Tpq$ is~$0\pathx{0}0$, hence the statement.
\end{proof}
%yielding

%
For every integer~$k$ let us define the (total)
function~$f_{k}\colon\N\rightarrow{\Ap}^{\!k}$ in the following way.
\requation{tpq-path} implies that every state of~$\Tpq$
(incuding~$0$) has exactly one incoming transition, hence, by
induction on~$k$, exactly one incoming path of length~$k$:
for every integer~$m$,
$\msp f_k(m)=u\msp$
where~$u$ is the label of this unique path of length~$k$ ending
in~$m$.
%
% %
% Let~$k$ be an integer.
% %
% Let us define a function~$f_k:\N\rightarrow {\Ap}^{\!k}$ as follows.
% %
% Let~$m$ be an integer.
% %
% It follows from \requation{tpq-path} that every state has exactly one incoming transition.
% %
% Hence, there is a unique word~$u$ in~${\Ap}^{\!k}$ and state~$n$ such that~$n\pathx{u}m$.
% %
% We set~$f_k(m)=u$.
% %

%
\begin{lemma}[{\citealp[Proposition~10]{AkiyEtAl08}}]\llemma{tpq-past}
  Let~$m$,~$m'$ be two integers.
  For every integer~$k$, $m$ and~$m'$ are congruent modulo~$p^k$ if
  and only if~$f_k(m)=f_k(m')$.
\end{lemma}
%

% The two next two lemmas follow immediately.

%
\begin{lemma}\llemma{tpq-suff-prec}
  For every integer~$k$, $f_{k}$ is a bijection between any integer
  interval~$S$  of cardinal~$p^k$ and~${\Ap}^{\!k}$.
%   Let~$k$ be an integer and~$S$ an integer interval of cardinal~$p^k$.
%   %
%   Then~$f_k$ restricted to~$S$ is a bijection~$S\rightarrow{\Ap}^{\!k}$.
\end{lemma}
\begin{proof}
  Two integers~$m,m'$ in~$S$ are necessarily in different residue
  classes modulo~$p^k$,
  hence from \rlemma{tpq-past}, satisfy~$f_{k}(m)\neq f_{k}(m')$.
  It follows that~$f_{k}(S)$ is of cardinal~$p^k$.
\end{proof}

  Applying \rlemma{tpq-suff-prec} to every integer~$k$ yields the
  following.

\begin{lemma} \llemma{tpq-suff}
  Every word in~$\Aps$ is the label of some path of~$\Tpq$.
\end{lemma}

% %
% \begin{lemma}\llemma{tpq-suff}
%   Let~$u$ be a word in~$\Aps$ and~$k$ the length of~$u$.
%   %
%   For every integer~$i$, there is an integer~$m$ in~$\{i,i+1\ldots,i+p^k\}$
%   and another integer~$n$ such that~$n\pathx{u}m$ in~$\Tpq$.
% \end{lemma}
% %

%%%%%%%%%%%%%%%%%%%%%%%%%%%%%%%%%%%%%%%%%%%%%%%%%%%%%%%%%%%%%%%%%%%%%%%%%%%%%%%%%%%%%%%%%%%%%%%%%%%
%%%%%%%%%%%%%%%%%%%%%%%%%%%%%%%%%%%%%%%%%%%%%%%%%%%%%%%%%%%%%%%%%%%%%%%%%%%%%%%%%%%%%%%%%%%%%%%%%%%
\subsection{Representation of real numbers}

Let us define a second \emph{evaluation function}~$\realval{}$.
It evaluates an \oword after the radix point (for short a.r.p.) hence
computes a real number.
%
% Given
The \realvalue of an \oword $\msp\ow{w}=a_1\xmd a_2 \cdots\msp$ over the
alphabet~$\Ap$,
or indeed \emph{over any digit alphabet},  is
\begin{equation}\lequation{pq-real-val}
%   \begin{array}{ccccl}
%     \realval{}: & \Apo &\longrightarrow &\R \\
\realval{a_1 \xmd a_2 \cdots} =
\sum_{i\geq1} \frac{a_i}{q} \left(\frac{p}{q}\right)^{\!\!-i} \eqpnt
%   \end{array}
\end{equation}
\begin{proposition}\lproposition{ro-cont}
  The function~$\realval{}$ is uniformly continuous.
\end{proposition}
Let us stress that the
function~$\realval{}$ is \emph{not} order-preserving.
% , as shown by an example below.
%
Since for every (non integer) rational base~$\pq$, $q\geq2$
and~$p\geq3$ hold, the following inequalities hold
% The inequality of the second equation holds from the fact that in
% , it holds.
\begin{equation}
  0\xmd (p-1)\xmd 0^\omega \lex<  1\xmd0^\omega
  \ee\text{and}\ee
  \realval[b]{0\xmd (p-1)\xmd 0^\omega} =
  \frac{q\xmd (p-1)}{p^2} > \frac{1}{p} = \realval{10^\omega}
  \eqpnt
  \notag
\end{equation}
% \begin{align*}
%   0\xmd (p-1)\xmd 0^\omega &\lex<  1\xmd0^\omega \\
%   \realval[b]{0\xmd (p-1)\xmd 0^\omega} = \frac{q\xmd (p-1)}{p^2}
%   &> \frac{1}{p} = \realval{10^\omega}
% \end{align*}
%
However, $\realval{}$ is order-preserving on the \olanguage accepted
by~$\Tpq$ (\rproposition{ro-orde-pres} below).
\begin{definition}\ldefinition{wpq}
  We denote by~$\mes\Wpq$ the \olanguage accepted by~$\Tpq$, that is,
  $\mes\Wpq = \ibehav{\Tpq}$~.
\end{definition}
For instance, Figures  \rfigure*{w32} and \rfigure*{w73}
(pages \pfigure*{w32} and \pfigure*{w73}) are representations of~$\Wtd$
and~$\Wst$ as
fractal trees.
In these figures, consider a path from the root to a node~$X$ labelled by a word~$u$.
The node~$X$ is then at the ordinate~$\realval{u\xmd 0^\omega}$ and is labelled
by~$\val{u}$.
The abscissa has no particular meaning except that it grows with the length of~$u$.
For example, in \rfigure{w32}, there is a path starting from the root and labelled
by~$u=21$; the endpoint of this path is a node labelled by~$\val{21}=2$
and positioned at the ordinate~$\realval{210^\omega}[\td]=0.888\cdots$.
Similarly, the run of~$u=210$ reaches a node labelled by~$\val{210}=3$
and whose ordinate is also~$\realval{2100^\omega}[\td]=\realval{210^\omega}[\td]=0.888\cdots$.
\begin{proposition}[{\citealp[Lemma~34]{AkiyEtAl08}}]\lproposition{ro-orde-pres}
$\msp\forall \ow{u},\ow{v} \in\Wpq\quantsp
\realval{\ow{u}}\leq\realval{\ow{v}} \e\Longleftrightarrow\e
\ow{u}\lex\leq \ow{v}\msp$.
%   The function~$\realval{}$ preserves the order
%   from~$(\Wpq,\lex\leq)$ to~$(\R,\leq)$.
\end{proposition}
As figures suggest, the set~$\Wpq$, when projected to~$\R$
by~$\realval{}$, produces an interval, as stated below.
\begin{theorem}[{\citealp[Theorem~2]{AkiyEtAl08}}]\ltheorem{ro-wpq-inte}
  The image of~$\mes\Wpq$ by~$\realval{}$ is an interval.
\end{theorem}
%

%%%%%%%%%%%%%%%%%%%%%%%%%%%%%%%%%%%%%%%%%%%%%%%%%%%%%%%%%%%%%%%%%%%%%%%%%%%%%%%%%%%%%%%%%%%%%%%%%%%
%%%%%%%%%%%%%%%%%%%%%%%%%%%%%%%%%%%%%%%%%%%%%%%%%%%%%%%%%%%%%%%%%%%%%%%%%%%%%%%%%%%%%%%%%%%%%%%%%%%
\subsection{Bottom and top words}

\rlemma{tpq-righ-exte} states that every state~$n$ of~$\Tpq$ is
the root of an infinite subtree.
We now turn our attention to the \owords that are the frontiers of
these subtrees.
Let us first call \emph{lower alphabet}, and denote by~$\Aq$, the set of
the smallest~$q$ integers: $\Aq=\{0,1,\ldots, q-1\}$.
%       \item \ldefinition{mini-alpha}

\begin{definition}\ldefinition{bot-top}
    \begin{enumerate}
       \item \ldefinition{mini-word}
      We call \emph{bottom word\/%in the lexicographic order
            \footnote{\emph{Bottom words} were called \emph{minimal words} in \cite{AkiyEtAl08}.}
      of~$n$}, and denote by~$\minword{n}$, the
      smallest \oword that labels a \opath of~$\Tpq$
      originating from~$n$.

      \item
      Let~$\minwords$ denote the set of the bottom words:
      $\minwords=\{ \minword{n}~|~ n\in\N \}$.
    \end{enumerate}
\end{definition}

\begin{example}
    One reads on \rfigure{l32} some bottom words in  base~$\td$:
  \begin{equation*}
    \minword{1} = 1 \xmd 0 \xmd 1 \xmd 1 \xmd 0 \xmd 0 \xmd 0 \cdots
    \quad\text{,}\text\quad\quad
    \minword{3} = 1 \xmd 1 \xmd 0 \xmd 0 \xmd 0 \cdots
    \quad\quad\text{and}\quad\quad
    \minword{4} = 0\xmd 0 \xmd 1 \xmd 0 \xmd 1 \cdots
  \end{equation*}
%
%   Since the path~$0\pathx{10}3$ ,~$1$,~$\minword{3}$ is a suffix of~$\minword{1}$.
\end{example}

Bottom words are characterised by the alphabet they are written on:

\begin{property}\lproperty{mini-mini}
 $\msp\minwords = \Wpq \cap \Aq^{\omega}\msp$.
%     Let~$\ow{w}$ be an \oword labelling a branch of~$\Tpq$.
%     It belongs to~$\minwords$ if and only if it belongs to~$\Aqo$.
\end{property}

This property will be used under the following form.

\begin{property}\lproperty{mini-nece}
Let~$n$ be in~$\N$ and~$u$ in~$\Aq^{*}$.
If
$\msp n\pathx{u}[\Tpq]m\msp$,
then~$u$ is a prefix of~$\minword{n}$.
%   Let~$u$ be a finite word over the lower alphabet~$\Aq$ and let~$n$ be an integer.
%   %
%   If~$u$ is the label of a path originating from~$n$, then~$u$ is the prefix of~$\minword{n}$.
\end{property}

From \rlemma{tpq-suff} and \rproperty{mini-mini} follows the next statement.

\begin{lemma}\llemma{minw-dens}
  The set~$\minwords$ is dense in~$\Aqo$.
\end{lemma}

Symmetrically, we denote by~$\maxword{n}$ the \emph{top word
\footnote{\emph{Top words} were called \emph{maximal words} in \cite{AkiyEtAl08}.} of~$n$},
by~$\maxwords$ the set of the top words and call \emph{upper alphabet}
the alphabet~$\Amax=\{p-q,p-q+1,\ldots,p-1\}$.
Statements much similar to \rproperty{mini-mini}, \rproperty{mini-nece} and \rlemma{minw-dens} could
be made about the top words and the upper alphabet.

\begin{example}
    One reads on \rfigure{l32} some top words in  base~$\td$:
%   Below are given the top words of~$\mes1$,~$3$ and~$4$ in base~$\td$,
%   they are the label of the highest branches of~$\Ttd$ (\rfigure{l32})
%   originating from these states.
  \begin{equation*}
    \maxword{1} = 1 \xmd 2 \xmd 2 \xmd 1 \xmd 1 \xmd 1 \xmd 2 \cdots
    \quad\text{,}\quad\quad
    \maxword{3} = 1 \xmd 1 \xmd 2 \xmd 1 \xmd 2 \cdots
    \quad\quad\text{and}\quad\quad
    \maxword{4} = 2 \xmd 1 \xmd 1 \xmd 1 \xmd 2 \cdots
  \end{equation*}
%     We may see that~$\minword{4}$ is a suffix of~$\minword{1}$.
\end{example}
%
% We take now interest in the close relation betweenfollowing
The bottom word of~$(n+1)$ and the top word of~$n$ are related by the
function
$\mu\colon\Amax\rightarrow\Aq$ defined by
%
% First, we define the function~$\mu$ as follows.
\begin{equation}\lequation{defi-mu}
    \mu(c)=c-(p-q)
    \eqvrg
\end{equation}
%
% The function~$\mu$ is
and extended to a (letter-to-letter) morphism
from~$\Amaxs$ to~$\Aqs$, and from~$\Amaxo$ to~$\Aqo$.
% as a letter-to-letter morphism.
% \begin{align}
%   \mu(a_{k}a_{k-1}\cdots  a_0) ={}& \mu(a_{k})\mu(a_{k-1})\cdots\mu(a_0)  \\
%   \mu(a_0a_{1}\cdots) ={}& \mu(a_0)\mu(a_1)\cdots
% \end{align}
%
%   \begin{array}{lcll}
%     \mu&  & & \\
%         & c &\longmapsto &
%   \end{array}

%
\begin{lemma}\llemma{mini-to-maxi}
  For every integer~$n$,~$\minword{n+1}=\mu(\maxword{n})$.
\end{lemma}
\begin{proposition} \lproposition{rela-mini-maxi}
  Let~$n,m$ be two integers and let~$a$ be a letter of~$\Ap$ such that
    $n \pathx{a}[\Tpq] m$ and~$n \pathx{a\,{+}\,q}[\Tpq] m+1$.
%   \end{equation*}
  Then,~$\realval[b]{(a\,{+}\,q)\xmd\minword{m+1}}=\realval{a\xmd\maxword{m}}$.
\end{proposition}
%

%%%%%%%%%%%%%%%%%%%%%%%%%%%%%%%%%%%%%%%%%%%%%%%%%%%%%%%%%%%%%%%%%%%%%%%%%%%%%%%%%%%%%%%%%%%%%%%%%%%
%%%%%%%%%%%%%%%%%%%%%%%%%%%%%%%%%%%%%%%%%%%%%%%%%%%%%%%%%%%%%%%%%%%%%%%%%%%%%%%%%%%%%%%%%%%%%%%%%%%
%%%%%%%%%%%%%%%%%%%%%%%%%%%%%%%%%%%%%%%%%%%%%%%%%%%%%%%%%%%%%%%%%%%%%%%%%%%%%%%%%%%%%%%%%%%%%%%%%%%
\section{Span-words}
\lsection{spa-wor}%

%
% In this section, we introduce
The notion of span-word will be central in the proof of both
Theorems~I and~II via the construction of a new automaton denoted
by~$\Tpqp$ and obtained from~$\Tpq$ by enlarging, or restricting, the
alphabet.

% and show that,
% up to an alphabet change, the \olanguage of the span-words is accepted
% by~$\Tpq$.
%
% This result is itself, but give
%

%%%%%%%%%%%%%%%%%%%%%%%%%%%%%%%%%%%%%%%%%%%%%%%%%%%%%%%%%%%%%%%%%%%%%%%%%%%%%%%%%%%%%%%%%%%%%%%%%%%
%%%%%%%%%%%%%%%%%%%%%%%%%%%%%%%%%%%%%%%%%%%%%%%%%%%%%%%%%%%%%%%%%%%%%%%%%%%%%%%%%%%%%%%%%%%%%%%%%%%
% \subsection{Span-words and span-automaton}

%
\begin{definition}\ldefinition{bpq}
    Let~$\Bpq$ denote the set of the differences between letters from the
    upper alphabet and letters from the lower alphabet:
    \begin{equation*}
      \Bpq= \Amax - \Aq = \Defi{d\in\Z}{\ext\xmd c\in\Amax,\,
      \ext\xmd b\in\Aq\quantsmsp d = c-b}
\eqpnt
    \end{equation*}
\end{definition}

The alphabet~$\Bpq$ is the integer interval whose cardinal is the odd
integer~$(2\xmd q -1)$, whose largest element is~$(p-1)$.
% and whose
Its
`central element', called \emph{middle-point}, is~$p-q$:
\begin{equation*}
  \Bpq= \set{p-(2\xmd q-1),\ldots,(p-1)} \eqpnt
\end{equation*}
%
% The next properties directly follow.
%

\begin{property}\lproperty{bpq}

  \begin{enumerate}
%     \item \sublabel{pp.bpq-diff}
%       Tletters
    \item \sublabel{pp.bpq-maxi}
      $\Amax\subseteq\Bpq$.

    \item \sublabel{pp.bpq-equ}
      If~$p=(2\xmd q -1)$, then~$\Bpq=\Ap$.

    \item \sublabel{pp.bpq-smal}
    If~$p<(2\xmd q -1)$, then~$\Bpq \supsetneq \Ap$ and contains \textbf{negative} digits.

    \item \sublabel{pp.bpq-big}
      If~$p>(2\xmd q -1)$, then~$\Bpq \subsetneq \Ap$; more precisely,~$\Bpq$
        is the set of the largest~$(2\xmd q -1)$ digits of~$\Ap$.
  \end{enumerate}
\end{property}

\begin{definition}
    \ldefinition{dw-min-max}
    We denote by~$\dwplus$ and~$\dwminus$ the
    digitwise addition and subtraction of words of the same length respectively, that is,
    \begin{gather*}
      (a_{k} \cdots a_1 \xmd a_0) \dwplus (b_{k} \cdots b_1 \xmd b_0)
        = (a_{k}+b_{k})\cdots (a_1+b_1)(a_0+b_0)\quad;\\
      (a_{k} \cdots a_1 \xmd a_0) \dwminus (b_{k} \cdots b_1 \xmd b_0)
        = (a_{k}-b_{k})\cdots (a_1-b_1)(a_0-b_0)\quad.
    \end{gather*}
    Digitwise addition and subtraction of \owords are defined similarly.
\end{definition}

\begin{property}\lproperty{sub}
For any~$w$ in~$\Bpqs$, there exist~$u$ in~$\Aqs$ and~$v$
in~$\Amaxs$ such that~$w=v\dwminus u$.
\end{property}

%   %
%   There
%   such that; each letter~$d$ of~$\Bpq$ can be written as~$(c-b)$, with~$b\in \Aq$
%   and~$c\in\Amax$ defined as follows.
%   %
%   If~$d\geq p-q$, then we set~$b=0$ and~$c=d$, otherwise it holds~$p-(2\xmd q-1) \leq
%   d < p-q$ and we set~$b=q-1$ and~$c=d+(q-1)$.
%   %
%   (Other letters~$b'$ in~$\Aq$ and~$c'$ in~$\Amax$ satisfy~$d=c'-b'$.)
%   %
%

\begin{definition}
  \begin{enumerate}
  \item\sublabel{d.span-word}
    We call \emph{span-word\/%
    \footnote{%
              The denomination \emph{span-word} comes from the
              \realvalue of those \owords,
              and will be explained in \rsection{on-span} (\rdefinition{spans}).}
    of~$n$}, and denote by~$\spanword{n}$,
    the \oword~$\maxword{n}\dwminus\minword{n}$.
  \item\sublabel{d.span-words}
    We denote by~$\spanwords$ the set of all span-words:~$\spanwords=\set{\spanword{n}~|~n\in\N}$~.
  \end{enumerate}
\end{definition}
\begin{example}
  In base~$\td$, it reads:
  \begin{align*}
    \spanword{1} &= \maxword{1} \dwminus \minword{1}
                 =         (1 \xmd 2 \xmd 2 \xmd 1 \xmd 1 \xmd 1 \xmd 2 \cdots)
                  \dwminus (1 \xmd 0 \xmd 1 \xmd 1 \xmd 0 \xmd 0 \xmd 0 \cdots)
                 =          0 \xmd 2 \xmd 1 \xmd 0 \xmd 1 \xmd 1 \xmd 2 \cdots
    \\
    \spanword{3} &= \maxword{3} \dwminus \minword{3}
                 =          (1 \xmd 1 \xmd 2 \xmd 1 \xmd 2 \cdots)
                   \dwminus (1 \xmd 1 \xmd 0 \xmd 0 \xmd 0 \cdots)
                 =           0 \xmd 0 \xmd 2 \xmd 1 \xmd 2 \cdots
    \\
    \spanword{4} &= \maxword{4} \dwminus \minword{4}
                 =          (2 \xmd 1 \xmd 1 \xmd 1 \xmd 2 \cdots)
                   \dwminus (0 \xmd 0 \xmd 1 \xmd 0 \xmd 1 \cdots)
                 =          2 \xmd 1 \xmd 0 \xmd 1 \xmd 1 \cdots
  \end{align*}
\end{example}
Since bottom words belong to~$\Aqo$ and top words to~$\Amaxo$,
it follows:
% the next statement
% follows from the \rdefinition{bpq} of~$\Bpq$.
%

%
\begin{property}
  $\spanwords\subseteq\Bpqo$.
\end{property}
%

%

%
% %
% Moreover, the cardinal of~$\Bpq$ is~$(2\xmd q -1)$, an odd number.
% %
% The alphabet~$\Bpq$ hence has a middle-point,~$(p-q)$.
% %
% This middle-point will have a role later on in \rsection{rempl-later}.
% %

% %
% The denomination \emph{span-word} is a reference to the \realvalue of
% those \owords (\cf \rsection{on-span}).
% %
% Moreover, the following property follows directly from Definitions
% \rdefinition*{bpq}, \rdefinition*{dwmin} and \rdefinition*{span-word}.
%
% %
% \begin{property}\lproperty{span-word-eval}
%   For every integer~$n$, the \oword $\spanword{n}$ belongs to~$\Bpqo$.
% \end{property}
% %

%
% We now define an automaton which is shown to accept the span-words
% afterwards (\rtheorem{span-auto}).
%

%
\begin{definition}\ldefinition{tpqp}
    Let~$\Tpqp$ be the automaton defined by
  \begin{equation*}
    \Tpqp  = \aut{\,\Bpq,\,\N,\, 0,\, \taupq\,}\quantvrg
  \end{equation*}
  where~$\taupq$ is defined by \requation{tau} with domain
  restricted to~$\N\times\Bpq$.
\end{definition}

The transitions of~$\Tpqp$ are characterised by:
%
% (Note that it differs from \requation*{tpq-path} only by the alphabet in which~$a$ is taken.)
% %
  \begin{equation}\lequation{tpqp-path}
  \forall n,m \in \N \quantvrg \forall a\in\Bpq \quantsp n\pathx{a}[\Tpqp]m \iff q\xmd m  = p\xmd n + a\eqpnt
  \end{equation}
Using \requation*{tpqp-path}, it is a routine to show that \rlemma{tpq-val}
extends to~$\Tpqp$. %when~$\Bpq\subsetneq \Ap$.
% (
% the words in~$\behav{\Tpqp}$ are  \rlemma{tpq-val}
% is also true for~$\Tpqp$.
%

%
\begin{lemma}\llemma{tpqp-val}
      Let~$u$ be in~$\behav{\Tpqp}$.
  Then, $\val{u}$ is in~$\N$ and
  $\msp0\pathx{u}[\Tpqp]\val{u}\msp$.
%   Let~$u$ be a word of~$\behav{\Tpqp}$.
%   %
%   Its evaluation, $\val{u}$, is an integer and the run of~$u$ in~$\Tpqp$
%   reaches the state~$\val{u}$.
\end{lemma}
\begin{example}
  \begin{enumerate}[itemsep=0.5\thmspace, topsep=0.5\thmspace]
    \item
    The base~$\td$ satisfies~$p=(2\xmd q - 1)$, hence~$\Btd=A_3$.
    In this case,~$\Ttdp$ is simply equal to~$\Ttd$.
    \item
    The base~$\frac{4}{3}$  satisfies~${p<(2\xmd q - 1)}$, hence~$\Bqt$ contains~$A_4$
      plus some negative digits (here only one:~$-1$).
    Transitions are added to~$\Tqt$ in order to build~$\Tqtp$.
    These transitions are drawn with a thick line in \rfigure{t43p} (\pfigure{t43p}).

    \item
    The base~$\frac{7}{3}$ satisfies~${p>(2\xmd q - 1)}$, hence~$\Bst$ is a strict subset of~$A_4$.
    The transitions labelled by the smallest two letters of~$A_4$
    are deleted from~$\Tst$ in order to produce~$\Tstp$.
    These transitions are dashed in \rfigure{t73p}.
  \end{enumerate}
\end{example}
    \begin{figure}[h]
      \centering
      \includegraphics[scale=\TreeScale]{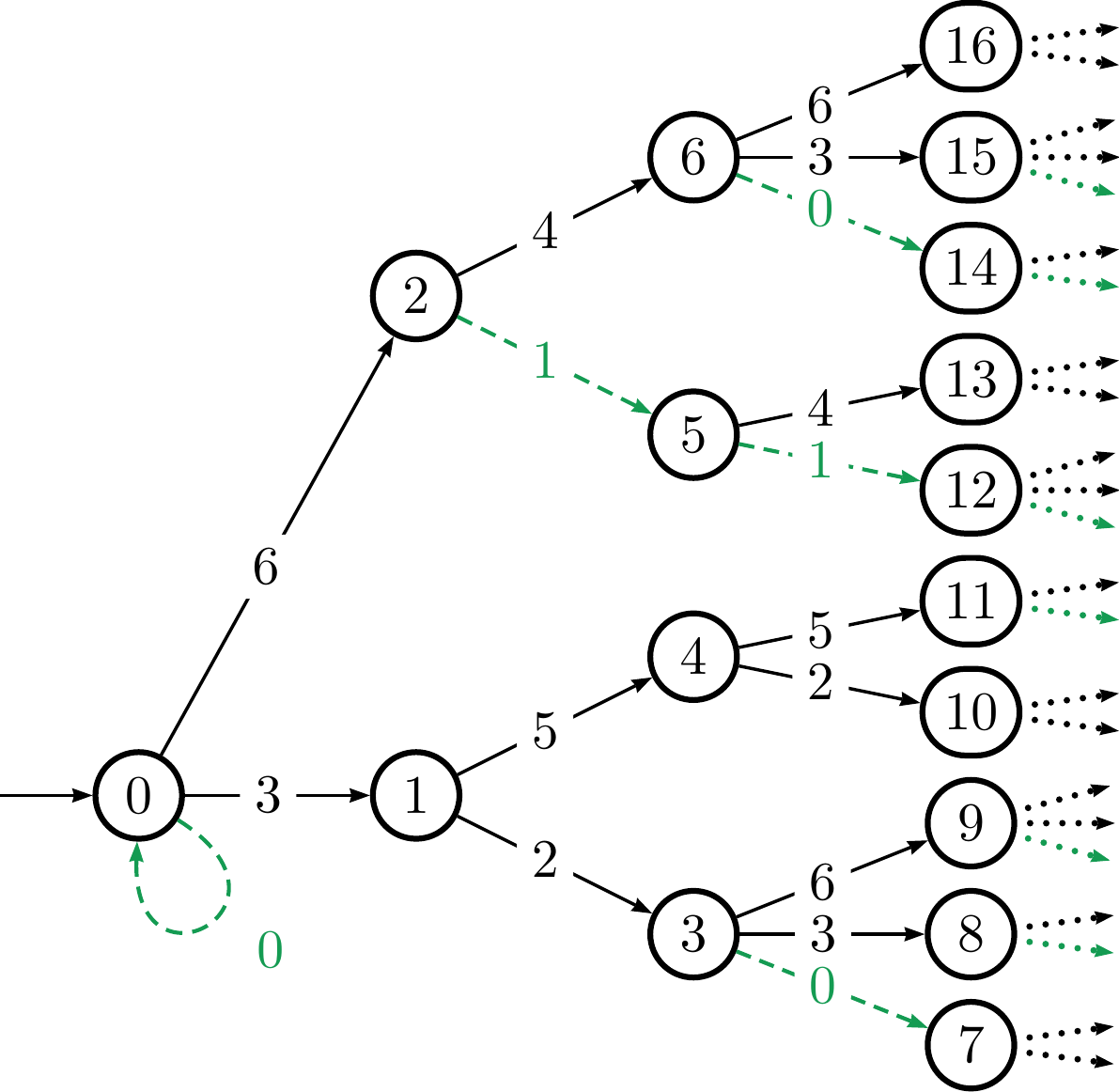}
      \caption{Construction of~$\Tstp$ from~$\Tst$, that is, deletion
      of the transitions labelled by the letters~$0$ and~$1$}
      \lfigure{t73p}
    \end{figure}%
The main result of the section states that~$\Tpqp$ accepts the
span-words, and more precisely reads as follows.
%
% The remainder of Section \thesection{} is dedicated to prove
% that~$\Tpqp$ accepts the
% span-words, as stated more precisely by next Theorem.
% %

%
\begin{theorem}\ltheorem{span-auto}\eqspacing$ \ibehav{\Tpqp} = \adh{\spanwords}$
\end{theorem}
The proof essentially boils down to the linearity of~$\taupq$
(the transition function of~$\Tpq$ and~$\Tpqp$) as
expressed by the next lemma, which follows immediately
from \requation*{tpq-path} and \requation*{tpqp-path}.
\begin{lemma}\llemma{tau-line}
  Let~$n,m$ in~$\N$ and~$x,y$ in~$\Z$ and suppose that~$\taupq(n,x)$ is defined.
  Then,~$\taupq(m,y)$ is defined if and only if~$\taupq(n+m,\,x+y)$
  is defined.

%   Any two of the following imply the third.
%   \begin{enumerate}
%     \item~$\taupq(n,x)$ is defined.
%     \item.
%     \item.
%   \end{enumerate}
%   %
  In this case moreover, $\taupq(n+m,\,x+y)=\taupq(n,x)+\taupq(m,y)$.
\end{lemma}
%

%
% Let us recall that a letter of~$\Amax$ minus a letter of~$\Aq$
% necessarily yields a letter of~$\Bpq$.
%

\begin{proposition}\lproposition{corr-word}
  Let~$u$ be in~$\Aqs$ and~$n$ and~$m$ in~$\N$ such that
  $n\pathx{u} m$ in~$\Tpq$.
  Let~$v$ be in~$\Amaxs$ of the same length as~$u$ and~$i$ and~$j$
  in~$\N$.
  Then:
  \begin{equation}
(n+i)\pathx{v}[\Tpq] (m+j)
\e\Longleftrightarrow\e
i\pathx{v\dwminus u}[\Tpqp] j
\eqpnt
\label{eq.corr-word}
% \notag
  \end{equation}
\end{proposition}
% \begin{proposition}\lproposition{corr-word}
%   Let~$u$ be in~$\Aqs$ and~$v$ in~$\Amaxs$ of the same length.
% %   Let~$u$ and~$v$ be two words of the same length
% %   respectively over the alphabets~$\Aq$ and~$\Amax$.
%   %
%   Let~$n,i,j,m$ be four integers.
%   %
%   Any two of the following conditions imply the third:
%   %
%   \begin{enumerate}
%     \item $n\pathx{u} m$ in~$\Tpq$
%     \item $i\pathx{v\dwminus u} j$ in~$\Tpqp$
%     \item $(n+i)\pathx{v} (m+j)$ in~$\Tpq$
%   \end{enumerate}
% \end{proposition}
\begin{proof}
First, the statement holds if~$\wlen{u}=\wlen{v}=1$:
~$u$ is then reduced to one letter~$b$ of~$\Aq$,
$v$ to one letter~$c$
of~$\Amax$,
and~$v\ominus u$ to the letter~$(c-b)$ which belongs to~$\Bpq$.
By hypothesis,  $\taupq(n,b)$ is defined and equal to~$m$, and
\rlemma{tau-line} yields exactly \requation{corr-word}.
%   Items \enumstyle{a}, \enumstyle{b} and, \enumstyle{c} are then respectively equivalent
%   to the following.
%   \begin{enumerate}[label=\enumstyle{\alph{*}'}]
%     \item $\taupq(n,b)$ is defined and equal to~$m$
%       \hfill (equiv.\@ to \enumstyle{a} since $b\in\Aq\subseteq\Ap$)
%     \item $\taupq(i,\,c-b)$ is defined and equal to~$j$
%       \hfill (equiv.\@ to \enumstyle{b} since $(c-b)\in\Bpq$)
%     \item $\taupq(n+i,\,c)$ is defined and equal to~$m+j$
%       \hfill (equiv.\@ to \enumstyle{c} since $c\in\Amax\subseteq\Ap$)
%   \end{enumerate}
%   %
%   Applying \rlemma{tau-line} concludes the case~$\wlen{u}=\wlen{v}=1$.
%   %

  \bigskip

The case~$u=v=\epsilon$ is trivial.
Let us suppose that~$u=b\xmd u'$, $v = c\xmd v'$ and that
\begin{equation}
n\pathx{b}[\Tpq] n' \pathx{u'}[\Tpq] m
\eqpnt
% \label{eq.corr-word}
\notag
\end{equation}
If
$\msp\displaystyle{i\pathx{c-b}[\Tpqp]i'\pathx{v'\dwminus u'}[\Tpqp]j}\msp$
then
$\msp\displaystyle{n+i\pathx{c}[\Tpq]n'+i'}\msp$ and
$\msp\displaystyle{n'+i'\pathx{v}[\Tpq]m+j}\msp$, and  hence\linebreak
$\msp\displaystyle{n+i\pathx{c\xmd v'}[\Tpq]m+j}\msp$.
And Conversely, if
~$\msp\displaystyle{n+i\pathx{c}[\Tpq]n'+i'\pathx{v'}[\Tpq]m+j}\msp$~
then
$\msp\displaystyle{i\pathx{c-b}[\Tpqp]i'}\msp$ and
$\msp\displaystyle{i'\pathx{v\dwminus u}[\Tpqp]j}\msp$, and hence
$\msp\displaystyle{i\pathx{c\xmd v'\dwminus b\xmd u}[\Tpqp]j}\msp$.
\end{proof}
\begin{theorem}\ltheorem{span-auto++}
Let~$i$ be an integer and~$w$ a word in~$\Bpqs$.
The following are equivalent.
\begin{enumerate}
\item There exists an integer~$j$ such that~$\mes i\pathx{w} j\mes$
is a path of~$\Tpqp$.
\item There exists an integer~$n$ such that $w$ is a prefix
of~$\mes\maxword{n+i}\dwminus\minword{n}\mes$.
\end{enumerate}
\end{theorem}
\begin{proof} $\enumstyle{a}\!\Rightarrow\!\enumstyle{b}$.
%   Let~$w$ be a word and~$j$ an integer such that~$i\pathx{w}j$ is a path of~$\Tpqp$.
  %
Let~$u$ in~$\Aqs$ and~$v$ in~$\Amax^*$ such that~$w=v\dwminus u$
(\rproperty{sub}).  %

Since every word in~$\Aps$ labels a path of~$\Tpq$ (\rlemma{tpq-suff}),
  there exist~$n$ and~$m$ in~$\N$ such that~$n\pathx{u} m$.
By hypothesis, the path $i\pathx{w} j$ is in~$\Tpqp$, and by the choice of~$u$
and~$v$,
%   it holds~$w=v\dwminus u$, applying
\rproposition{corr-word} yields that ${(n+i)\pathx{v} (m+j)}$.
  %
%   Since~$u$ is over the lower alphabet, it is a prefix of~$\minword{n}$ (\rproperty{mini-mini}).
%   %
%   Similarly,~$v$ is necessarily a prefix of~$\maxword{n+i}$.
Since~$u$ is in~$\Aqs$, it is a prefix of~$\minword{n}$
(\rproperty{mini-nece}).
  Similarly,~$v$ is a prefix of~$\maxword{n+i}$.
  Hence,~$w=v\dwminus u$ is a prefix of~${\maxword{n+i}\dwminus\minword{n}}$.

  \bigskip

$\enumstyle{b}\! \Rightarrow \!\enumstyle{a}$.
  Let~$w$ be a prefix of~$\maxword{n+i}\dwminus\minword{n}$.
  %
%   We write~$u,v$ the respective prefixes (of length~$\wlen{w}$)
%   of~$\minword{n}$ and~$\maxword{n+i}$, hence it holds~$w=v\dwminus u$
%   (and~$v=u\dwplus w$).
  %
%   Moreover, we write~$v=w\dwminus u$, a word of~$\Bpqs$ which is the
%   prefix of length~$w$ of the span-word of~$n$.
We write~$u$ and~$v$ for the prefixes of length~$\wlen{w}$
  of~$\minword{n}$ and~$\maxword{n+i}$ respectively.
  Hence it holds~$w=v\dwminus u$
  (and~$v=u\dwplus w$).
  We denote by~$m$ and~$m'$ the endpoints of the paths~$n\pathx{u} m$
  and~$(n+i)\pathx{v} m'$ of~$\Tpq$.
  %
%   Since~$(n+i) \geq n$, it holds~$m' \geq m$ and we denote by~$j$
%   the integer~$m'-m$.
%   %
%   Applying \rproposition{corr-word} yields the existence of the path~$i\pathx{w} j$ in~$\Tpqp$.
  Since~$(n+i) \geq n$, it holds~$m' \geq m$ and we write~$j=m'-m$.
  \rproposition{corr-word} yields the existence of the
  path~$\msp i\pathx{w} j\msp$ in~$\Tpqp$.
\end{proof}

\begin{corollary}\lcorollary{span-auto++}
For every~$n$ and~$i$ in~$\N$, the \oword
$\msp\ow{u}=\maxword{n+i}\dwminus\minword{n}\msp$ is the label of a
branch of~$\Tpqp$ originating in state~$i$.
\end{corollary}

% applying
\rtheorem{span-auto} is the direct consequence of \rtheorem{span-auto++}
with~$i=0$, together with \rlemma{beha-ibeh}.
%

%%%%%%%%%%%%%%%%%%%%%%%%%%%%%%%%%%%%%%%%%%%%%%%%%%%%%%%%%%%%%%%%%%%%%%%%%%%%%%%%%%%%%%%%%%%%%%%%%%%
%%%%%%%%%%%%%%%%%%%%%%%%%%%%%%%%%%%%%%%%%%%%%%%%%%%%%%%%%%%%%%%%%%%%%%%%%%%%%%%%%%%%%%%%%%%%%%%%%%%
%%%%%%%%%%%%%%%%%%%%%%%%%%%%%%%%%%%%%%%%%%%%%%%%%%%%%%%%%%%%%%%%%%%%%%%%%%%%%%%%%%%%%%%%%%%%%%%%%%%
\section{On the successor function for bottom words}
\lsection{suc-fun}

%
% In Section~\thesection{}, we study
We now consider
the function~$\Der$ that maps the bottom word
of~$n$ to the bottom word of~$n+1$.
This function is related to span-words by the following.
\begin{itemize}
  \item The span-word of~$n$ is the digitwise difference of the top word of~$n$
        and bottom word of~$n$.
        In some sense, it is a way to transform the later into the former.
  \item The letter-to-letter morphism~$\mu$ (previously defined
    in~\requation*{defi-mu}) maps, for all~$n$,
  the top word of~$n$ to the bottom word of~$n+1$.
\end{itemize}
Using these facts, we define in \rsection{dpq} a label-replacement
function~$\rempfun$,
which we apply to~$\Tpqp$ and obtain a transducer~$\Dpq$.
Finally we show \rtheorem{dpq=xi}, restated below.
\begin{falsetheorem}{\rtheorem*{dpq=xi}}
  Let~$p,q$ be two coprime integers such that~$p>q>1$.
  The infinite transducer~$\Dpq$ realises the continuous extension
  of~$\Der$.
\end{falsetheorem}
%

%%%%%%%%%%%%%%%%%%%%%%%%%%%%%%%%%%%%%%%%%%%%%%%%%%%%%%%%%%%%%%%%%%%%%%%%%%%%%%%%%%%%%%%%%%%%%%%%%%%
%%%%%%%%%%%%%%%%%%%%%%%%%%%%%%%%%%%%%%%%%%%%%%%%%%%%%%%%%%%%%%%%%%%%%%%%%%%%%%%%%%%%%%%%%%%%%%%%%%%
\subsection{The function~$\Der$}
\lsection{deri}

\begin{definition}\ldefinition{deri}
%   Let~$\Der$ be the function~$\minwords\rightarrow\minwords$ that
%   maps for every~$n$,~$\minword{n}$ to~$\minword{n+1}$.
  Let~$\Der\colon\minwords\rightarrow\minwords$ be the function that
  maps~$\minword{n}$ onto~$\minword{n+1}$ for every~$n$.
%   ~$\Aqo\rightarrow \Aqo$ whose domain is~$\minwords$
%   and defined by~$\Der(\minword{n})=\minword{n+1}$, for every integer~$n$.
\end{definition}

%
% Using \rlemma{tpq-futu}, one can show that~$\Der$ is ``letter-to-letter'' as
% stated below.
The function~$\Der$ is ``letter-to-letter'', or ``on-line'' and
``real-time'', as stated by the following.

\begin{lemma}\llemma{xi-ltl}
  Let~$n$ and~$m$ be two integers.
  For every integer~$i$, the prefixes of length~$i$ of~$\mes\minword{n}$ and of~$\mes\minword{m}$
    are equal if and only if the prefixes of length~$i$ of~$\Der(\minword{n})$ and of~$\Der(\minword{m})$ are.
\end{lemma}
\begin{proof}
%   We write~$u,v\in\Aqs$ for the prefixes of length~$i$ of~$\minword{n},\minword{m}$
%   and~$u',v'\in\Aqs$ for those of~$\minword{(n+1)},\minword{(m+1)}$.
Let~$u$ and~$v$ be the prefixes of length~$i$ of~$\minword{n}$
and~$\minword{m}$ respectively,
and~$u'$ and~$v'$ those of~$\minword{(n+1)}$ and~$\minword{(m+1)}$.
These four words belong to~$\Aqs$.
  If~$u=v$, then~$(n\cdot u)$ and~$(m\cdot u)$ both exist (in~$\Tpq$).
  It follows from~\rlemma{tpq-futu<-} that~${n\equiv m~[q^i]}$, hence
  also~${(n+1)\equiv (m+1)~[q^i]}$.
  Moreover, by definition of~$u'$,~${((n+1)\cdot u')}$ exists.
  Applying \rlemma{tpq-futu->} then yields that~$((m+1)\cdot u')$ exists
  as well.
  Since~$u'$ is over the lower alphabet~$\Aq$, it is a prefix
  of~$\minword{n+1}$ (\rproperty{mini-nece}) hence~$u'=v'$
  Showing that~$u'=v'$ implies~$u=v$ is similar.
\end{proof}

Recall that $\minwords$ is dense in~$\Aqo$ (\rlemma{minw-dens}).
Then, it follows from \rlemma{xi-ltl} that~$\Der$ may be extended by continuity
to a bijection~$\Aqo\rightarrow\Aqo$.
We still denote this function by~$\Der$.
\rlemma{xi-ltl} states that the knowledge of the first~$i$
letters of an \oword~$\ow{w}$ is enough to compute the first~$i$
letters of~$\Der(\ow{w})$.
In other words,~$\Der$ is realised by an (infinite,
letter-to-letter and sequential) transducer.
\subsection{Definition of the transducer~$\Dpq$}
\lsection{dpq}

Recall that~$\mu\colon\Amax\rightarrow\Aq$ is the function defined
by~$\mu(c)=c-(p-q)$, for every~$c$ in~$\Amax$.
%in \requation*{defi-mu}

%
\begin{definition}\ldefinition{remp-fun}
We denote by~$\rempfun$ the function from~$\Bpq$
into~$\msp\powerset(\Aq\times\Aq)$ defined by:
\begin{equation*}
  \rempfun(d) =\Big\{~ \big(b, \mu(c) \big)~~\Big|~~ b\in\Aq~,~~c\in\Amax~,~~c-b=d ~\Big\}
  \eqpnt
\end{equation*}
%     \begin{array}{cccl}
%       \rempfun: & \Bpq & \rightarrow & \powerset(\Aq\times\Aq) \\[1mm]
%                 & d    & \mapsto     & \left\{~(b, \mu(c)) ~~\middle|~~ \begin{array}{@{}l@{}} b\in\Aq \\
%                                                                       c\in\Amax \\
%                                                                       c-b=d
%                                                                       \end{array}
%                                                                       ~\right\}
%     \end{array}
%   \end{equation*}
\end{definition}

The function~$\rempfun$ may be given a more self-contained definition:
the function~$\mu$ extended to~$\Bpq$ computes the (signed)
distance~$\mu(d)=d-(p-q)$ of~$d$ to the middle-point
of~$\Bpq$ and the set~$\rempfun(d)$ is the set of all
pairs~$(b,b')$ in~$\Aq{\times}\Aq$
whose difference,~$b'-b$, is equal to this distance.
%every~$d$ in~$\Bpq$ is at
\begin{property}
    $\msp\fa\xmd d\in\Bpq\quantsp
    \rempfun(d) ~=~ \big\{~(b, b')~\big|~b'-b=d-(p-q)~\big\}\msp$.
\end{property}
% The function~$\mu$ simply subtracts~$(p-q)$,
% hence~$\rempfun$ may be given a more self-contained definition.
% %
% Each letter~$d$ in~$\Bpq$ is at the (signed) distance~$d-(p-q)$ of the middle-point
% of~$\Bpq$.
% % %
% The set~$\rempfun(d)$ is the set of every pair~$(b,b')\in\Aq{\times}\Aq$
% whose difference,~$b'-b$, is equal to this distance.
The next property follows immediately.

\begin{property}\lproperty{psi-part}
  For every pair of distinct~$d$ and~$d'$ in~$\Bpq$,~$\mes\rempfun(d)\cap\rempfun(d') =\emptyset$.
\end{property}

\begin{definition}\ldefinition{dpq}
Let~$\Dpq$ be the transducer
\begin{equation*}
  \Dpq  = \aut{\, \Aq \times \Aq,\,\N,\, 0,\, \delta \,}\eqvrg
\end{equation*}
defined by
$\msp\delta(n,(b,b'))=\taupq(n,((b'-b)+(p-q)))\msp$
for every~$n$ in~$\N$ and letters~$b,b'$ of~$\Aq$.
% ,~$$
% if there exists~$d\in\Bpq$ such that~$(b,b')\in\psi(d)$;
% otherwise,~$\delta(n,(b,b'))$ is undefined.%
% \footnote{From \rproperty{psi-part}, there is at most one~$d$ such that~$(b,b')\in\phi(d)$, hence~$\delta$ is well defined.}
%
In other words,
\begin{equation*}
    \forall n,m \in \N \quantvrg \forall b,b'\in\Aq \quantsp
    n\pathx{(b,b')}[\Dpq]m
    \e\Longleftrightarrow\e
    n\pathx{d}[\Tpqp]m
    \e\text{and}\e
    (b,b')\in\rempfun(d)
    \eqvrg
\end{equation*}
that is, $\Dpq$ is obtained from~$\Tpqp$ by substituting every
label~$d$ by~$\rempfun(d)$.
\end{definition}
%
% \begin{definition}\ldefinition{dpq}
%   We denote by~$\Dpq$ the transducer
%   \begin{equation*}
%     \Dpq  = \aut{\, \Aq \times \Aq,\,\N,\, 0,\, \delta \,}\eqvrg
%   \end{equation*}
%   where, for every integer~$n$ and letters~$b,b'$ of~$\Aq$,~$\delta(n,(b,b'))=\tau(n,d)$
%   if there exists~$d\in\Bpq$ such that~$(b,b')\in\psi(d)$;
%   otherwise,~$\delta(n,(b,b'))$ is undefined.%
%   \footnote{From \rproperty{psi-part}, there is at most one~$d$ such that~$(b,b')\in\phi(d)$, hence~$\delta$ is well defined.}
%   %
%   In other words,~$\Dpq$ features the transition $n\pathx{(b,b')}m$
%   if and only if there exists a letter~$d\in\Bpq$ such that~$(b,b')\in\rempfun(d)$
%   and~the transition $n\pathx{d}m$ exists in~$\Tpqp$.
% \end{definition}
% %

The transitions of~$\Dpq$ are then also characterised by:
\begin{equation}\lequation{dpq-path}
  \forall n,m \in \N \quantvrg \forall b,b'\in\Aq \quantsp
    n\pathx{(b,b')}[\Dpq]m ~\iff~ q\xmd m  = p\xmd n +(b'-b)+(p-q)
\end{equation}
%
% The transducer~$\Dpq$ may be build by applying a label
% substitution~$\rempfun$, defined, on the transition of~$\Tpqp$.
% %

%
% \begin{proposition}\lproposition{dpq-cara}
%   Let~$n,m$ be two integers and~$a,b$ be two letetrs of~$\Aq$.
%   %
%   The transduce
% \end{proposition}
% \begin{proof}Computation...
%
% \end{proof}
%

\begin{figure}[t]
  \centering
  \includegraphics[scale=\TreeScale]{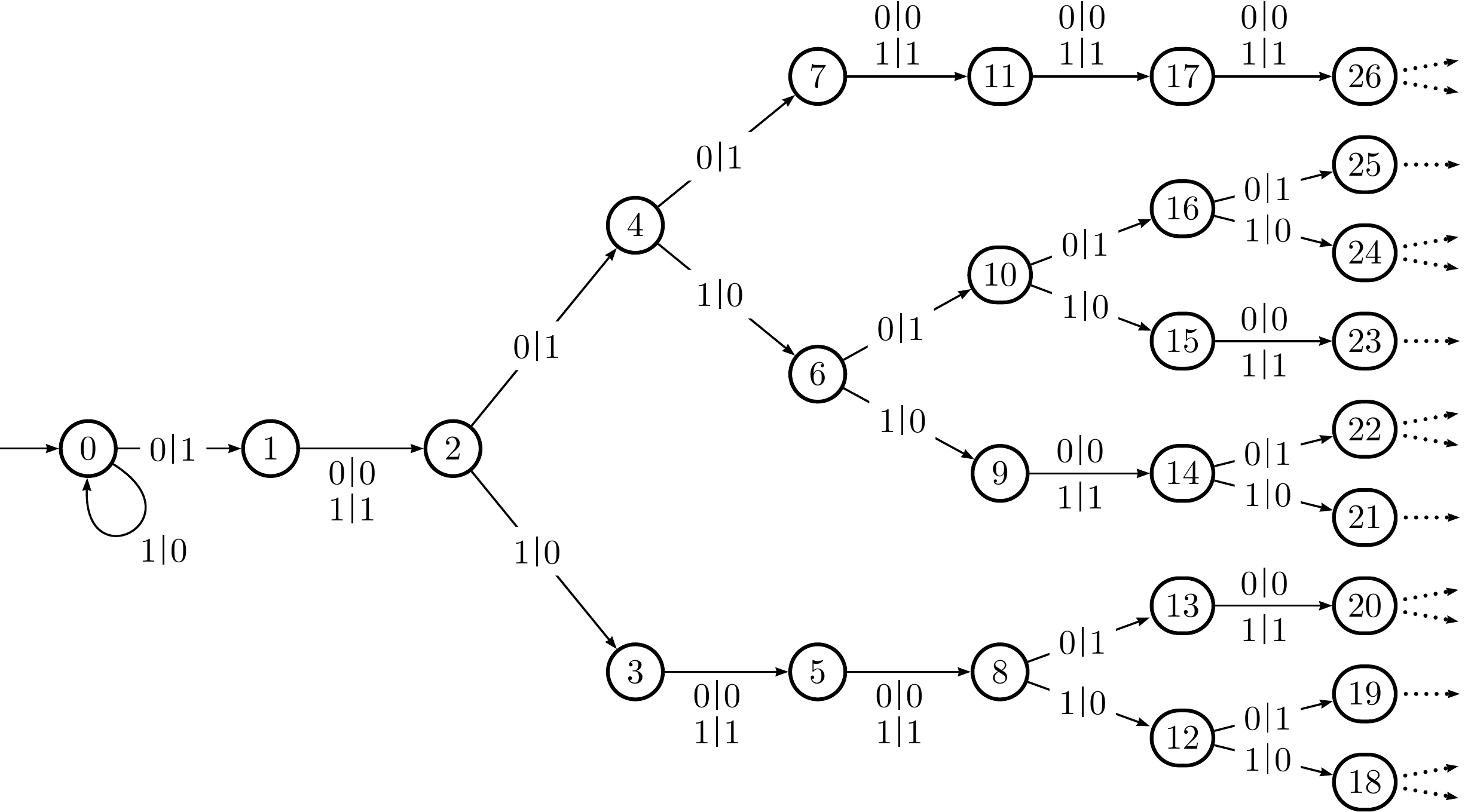}
  \caption{$\Dtd$}
  \lfigure{d32}
\end{figure}

\begin{example}%
%   In these examples, we denote by~$\rempfunp$ the function mapping each letter~$d$
%   of~$\Bpq$ to its distance to the middle-point:~$\rempfunp(d)=d-(p-q)$.
  %

  %
\begin{enumerate}[topsep={0.75\thmspace},itemsep={0.75\thmspace}]
\item
In base~$\td$, the middle-point of~$\Bpq$ is~$(p-q)=1$ and it reads:
% ~$\rempfun, \rempfunp$ are as follows.
\begin{align*}%{2}
  \mu(0) &= -1\eee & \rempfun(0) & = \set{~\transpair{1}{0}~} \\
  \mu(1) &= 0 & \rempfun(1) & = \set{~\transpair{1}{1},~\transpair{0}{0}~} \\
  \mu(2) &= 1 & \rempfun(2) & = \set{~\transpair{0}{1}~}
\end{align*}
%   \begin{equation*}
%     \begin{array}{rrcrcrlcl}
%       \rempfunp\colon & 0 & \longmapsto & -1 & \quad\quad\quad\quad & \rempfun\colon & 0 & \longmapsto & \set{~\transpair{1}{0}~} \\
%       \rempfunp\colon & 1 & \longmapsto &  0 &  & \rempfun\colon & 1 & \longmapsto & \set{~\transpair{1}{1},~\transpair{0}{0}~} \\
%       \rempfunp\colon & 2 & \longmapsto &  1 &  & \rempfun\colon & 2 & \longmapsto & \set{~\transpair{0}{1}~}
%    \end{array}
% %   \notag
%   \end{equation*}
The transducer~$\Dtd$ is shown in~\rfigure{d32}.
Since~$p=2\xmd q-1$, it has the same underlying graph as~$\Tpq$.
\item
In base~$\qt$, the middle-point is~$1$ as well and it reads:
\begin{align*}%{2}
  \mu(-1) &= -2\eee & \rempfun(-1) & = \set{~\transpair{2}{0}~} \\
  \mu(0) &= -1  & \rempfun(0) & = \set{~\transpair{2}{1},~\transpair{1}{0}~} \\
  \mu(1) &= 0 & \rempfun(1) & = \set{~\transpair{2}{2},~\transpair{1}{1},~\transpair{0}{0}~} \\
  \mu(2) &= 1 & \rempfun(2) & = \set{~\transpair{1}{2},~\transpair{0}{1}~} \\
  \mu(3) &= 2 & \rempfun(3) & = \set{~\transpair{0}{2}~}
\end{align*}
% \begin{equation*}
%   \begin{array}{rrcrcrrrlr}
%     \rempfunp\colon & -1 & \longmapsto & -2 & \quad\quad\quad\quad & \rempfun\colon & -1 & \longmapsto & \set{~\transpair{2}{0}~} \\
%     \rempfunp\colon & 0 & \longmapsto &  -1 &   & \rempfun\colon & 0 & \longmapsto & \set{~\transpair{2}{1},~\transpair{1}{0}~} \\
%     \rempfunp\colon & 1 & \longmapsto &  0 &   & \rempfun\colon & 1 & \longmapsto & \set{~\transpair{2}{2},~\transpair{1}{1},~\transpair{0}{0}~} \\
%     \rempfunp\colon & 2 & \longmapsto &  1 &   & \rempfun\colon & 2 & \longmapsto & \set{~\transpair{1}{2},~\transpair{0}{1}~} \\
%     \rempfunp\colon & 3 & \longmapsto &  2 &   & \rempfun\colon & 3 & \longmapsto & \set{~\transpair{0}{2}~} \\
%   \end{array}
% \end{equation*}
Figures \rfigure*{l43}, \rfigure*{t43p} and \rfigure*{d43}  sum up the construction of~$\Dqt$.

\item
In base~$\st$, $\Bst=\set{2,3,4,5,6}$, its middle-point is~$4$ and it reads:
\begin{align*}%{2}
  \mu(2) &= -2\eee & \rempfun(2) & = \set{~\transpair{2}{0}~} \\
  \mu(3) &= -1& \rempfun(3) & = \set{~\transpair{2}{1},~\transpair{1}{0}~} \\
  \mu(4) &= 0 & \rempfun(4) & = \set{~\transpair{2}{2},~\transpair{1}{1},~\transpair{0}{0}~} \\
  \mu(5) &= 1 & \rempfun(5) & = \set{~\transpair{1}{2},~\transpair{0}{1}~} \\
  \mu(6) &= 2 & \rempfun(6) & = \set{~\transpair{0}{2}~}
\end{align*}
%   \begin{equation*}
%     \begin{array}{rrcrcrrrlr}
%       \rempfunp\colon & 2 & \longmapsto & -2 & \quad\quad\quad\quad & \rempfun\colon & 2 & \longmapsto & \set{~\transpair{2}{0}~} \\
%       \rempfunp\colon & 3 & \longmapsto &  -1 &  & \rempfun\colon & 3 & \longmapsto & \set{~\transpair{2}{1},~\transpair{1}{0}~} \\
%       \rempfunp\colon & 4 & \longmapsto &  0 &  & \rempfun\colon & 4 & \longmapsto & \set{~\transpair{2}{2},~\transpair{1}{1},~\transpair{0}{0}~} \\
%       \rempfunp\colon & 5 & \longmapsto &  1 &  & \rempfun\colon & 5 & \longmapsto & \set{~\transpair{1}{2},~\transpair{0}{1}~} \\
%       \rempfunp\colon & 6 & \longmapsto &  2 &   & \rempfun\colon & 6 & \longmapsto & \set{~\transpair{0}{2}~} \\
%     \end{array}
% %     \notag
%   \end{equation*}
  The transducer~$\Dst$ is shown in~\rfigure{d73}; its inaccessible part is dashed out.
  \end{enumerate}
\end{example}

\begin{figure}[p]
  \centering
  \includegraphics[scale=\TreeScale,trim={0pt 10mm 0pt 5mm}]{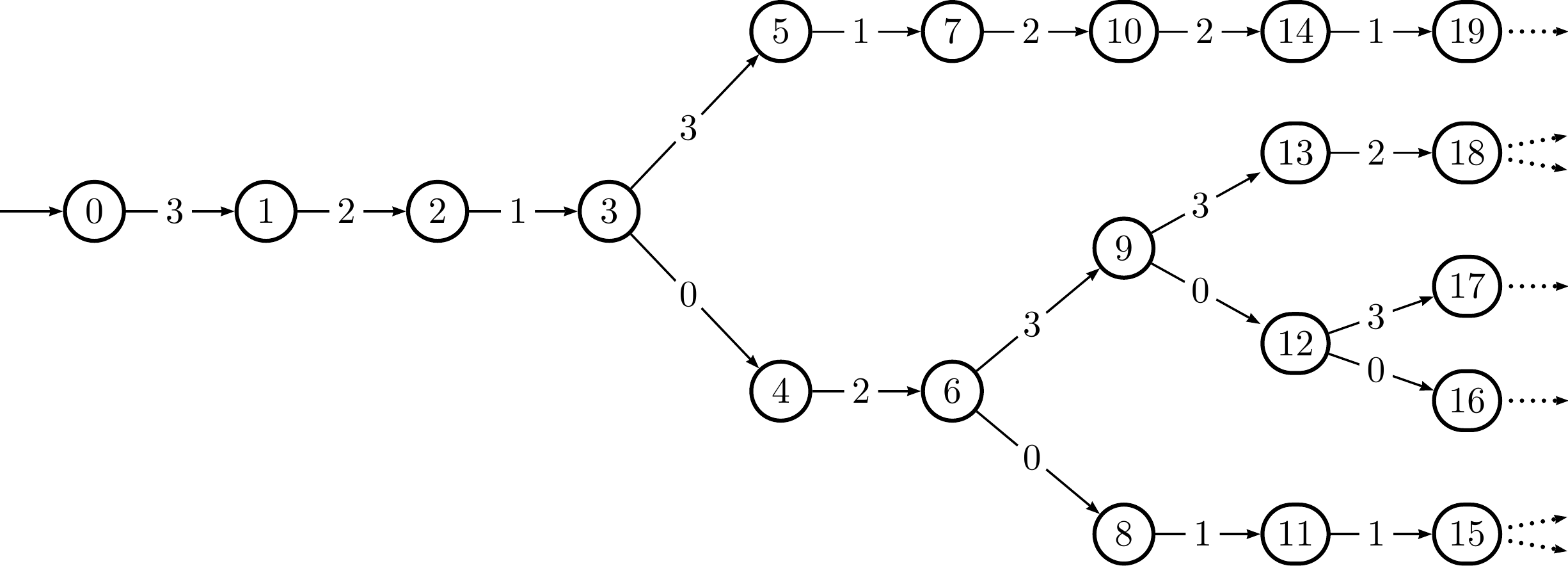}
  \caption{The language~$L_{\frac{4}{3}}$}
  \lfigure{l43}
\end{figure}
\begin{figure}[p]
  \centering
  \includegraphics[scale=\TreeScale,trim={0pt 10mm 0pt 5mm}]{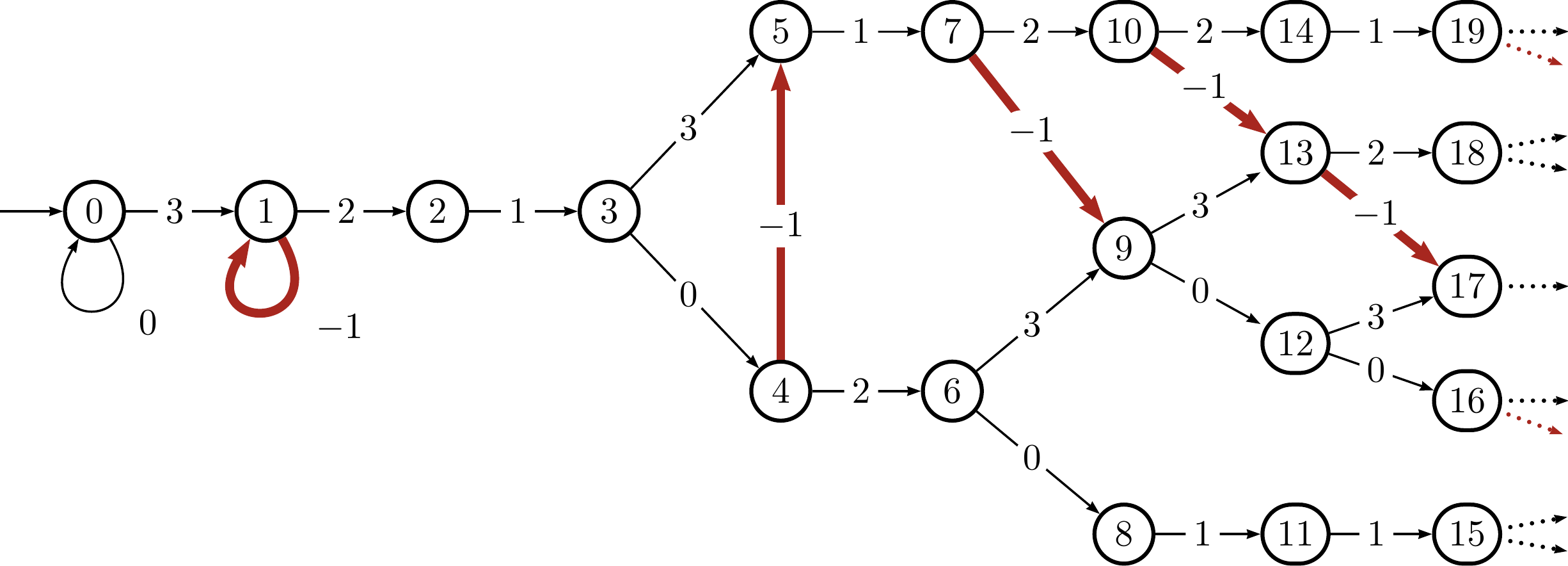}
  \caption{Transforming~$\Tqt$
            into~$\Tqtp$}
  \lfigure{t43p}
\end{figure}
\begin{figure}[p]
  \centering
  \includegraphics[scale=\TreeScale,trim={0pt 10mm 0pt 5mm}]{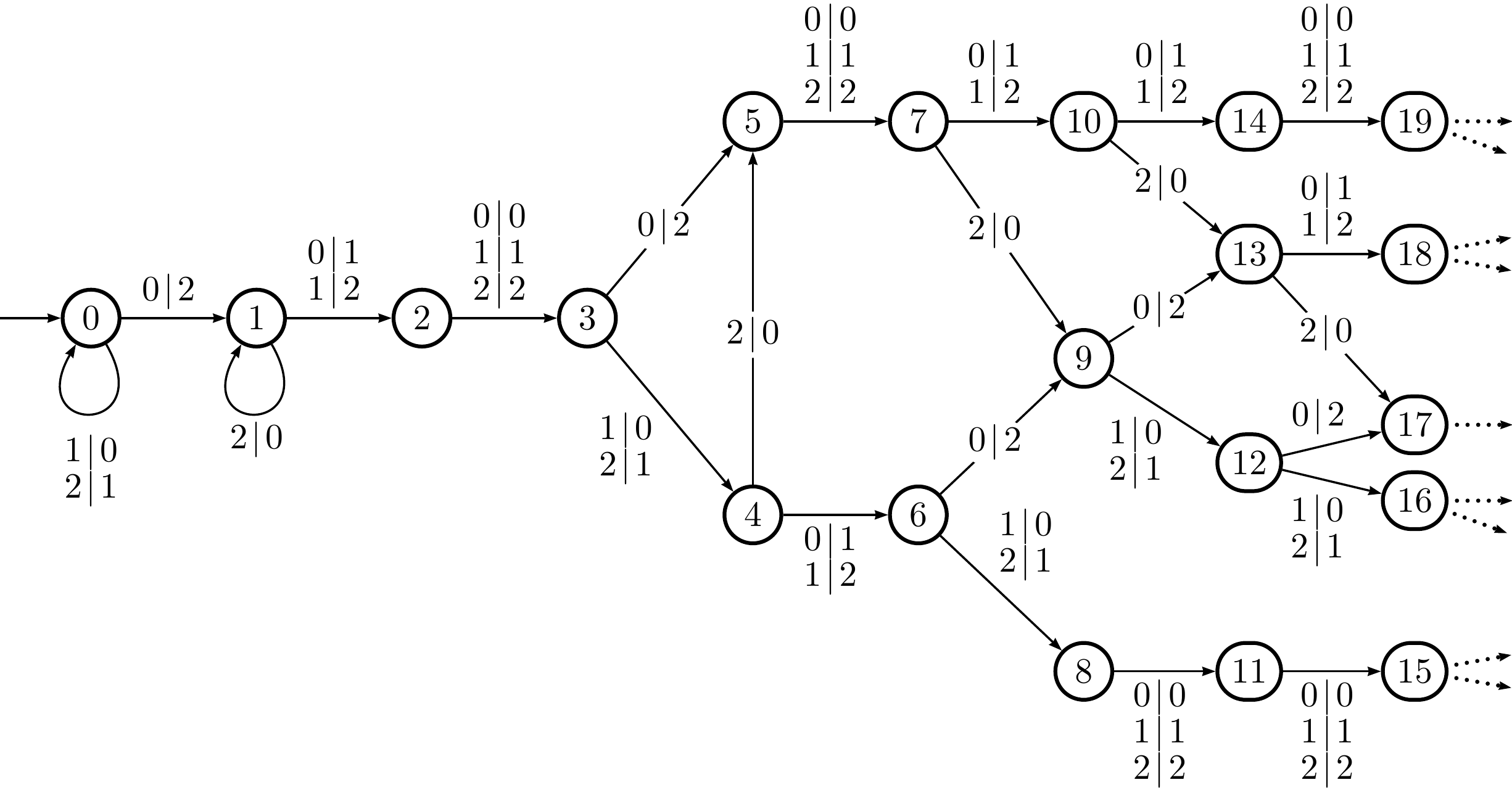}
  \caption{The   transducer~$\Dqt$}
  \lfigure{d43}
\end{figure}

\begin{figure}[t]
  \centering
  \includegraphics[scale=\TreeScale]{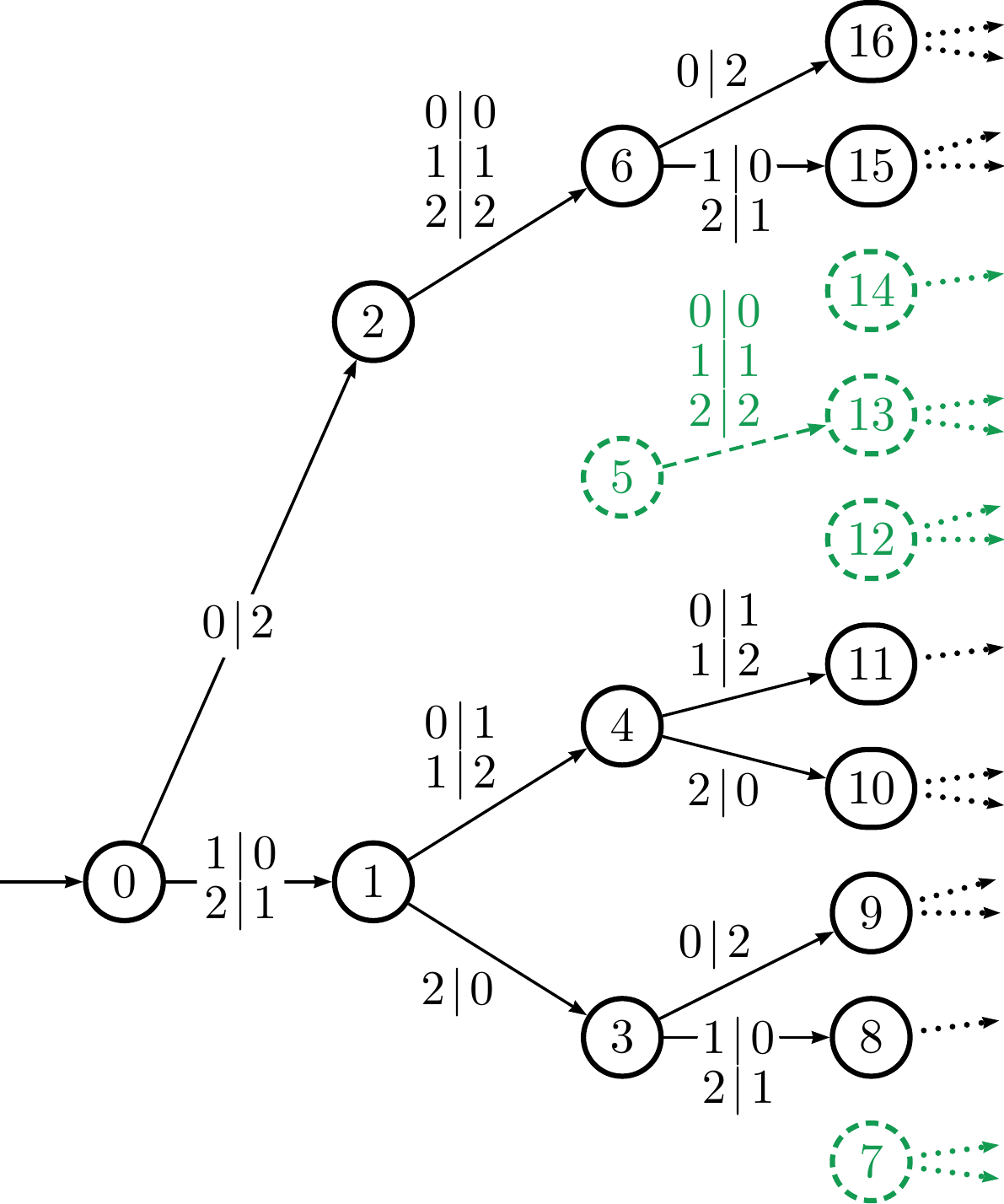}
  \caption{The transducer~$\Dst$}
  \lfigure{d73}
\end{figure}

%
% \begin{remark}
% Les fonctions~$\rempfun$ des exemples \rexample*{dqt} et \rexample*{dst} sont identiques à un décalage sur l'entrée près.
% Ceci toujours est vérifié dans le cas où deux bases~$\frac{p}{q}$ et~$\frac{p'}{q}$ ont le même dénominateur.
% % En général, les différentes parties de~$\Aq\times \Aq$ qui existent parmi les étiquettes des transitions de \Dpq* dépendent uniquement de~$q$.
% \end{remark}
%

%%%%%%%%%%%%%%%%%%%%%%%%%%%%%%%%%%%%%%%%%%%%%%%%%%%%%%%%%%%%%%%%%%%%%%%%%%%%%%%%%%%%%%%%%%%%%%%%%%
%%%%%%%%%%%%%%%%%%%%%%%%%%%%%%%%%%%%%%%%%%%%%%%%%%%%%%%%%%%%%%%%%%%%%%%%%%%%%%%%%%%%%%%%%%%%%%%%%%
\subsection{Behaviour of~$\Dpq$}

The transducer~$\Dpq$ is locally bijective, as both the underlying
input and the underlying output automata are complete deterministic
automata.
More precisely:

\begin{lemma}\llemma{dpq-loca-bije}
For every state~$n$ of~$\Dpq$ and every letter~$x$ in~$\Aq$, there
exist:\\[-.7ex]

\noindent
\hspace*{.8em}\enumstyle{a}\hspace*{.5em}
a unique transition
$\msp\displaystyle{n\pathx{(b,x)}[\Dpq]m}\msp$, and
\ee
\enumstyle{b}\hspace*{.5em}
a unique transition
$\msp\displaystyle{n\pathx{(x,b')}[\Dpq]m'}\msp$.

% %   \vspace*{-2mm}
%   \begin{enumerate}%[label=\alph{*}.]
%     \item \sublabel{l.dpq-comp}
%       There exists a unique letter~$b$ in $\Aq$ and a unique state~$m$ such that~.
%     \item \sublabel{l.dpq-cocomp}
%       There exists a unique letter~$b'$ in $\Aq$ and a unique state~$m'$ such that~$n\pathx{(x,b')}m'$.
%   \end{enumerate}
%   Let~$n$ be a state of~$\Dpq$ and~$x$ a letter in~$\Aq$.
% %   \vspace*{-2mm}
%   \begin{enumerate}%[label=\alph{*}.]
%     \item \sublabel{l.dpq-comp}
%       There exists a unique letter~$b$ in $\Aq$ and a unique state~$m$ such that~$n\pathx{(b,x)}m$.
%     \item \sublabel{l.dpq-cocomp}
%       There exists a unique letter~$b'$ in $\Aq$ and a unique state~$m'$ such that~$n\pathx{(x,b')}m'$.
%   \end{enumerate}
\end{lemma}

\begin{proof}
\enumstyle{a}\hspace*{1ex}
From \requation*{dpq-path},
$\msp\displaystyle{n\pathx{(b,x)}[\Dpq]m}\msp$ exists if and only if
$\msp\displaystyle{q\xmd m = p\xmd n + x-b + p -q}\msp$, that is,  if and only if
\begin{equation}\lequation{dpq-loca-bije-i}
  q\xmd m + b = p\xmd n + x + p -q \eqpnt
\end{equation}
The unicity of the pair~$(m,b)$ in \requation*{dpq-loca-bije-i}
follows, since~$b$ is in~$\Aq=\set{0,1,\ldots,q-1}$.
% Note that since and~$m\in\N$, it corresponds to the of~$(p\xmd n +x+(p-q))$ by~$q$:
% the quotient is~$m$ and the remainder is~$a$.
% %
% Thus, there is a unique pair .
%

%
A similar reasoning yields \enumstyle{b}.
\end{proof}
\begin{corollary}\lcorollary{dpq-stat-bije}
For every state~$n$ of~$\Dpq$ and every \oword~$\ow{w}$ in~$\Aqo$, there
exist:%\\[-.7ex]
\begin{enumerate}
\item a unique \oword~$\ow{u}$ in~$\Aqo$ such that
$\msp\displaystyle{n\pathx{(\ow{u},\ow{w})}[\Dpq]\cdots}\msp$, and

\item a unique \oword~$\ow{v}$ in~$\Aqo$ such that
$\msp\displaystyle{n\pathx{(\ow{w},\ow{v})}[\Dpq]\cdots}\msp$.
\end{enumerate}

%   Let~$n$ be a state of~$\Dpq$ and~$\ow{w}\in\Aqo$ an \oword.
% %
%   \begin{enumerate}
%     \item There exists a unique \oword~$\ow{u}\in\Aqo$ such that~$n\pathx{(\ow{u},\ow{w})}\cdots$
%     \item There exists a unique \oword~$\ow{v}\in\Aqo$ such that~$n\pathx{(\ow{w},\ow{v})}\cdots$
%   \end{enumerate}
% %     il existe un unique couple de mot infinis~$(w',w'')$ tels
% %     que~$i\pathx{\tpair{w}{w'}}\cdots$ et~$i\pathx{\tpair{w''}{w}}\cdots$
\end{corollary}

\begin{corollary}\lcorollary{dpq-bije}
  The transducer~$\Dpq$ realises a bijection:~$\Aqo\rightarrow\Aqo$.
\end{corollary}

%%%%%%%%%%%%%%%%%%%%%%%%%%%%%%%%%%%%%%%%%%%%%%%%%%%%%%%%%%%%%%%%%%%%%%%%%%%%%%%%%%%%%%%%%%%%%%%%%%%
%%%%%%%%%%%%%%%%%%%%%%%%%%%%%%%%%%%%%%%%%%%%%%%%%%%%%%%%%%%%%%%%%%%%%%%%%%%%%%%%%%%%%%%%%%%%%%%%%%%
% \subsection{Proof of \protect\rtheorem{dpq=xi}}

%
% In this section \thesubsection{} we show \rtheorem{dpq=xi}, recalled below.
%

% %
% \rtheorem{dpq=xi} requires a few additional statements.
% %

%
For every~$i$ in~$\N$, we define the transducer~$\Dpqi$ obtained
from~$\Dpq$ by changing  the initial state~$0$ into the state~$i$:
\begin{equation*}
  \Dpqi  = \aut{\, \Aq \times \Aq,\,\N,\, i,\, \delta \,}\eqpnt
\end{equation*}
\rtheorem{dpq=xi} is the direct consequence of the following more
general statement.
\begin{theorem}\ltheorem{dpqi}
For every integer~$n$,~$\Dpqi$ accepts the pair~$(\minword{n},\minword{n+i+1})$.
\end{theorem}

\begin{proof}Let us write:
%   \begin{subequations}
    \begin{align*}%{2}
      \minword{n} ={}& b_1\xmd b_2\cdots &&(\text{an \oword over } \Aq) \\
      \maxword{n+i}={}& c_1\xmd c_2\cdots &&(\text{an \oword over } \Amax)\\
      \maxword{n+i}\dwminus\minword{n} = \ow{u} ={}& a_1\xmd a_2\cdots &&(\text{an \oword over } \Bpq) \\
      \minword{n+i+1} ={}& b_1'\xmd b_2'\cdots &&(\text{an \oword over } \Aq)
    \end{align*}
%   \end{subequations}
%
By \rcorollary{span-auto++}, the \oword~$\ow{u}$ is the label of a \opath of~$\Tpqp$
originating from the state~$i$.
We write:
\begin{equation*}
    i \pathx{a_1}[\Tpqp] m_1 \pathx{a_2}[\Tpqp] m_2 \pathx{a_3}[\Tpqp] \cdots
\end{equation*}
For every index~$k$,~$\mu(c_k)=b_k'$ (\rlemma{mini-to-maxi}).
Hence~$(b_k,b_k')=(b_i,\mu(c_k))$ satisfies the three
conditions: ${b_k\in\Aq}$, ~$c_k\in\Amax$ and $a_k=c_k-b_k$~;
in other words,~$(b_k,b_k')$ belongs to~$\rempfun(a_k)$
(\rdefinition{remp-fun}).

It then follows from \rdefinition{dpq} of~$\Dpq$ that the following
\opath exists in~$\Dpq$:
\begin{equation*}
  i \pathx{(b_1,b_1')}[\Dpq] m_1
    \pathx{(b_2,b_2')}[\Dpq] m_2
    \pathx{(b_3,b_3')}[\Dpq] \cdots
\end{equation*}
In other words,~$\Dc_{\base,i}$ accepts the pair~$(\minword{n},\minword{n+i+1})$.
\end{proof}
%
%
%   %
%
%   %
%   \rtheorem{span-auto++} states that
%   %
%   For every integer~$k$, we previously showed that the pair~$(b_k,b_k')$ belongs to~$\psi(a_k)$,
%    %

%
In particular, \rtheorem{dpqi} implies, for~$i=0$, that~$\Dpq$ accepts every
pair~$(\minword{n},\minword{n+1})$, for~$n$ in~$\N$.
Since~$\Dpq$ is letter-to-letter (\rdefinition{dpq}), it realises a
\emph{continuous function};
since its domain is~$\Aqo$ (\rcorollary{dpq-bije}) and
since~$\minwords$ is dense in~$\Aqo$ (\rlemma{minw-dens}),
$\Dpq$ realises~$\Der\colon\Aqo\rightarrow\Aqo$.
This concludes the proof of \rtheorem{dpq=xi}.
%

%%%%%%%%%%%%%%%%%%%%%%%%%%%%%%%%%%%%%%%%%%%%%%%%%%%%%%%%%%%%%%%%%%%%%%%%%%%%%%%%%%%%%%%%%%%%%%%%%%%
%%%%%%%%%%%%%%%%%%%%%%%%%%%%%%%%%%%%%%%%%%%%%%%%%%%%%%%%%%%%%%%%%%%%%%%%%%%%%%%%%%%%%%%%%%%%%%%%%%%
%%%%%%%%%%%%%%%%%%%%%%%%%%%%%%%%%%%%%%%%%%%%%%%%%%%%%%%%%%%%%%%%%%%%%%%%%%%%%%%%%%%%%%%%%%%%%%%%%%%
\section{The set of spans}
\lsection{on-span}

The proof of \rtheorem{dpq=xi} draws the attention to the \owords
$\msp\spanword{n}=\maxword{n}\dwminus\minword{n}\msp$
% $\msp\ow{u}_{n}=\maxword{n}\dwminus\minword{n}\msp$
and naturally to their evaluation by the function~$\realval{}$.
For every integer~$n$, let us write~$u_{n}=\cod{n}$;
% and~$\ow{u}_{n}=\nobreak\spanword{n}$;
the real number~$\realval{u_{n}\spanword{n}}$ is the length of the
interval of the real line delimited, so to speak, by the `end-points'
of the \owords~$u_{n}\minword{n}$ and~$u_{n}\maxword{n}$ when the
representation trees are drawn in a \emph{fractal way}, as in the
first \rfigure{rep-tre-sys} or in the following~\rfigure{w32}.

Of course, this value will decrease exponentially with the
length~$~\ell$
of~$u_{n}$ and a reasonable `renormalisation' consists in considering
the value~$\realval{\spanword{n}}$ instead, which we call the
\emph{span of~$n$}.
%simply
In the case of a classical integer base numeration system, this
notion is obviously uninteresting as this value is~$1$ for every~$n$.
And it is as easy to observe, for instance on \rfigure{w32}, that in
a rational base numeration system, distinct integers may have
distinct spans.

In this section we study the topological structure of the set of spans
in a given system, and show that it depends upon whether~$z=\pq$ is
larger than 2 or not (\rtheorem{span}).

% %
% Figures~\rfigure*{w32} and~\rfigure*{w73} (resp.~pages~\pfigure*{w32} and
% \pfigure*{w73}) show the projection of~$\Wpq$ to an interval of~$\R$ by
% the evaluation function a.r.p.
% %
% If~$u$ is a word of~$\pref{\Wpq}$, then the subset of~$\Wpq$ of the words that starts with~$u$
% ks projected to an interval of~$\R$ included in~$\realval{\Wpq}$;
% the length of the interval is called the span of~$u$.
% %
% Observing the figures, one may see that two words~$u$, $v$
% in~$\pref{\Wpq}$ do not necessarily have the same span: in~\rfigure{w32}
% for instance, the span of~$u=2101$ is over twice the one of~$v=2122$.
% %
%
% %
% This definition of span is problematic since it decreases
% exponentially with the length of~$u$.
% %
% In the following, we consider a "normalised" span, equal to~$(\frac{p}{q})^d$ times
% the span, where~$d$ is the length of~$u$.
% %
% Then, two words~$u$ and~$v$ such that~$\val{u}=\val{v}$,
% have the same normalised span; and we can talk about the normalised
% span of an integer~$n$, which we simply call span of~$n$.
% %
%
% %
% In Section \thesection, we define directly the span of an integer
% (\rdefinition{spans}) and then study the topological closure of the set of the
% spans of integers.
% %
% We show that this set has different topological
% properties depending on whether~$z=\pq$ is larger than 2
% (\rtheorem{span}).
% %

\begin{figure}[ht!]
  \centering
  \includegraphics[width=0.8\linewidth]{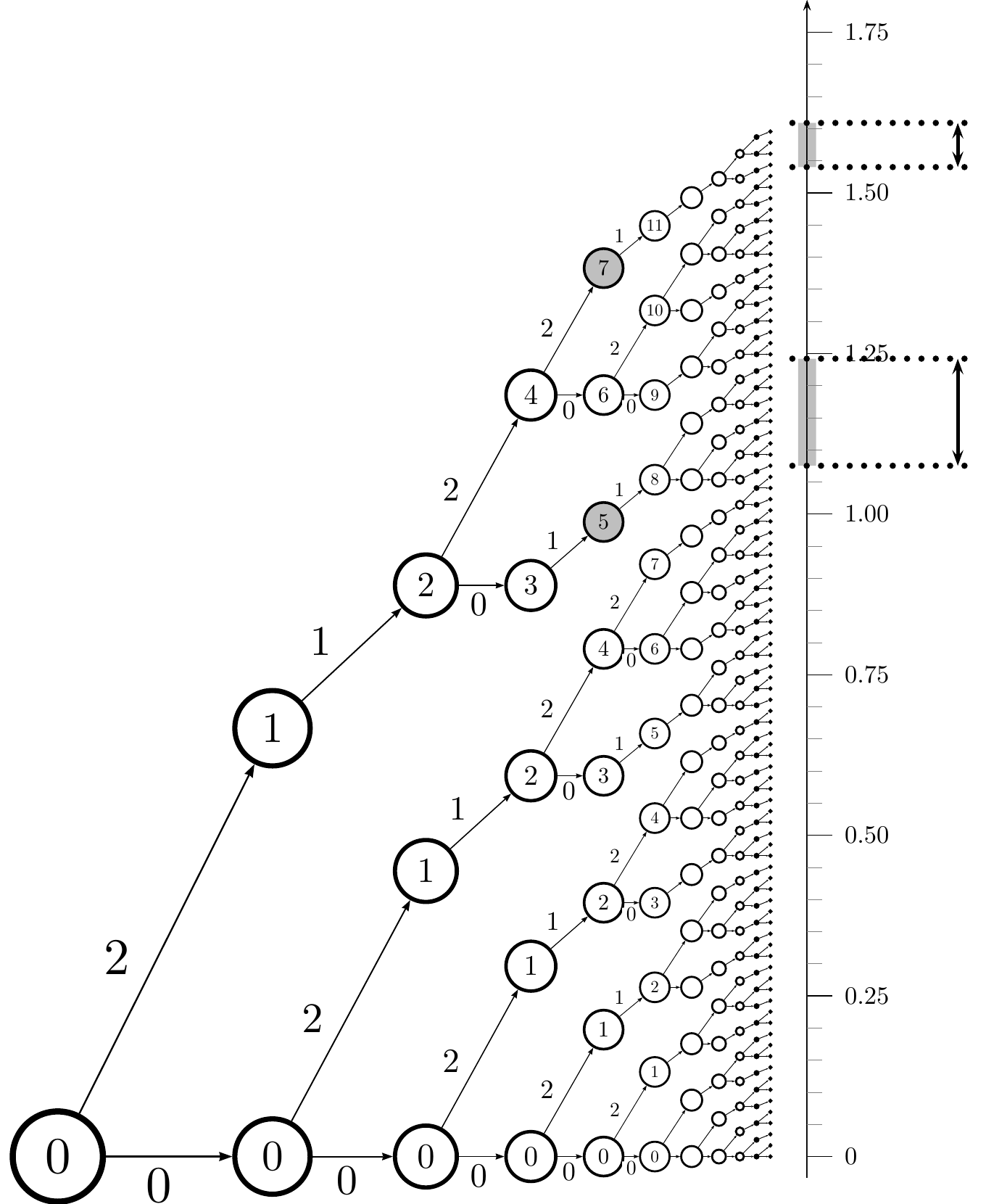}
  \caption{Fractal drawing of real number representations in
  base~$\td$}
%   \caption{Fractal Tree of the representation of real numbers in base~$\td$}
  \lfigure{w32}
\end{figure}

%%%%%%%%%%%%%%%%%%%%%%%%%%%%%%%%%%%%%%%%%%%%%%%%%%%%%%%%%%%%%%%%%%%%%%%%%%%%%%%%%%%%%%%%%%%%%%%%%%%
%%%%%%%%%%%%%%%%%%%%%%%%%%%%%%%%%%%%%%%%%%%%%%%%%%%%%%%%%%%%%%%%%%%%%%%%%%%%%%%%%%%%%%%%%%%%%%%%%%%
\subsection{Span of a node}

\begin{notation}\lnotation{defi-vn}
  For every integer~$n$, we denote by~$\Vn$ the set of all \owords~$\ow{w}$
  such that~$n\pathx{\ow{w}}\cdots$ is a branch of~$\Tpq$:
  \begin{equation*}
    \Vn =  {\cod{n}}^{\!-1} \xmd \Wpq = \Big\{\, \ow{w} ~\Big|~ \big(\cod{n}\ow{w}\big)\in\Wpq \,\Big\}\quad.
  \end{equation*}
\end{notation}

Note that~$\mathsf{V}_0=\Wpq$ and that for every integer~$n$, the \owords
$\minword{n}$ and~$\maxword{n}$ belong to~$\Vn$.
\rtheorem{ro-wpq-inte} states that~$\realval{\ows{V}_0}$ is an interval,
and  next proposition extends it to any~$\Vn$.

\begin{proposition}\lproposition{ro-vn-inte}
%   For every integer~$n$, the real-set~$\smash{\realval{\Vn}}$ is the interval
  \eqspacing $\realval{\Vn}= \Big[\, \realval{\minword{n}},~\realval{\maxword{n}} \,\Big]$
\end{proposition}

\begin{proof}
For readability, we write~$u_{n}=\cod{n}$ and
let~$\ell=\wlen{u_{n}}$.
From the \rdefinition{bot-top} of bottom and top words, every
word~$\ow{w}$ in~$\Vn$ satisfies
\begin{equation*}
  \minword{n} \lex\leq \ow{w} \lex\leq\maxword{n}
  \qquad\text{hence}\qquad
  u_{n}\minword{n} \lex\leq u_{n}\ow{w} \lex\leq u_{n}\maxword{n}
%   \cod{n}\minword{n} \lex\leq \cod{n}\ow{w} \lex\leq \cod{n}\maxword{n}
  \eqpnt
\end{equation*}
Conversely, since the prefix of length~$\ell$ of any \owords~$\ow{v}$ such that
$\msp u_{n}\minword{n} \lex\leq \ow{v} \lex\leq u_{n}\maxword{n}\msp$
is~$u_{n}$, it holds:
\begin{equation*}
u_{n}\Vn ={}  \big\{~ \ow{v}\in\Wpq \,~\big|~\,
             u_{n}\minword{n} ~{\lex\leq}~ \ow{v} ~{\lex\leq}~ u_{n}\maxword{n} ~\big\}
  \eqpnt
\end{equation*}
%   %
%   The three words on the right belong to~$\Wpq$.
%   %
Since~$\realval{}$ preserves order on~$\Wpq$ (\rproposition{ro-orde-pres}),
it follows that~$\realval{u_{n}\Vn}$ is an interval
since~$\realval{\Wpq}$ is an interval.

For any \oword~$\ow{v}$, it holds:
\begin{equation*}
\realval{u_{n}\ow{v}} =  \realval{u_{n}0^\omega} +
              \left(\frac{p}{q}\right)^{-\ell}\realval{\ow{v}}
\e\text{hence}\e
\realval{\ow{v}} =
\left(\frac{p}{q}\right)^{\ell}\realval{u_{n}\ow{v}} -
\left(\frac{p}{q}\right)^{\ell}\realval{u_{n}0^\omega}
\eqpnt
\end{equation*}
It follows that~$\realval{\Vn}$ is the image of the
interval~$\realval{u_{n}\Vn}$ by an affine transformation, hence an
interval.
\end{proof}

\begin{definition}\ldefinition{spans}
\begin{enumerate}
\item
For every integer~$n$, we call \emph{span of~$n$}, and
denote by~$\spann$, the length of the interval~$\realval{\Vn}$:
%   of \rproposition{ro-vn-inte}:
$\spann  = \realval{\maxword{n}}-\realval{\minword{n}}
         = \realval{\spanword{n}}\msp$.

\item We denote by~$\Spq$ the set of spans:
$\msp\Spq=\Defi{\spann}{n\in\N}=\Defi{\realval{\spanword{n}}}{n\in\N}\msp$.

\end{enumerate}
\end{definition}
%

% %
% \begin{property}\lproperty{span-word-eval}
%   For every integer~$n$, the \oword $\spanword{n}$
%   evaluates a.r.p.\@ to~$\spann$:
%   \begin{equation*}
%     \forall n\in\N \quantsp \realval{\spanword{n}} = \spann \eqpnt
%   \end{equation*}
% \end{property}
% %

%
Since the function~$\realval{}$ is continuous,
and~$\ibehav{\Tpqp}=\adh{\spanwords}$ (\rtheorem{span-auto}), the next statement holds.

\begin{theorem}\ltheorem{ro-L-Tpqp=adh-Spq}
%   The set of the \realvalues of the \olanguage of~$\Tpqp$ is the topological closure of~$\Spq$:
  \eqspacing$ \realval{\ibehav{\Tpqp}} = \adh{\Spq}$.
\end{theorem}

% purpose of this section is to study the
The topological properties of the set~$\adh{\Spq}\subset\R$
depend on whether~$p$ is smaller or greater than~$2\xmd q-1$.\footnote{%
  It could seem simpler to write: `whether~$\base$ is smaller or
  greater than~$2$' which is logically equivalent since~$z=2$
  defines an integer base rather than a rational base.
  But this would hide that the true \emph{border case} is when
  $\msp p = 2\xmd q-1\msp$ and this case behaves sometimes like
  $\msp p<2\xmd q-1\msp$ --- as here in \rtheorem{span} ---
  and sometimes like $\msp p>2\xmd q-1\msp$ --- as
  in Theorem~3 in \citet{AkiyEtAl08}.
  Note also that~$p$ and~$q$ coprime and~$p > 2\xmd q-1$
  imply~$p\geq 2\xmd q+1$.}
Before stating the result, let us recall a definition.
A bounded closed set that is nowhere dense and has no isolated point
is called a \emph{Cantor set}.
The classical ternary Cantor set is of measure zero, but it is not
necessarily the case of
all Cantor sets (\cf~\citealp{Kech95}).
%
% The topological properties of the set~$\adh{\Spq}\subset\R$
% %
% depend on whether~$\base$ is smaller or greater than~$2$
% (or, equivalent whether~$p$ is smaller or greater than~$2\xmd q-1$, since~$z=2$
% defines an integer base rather than a rational base).
% %
% Recall that a bounded closed set that is nowhere dense and has no isolated point
% is called a \emph{Cantor set}.
% %
% The classical ternary Cantor set is of measure zero, but it is not
% necessarily the case of
% all Cantor sets (\cf~\cite{Kech95}).
% %

\begin{falsetheorem}{\rtheorem*{span}}
  Let~$p,q$ be two coprime integers such that~$p>q>1$ and~$z=\frac{p}{q}$.
  \begin{enumerate}
      \item
      If~$p\leq 2\xmd q-1$, then~$\adh{\Spq}$ is equal to \emph{the interval}~$\realval{\Wpq}$.
      \item
      If $p > 2\xmd q-1$, then~$\adh{\Spq}$ is a Cantor set
      of measure zero.
    \end{enumerate}
\end{falsetheorem}

The two parts of \rtheorem{span} are shown independently
in \rsection{span-smal} and \rsection{span-larg}.

Beforehand, we give a characterisation of~$\adh{\Spq}$ that holds in
all cases but the status of which lies in between the two parts of
\rtheorem{span}.
%
% We choose to state it here although this choice is arguable as explained
% hereafter.
% %
% It holds in general (whether $p\leq (2\xmd q-1)$ or $p> (2\xmd q-1)$).
%
% When the base is small, the proof
For small bases, its proof
uses a result from the \textbf{next}
\rsection{span-smal} and, this part of the statement is never applied in the
following.
%
% When the base is large,
For large bases,
the proof is easy but will be used in the proof
of \rtheorem{span-big} later on.
Recall that~$\Bpq$ is the integer interval whose length
is~$2\xmd q-1$ and whose largest element is~$p-1$ (\rdefinition{bpq}).

\begin{proposition}\lproposition{adh-Spq-ro-Wpq-cap-Bpqo}
  \eqspacing$\adh{\Spq}=\realval{\Wpq \cap \Bpqo}$
\end{proposition}
\begin{proof}
  If~$p\leq 2\xmd q-1$, then~$\Ap\subseteq\Bpq$ (Properties~\rproperty*{bpq-smal}
  and~\enumstyle{\rproperty*{bpq-equ*}}),
  hence~$\Wpq\subseteq\Apo\subseteq\Bpqo$.
  It follows that~$\Wpq \cap \Bpqo=\Wpq$.
  We will see in the next \rsection{span-smal} (\rproposition{span-dist-smal})
  that~$\realval{\ibehav{\Tpq}}=\realval{\ibehav{\Tpqp}}$.
  Finally, \rtheorem{ro-L-Tpqp=adh-Spq} concludes the proof in this case:
  \begin{equation*}
  \adh{\Spq}=\realval{\ibehav{\Tpqp}}
            =\realval{\ibehav{\Tpq}}
	    =\realval{\Wpq}
	    =\realval{\Wpq \cap \Bpqo}
      \eqpnt
  \end{equation*}
  If~$p> 2\xmd q-1$,
  %Then,
  ~$\Tpqp$ is built from~$\Tpq$ by deleting
  the transitions labelled by~$\Ap\setminus\Bpq$.
  An \oword~$\ow{w}$ of~$\Apo$ is accepted by~$\Tpqp$ if and only
  if 1) it is accepted by~$\Tpq$ and 2) every digit of~$\ow{w}$ belongs to~$\Bpq$.
  In other words:
  \begin{equation*}
  \ibehav{\Tpqp} = \Wpq \cap \Bpqo  \eqpnt
  \end{equation*}
  Since by \rdefinition{wpq},~$\Wpq=\ibehav{\Tpq}$, \rtheorem{ro-L-Tpqp=adh-Spq} concludes the proof.
\end{proof}

%%%%%%%%%%%%%%%%%%%%%%%%%%%%%%%%%%%%%%%%%%%%%%%%%%%%%%%%%%%%%%%%%%%%%%%%
%%%%%%%%%%%%%%%%%%%%%%%%%%%%%%%%%%%%%%%%%%%%%%%%%%%%%%%%%%%%%%%%%%%%%%%%
\subsection{The span-set in small bases~($\msp p\leq 2\xmd q-1\msp$)} %$(z<2)$
\lsection{span-smal}

%
% In this subsection~$p\leq 2q-1$.
% The purpose of section \thesubsection{} is to establish \rtheorem{span-smal}.
%
First, we show that the shortest run reaching
a given state~$n$ has the same length in~$\Tpq$ and in~$\Tpqp$, that
is, the fact that in this case~$\Tpqp$ is obtained from~$\Tpq$ by
adding new transitions does not allow nevertheless any `shortcuts'.
%
% Then, we prove \rproposition{span-dist-smal}, of which
% \rtheorem{span-smal} is corollary.
% %

%
\begin{lemma}\llemma{not-bet-rep}
%   Let~$u$ be a word of~$\Bpqs$ accepted by~$\Tpqp$.
%   %
%   We write~$m=\val{u}$.
  Let~$u$ be in~$\behav{\Tpqp}$ and~$m=\val{u}$.
  If~$p\leq 2\xmd q-1$, then~${\wlen{\cod{m}}\leq\wlen{u}}$.
\end{lemma}

\begin{proof}
  By induction over the length of~$u$.
  The case~$u=\epsilon$ is trivial.
  Let~$u=u'd$ be a non-empty word over~$\Bpq$ that is accepted by~$\Tpqp$.
  If~$\val{u}=0$, then the lemma holds; we assume in the following
  that~$\val{u}>0$.
  We denote the run of~$u$ as follows:
  \begin{equation*}%\lequation{not-bet-rep-i}
    0 \pathx{u'}[\Tpqp] n \pathx{d}[\Tpqp] m\quad.
    \end{equation*}
  \rlemma{tpqp-val} yields that~$n=\val{u'}$ and~$m=\val{u}>0$.
  From induction hypothesis, it holds
  \begin{equation}\lequation{not-bet-rep-ii}
     \wlen{\cod{n'}} \leq \wlen{u'} \quad.
  \end{equation}
  Since~$z$ is a small base,~$\Ap$ is included in~$\Bpq$.
  The remainder of the proof depends on whether~$d$ belongs to~$\Ap$ or to~$\Bpq\setminus\Ap$.

  \medskip

  \noindent\textbf{Case 1: $d\in\Ap$.} Then, the transition~$n \pathx{d} m$
  exists in~$\Tpq$ (in addition to existing in~$\Tpqp$).
  Since moreover~$m\neq 0$, it follow that
  \begin{equation*}%\lequation{not-bet-rep-iii}
    \cod{m} = \cod{n} \xmd d
    \e\text{and hence}\e
    \wlen{\cod{m}} = \wlen{\cod{n} \xmd d} \leq \wlen{u'\xmd d} = \wlen{u}
    \eqpnt
  \end{equation*}
  %
  %
%   Finally, using \requation*{not-bet-rep-iii}, \requation*{not-bet-rep-ii},
%   and the definition of~$u$, it holds
%   \begin{equation*}
%     \wlen{\cod{m}} = \wlen{\cod{n} \xmd d} \leq \wlen{u'\xmd d} = \wlen{u}
%   \end{equation*}
%   Thus, the lemma holds in Case 1.
  %

  \medskip

  \noindent\textbf{Case 2: $d\notin\Ap$.}
  The digit~$d$ belongs to~$\Bpq\setminus\Ap$,
  hence is negative (\rproperty{bpq-smal}).
  We apply the Euclidean division algorithm to~$m$ (\requation{MEA-alt-2}
  since~$m>0$):
  there exists a unique pair~$(n',a)$ in~${(\N\times\Ap)}$ such
  that~$\cod{m}=\cod{n'}a$.
  Thus, the state~$m$ has in~$\Tpqp$ the two incoming
  transitions~$n'\pathx{b}m$ and $n\pathx{a}m$.
  Hence from \requation{tpqp-path},~$q\xmd m$ is both equal
  to~$n'p+b$ and~$np+a$.
  Since~$a$ is negative and~$b$ is not,~$n'<n$.
  Moreover, since representation in base~$\pq$
  preserves order (\rproposition{pq-ans})
  \begin{equation}\lequation{not-bet-rep-iv}
    \cod{n'} \rad\leq \cod{n}
    \qquad\text{hence,}\qquad\wlen{\cod{n'}} \leq \wlen{\cod{n}}\quad.
  \end{equation}
  Finally, we conclude Case 2 by applying in succession the definition
  of~$(n',a)$, and Equations \requation*{not-bet-rep-iv} and
  \requation*{not-bet-rep-ii}:
  \begin{equation*}%\notag
  \wlen{\cod{m}} = \wlen{\cod{n'}\xmd a}
        = \wlen{\cod{n'}}+1 \leq \wlen{\cod{n}}+1 \leq \wlen{u'}+1
	= \wlen{u} \quad. \qedhere
  \end{equation*}
\end{proof}
\begin{corollary}\lcorollary{Tpqp-eq-eval-leng}
For every~$u$ in~$\behav{\Tpqp}$, there exists~$v$ in~$\behav{\Tpq}$
such that
%   For every word~$u$ accepted by~$\Tpqp$, there exists
%   a word~$v$ accepted by~$\Tpq$ such that
  \begin{equation*}
    \val{u}=\val{v} \qquad\text{and}\qquad \wlen{u}=\wlen{v} \quad.
  \end{equation*}
\end{corollary}
\begin{proof}
With notation of \rlemma{not-bet-rep}, let~$v=0^i\cod{m}$
with the suitable number~$i$ of~0's.
\end{proof}
Next, we show that although~$\Tpqp$ accepts more \owords than~$\Tpq$,
the extra accepted \owords do not bring new \realvalues.
\begin{proposition}\lproposition{span-dist-smal}
If~$p\leq 2\xmd q-1$,
then~$\realval{\ibehav{\Tpq}}=\realval{\ibehav{\Tpqp}}$.
\end{proposition}

\begin{proof}
Since~$p\leq 2\xmd q-1$,~$\Ap$ is included in~$\Bpq$.
It follows that every transition of~$\Tpq$ also appears in~$\Tpqp$ and
every \oword of~$\ibehav{\Tpq}$ thus belongs to~$\ibehav{\Tpqp}$
hence~$\realval{\ibehav{\Tpq}}\nobreak\subseteq\nobreak\realval{\ibehav{\Tpqp}}$.
Let~$\ow{w}$ be an \oword in~$\Bpqo$ that is accepted by~$\Tpqp$.
For every integer~$i$, we denote by~$w_i$ the prefix of~$\ow{w}$
of length~$i$.
From~\rcorollary{Tpqp-eq-eval-leng}, there exists a finite word~$v_i$
accepted by~$\Tpq$ such that~$\wlen{v_i}=i$ and~$\val{v_i}=\val{w_i}$.
Since~$\Tpq$ has no dead-end (\rlemma{tpq-righ-exte}),
there exists an \oword~$\ow{u}_i\in\ibehav{\Tpq}$
that features~$v_i$ as prefix.
For every integer~$i$, the \owords~$\ow{w}$ and~$\ow{u}_i$ have
respective prefixes of length~$i$ with the same value.
It follows that
  \begin{equation*}
    \abs{\realval{\ow{w}}-\realval{\ow{u}_i}} < \sum_{n=i+1}^{\infty}\frac{\card{\Ap}+\card{\Bpq}}{q}\left(\frac{p}{q}\right)^{-i} \quad.
  \end{equation*}
  Hence,~$\left(\realval{\ow{u}_i}\right)$ tends to~$\realval{\ow{w}}$ when~$i$
  tends to infinity.
  Besides, since~$\realval{\ibehav{\Tpq}}$ is a closed set (\rtheorem{ro-wpq-inte}),
  $\realval{\ow{w}}$ belongs to~$\realval{\ibehav{\Tpq}}$.
  In other words, there exists an \oword~$\ow{v}$ in~$\ibehav{\Tpq}$
  such that~$\realval{\ow{v}}=\realval{\ow{w}}$.
  Hence,~$\realval{\ibehav{\Tpq}} \supseteq\realval{\ibehav{\Tpqp}}$.
\end{proof}
%

% \begin{falsetheorem}{\rtheorem*{span-smal}}
%   Let~$p,q$ be two coprime integers such that~$p>q>1$.
%   %
%   If~$p\leq (2 q-1)$, then the closure of~$\Spq$ is an interval.
% \end{falsetheorem}
\begin{proof}[of \rtheorem{span-smal}]
\rtheorem{ro-L-Tpqp=adh-Spq} and \rproposition{span-dist-smal} imply
  \begin{equation}
    \adh{\Spq} = \realval{\ibehav{\Tpqp}}
              = \realval{\ibehav{\Tpq}}
	      \eqvrg
%               =\realval{\Wpq}
%               =\left[0,\realval{\maxword{0}}\right]
  \end{equation}
and $\realval{\ibehav{\Tpq}}$ is a closed interval by
\rtheorem{ro-wpq-inte}.
\end{proof}

%%%%%%%%%%%%%%%%%%%%%%%%%%%%%%%%%%%%%%%%%%%%%%%%%%%%%%%%%%%%%%%%%%%%%%%%
%%%%%%%%%%%%%%%%%%%%%%%%%%%%%%%%%%%%%%%%%%%%%%%%%%%%%%%%%%%%%%%%%%%%%%%%
\subsection{The span-set in large bases~($\msp p> 2\xmd q-1\msp$)} %($z>2$)
\lsection{span-larg}

In order to prevent any misinterpretation in case of cursory reading,
we repeat the hypothesis $\msp p> 2\xmd q-1\msp$ in every statement.
%
% In this subsection~$p> 2\xmd q-1$.
%In this subsection~$p\ge 2\xmd q-1$.
%
The proof is divided in two parts:
\rproposition{adh-spq-no-isol-point}
and~\rproposition{wpqp-null-meas}.
% The proof is divided in the three Propositions
% \rproposition*{adh-spq-clos-boun}, \rproposition*{adh-spq-no-isol-point},
% and~\rproposition*{wpqp-null-meas}, the first of which is immediate.
%
% \begin{proposition}\lproposition{adh-spq-clos-boun}
%   The set~$\adh{\Spq}$ is closed and bounded.
% \end{proposition}
%
%
Let us recall first that the set~$\adh{\Spq}$ is closed and bounded
and then the following two properties that hold in large bases.
\begin{property}
We assume~$p> 2\xmd q-1$.
\begin{enumerate}
  \item \sublabel{pp.big-tpq-floor-z-succ}
    Every state~$n$ of~$\Tpq$ has at least
    $\left\lfloor\displaystyle{\frac{p}{q}}\right\rfloor$
    outgoing transitions.
  \item \sublabel{pp.big-bpq-posi}
    The digits of~$\Bpq$ are strictly positive.
\end{enumerate}
\end{property}

\begin{lemma}\llemma{big-span-posi}%
We assume~$p>2\xmd q -1$.
For every integer~$n$, it
holds~$0<\gammapq\leq\spann\leq\omegapq$,\linebreak[2]
where~$\omegapq=\realval{\maxword{0}}$
and~$\gammapq=\realval{q\xmd\minword{1}}$.
\end{lemma}

\begin{proof}
We denote by~$X$ the set consisting of the \owords of~$\Wpq$
that do not start with the digit~$0$.
Hence~$\maxword{0}$ and~$q\xmd\minword{1}$ are respectively the
greatest and the least \oword of~$X$ in the lexicographic ordering.
From \rproposition{ro-orde-pres} then follows that~$\realval{X}$
is a subset of~$[\gammapq,\omegapq]$.
On the other hand, it follows from
\rproposition{adh-Spq-ro-Wpq-cap-Bpqo} that~$\spann$
belongs to~$\realval{\Wpq\cap\Bpqo}$.
From \rproperty{big-bpq-posi},~$\Bpq$ does not contain the digit~$0$,
hence~$\msp\ibehav{\Tpqp}\subseteq X\msp$
and it holds: $\msp\spann\in\realval{X}\subseteq[\gammapq,\omegapq]\msp$.
\end{proof}

%
% \begin{proof}
% \enumstyle{a}\hspace*{.5em} follows from \requation{tpq-path},
%   %
% \e\enumstyle{b}\hspace*{.5em} from the \rdefinition{bpq} of~$\Bpq$.
%   %
%
%   %
% \enumstyle{c}\hspace*{.5em}
%   %
%   The real number~$\spann$ is the evaluation of the span-word
%   of~$n$ (\rdefinition{spans}).
%   %
%   This word is over~$\Bpq$, an alphabet containing only (strictly)
%   positive digits, hence:
%    \begin{equation*}
%     \spann = \realval{\spanword{n}} >
%     \realval{1^\omega}=\frac{1}{q}\sum_{i\geq 1}
%     \left(\frac{p}{q}\right)^{-i} > 0 \eqpnt \qedhere
%   \end{equation*}
%   %
%   $(d)\quad$
%   Let~$w\lex>w'$ be two \owords such that~$n\pathx{w}[\Tpq]\cdots$ and~$n\pathx{w'}[\Tpq]\cdots$.
%   %
%   WLOG we may assume that the first letters of~$w$ and~$w'$ differ an
%   %
%   In this case~$\realval{w}\geq \realval{w'}+\spann$.
%   %
% \end{proof}
% %

%
% Now, we start the proof of that~$\adh{\Spq}$ contains no isolated point
% (\rproposition{adh-spq-no-isol-point}) by a preliminary lemma.
%

%
\begin{lemma}\llemma{big-two-branch}
We assume~$p>2\xmd q-1$.
For every integer~$n$, there exist in~$\Tpqp$ two branches originating
from~$n$ that are labelled by \owords with distinct \realvalues.
%
% Let~$n$ be an integer.
% %
% There exist in~$\Tpqp$ two branches originating from~$n$
% that are labelled by \owords with distinct \realvalues.
% %
\end{lemma}
\begin{proof}
  We write~$\ow{w}=\maxword{n}$.
  Since~$\Amax$ is included in~$\Bpq$ (\rproperty{bpq-maxi}),
  all the transitions of the branch~$n\pathx{\ow{w}}\cdots$ of~$\Tpq$
  also exists in~$\Tpqp$.
  Since~$\ow{w}$ is the label of a branch of~$\Tpq$, \rlemma{tpq-flip} yields
  that it is not equal to~$(p-q)^\omega$.
  (Recall that~$p-q$ is the smallest letter of~$\Amax$.)
  Thus, there exists a digit~$a\in \Amax$,~$a>p-q$, a
  prefix~$u$ of~$\ow{w}$ and two states~$n',m$ such that
  \begin{equation*}
    n \pathx{u}[\Tpqp] n' \pathx{a}[\Tpqp] m
    \eqpnt
  \end{equation*}
  The integer~$(a-q)$ is greater than~$p-(2\xmd q-1)$ (and smaller
  than~$p-1$), hence  a letter of~$\Bpq$.
  %
%   In other words,~$(a-q)$ is
  %
  Then, the definition of~$\Tpqp$ (\requation{tpqp-path}) implies that
  \begin{equation*}
    n' \pathx{a-q}[\Tpqp] (m-1)
  \end{equation*}
  We denote by~$\ow{v}$ the word~$\ow{v}=u\xmd (a-q) \xmd \maxword{m-1}$,
  which labels a branch originating from~$n$.
  \rproposition{rela-mini-maxi} (\pproposition{rela-mini-maxi}) implies
  that the words~$\mes \ow{v}=u\xmd (a-q)\xmd \maxword{m-1}\mes$
  and~$\mes u\xmd a\xmd \minword{m}\mes$ have the same \realvalue.
  Hence it holds
  \begin{equation*}
    \realval{\ow{w}}-\realval{\ow{v}}
      = \realval{u\xmd a\xmd \maxword{m}} - \realval{u\xmd a\xmd \minword{m}}
      = \left(\pq\right)^{-\wlen{ua}} \spann[m]
    \eqpnt
  \end{equation*}
  Since~$z$ is a large base, every span is
  positive (\rlemma{big-span-posi}) and the lemma holds.
\end{proof}
\begin{proposition}\lproposition{adh-spq-no-isol-point}
  If~$p>2\xmd q -1$, the set~$\adh{\Spq}$ contains no isolated point.
\end{proposition}

\begin{proof}
% of \rproposition{adh-spq-no-isol-point}]
  Let~$x$ be a real number in~$\adh{\Spq}$.
  There exists an \oword~$\ow{w}=a_0\xmd a_1\cdots a_{i}\cdots$ accepted
  by~$\Tpqp$
  such that~$\realval{\ow{w}}=x$.
  We denote its \orun in~$\Tpqp$ as follows:
  \begin{equation*}
    0=n_0 \pathx{a_1}[\Tpqp] n_1
          \pathx{a_2}[\Tpqp] n_2
          \pathx{a_3}[\Tpqp]
          \cdots
  \end{equation*}
  Let~$k$ be a positive integer.
  We apply the previous \rlemma{big-two-branch} to~$n_k$ : there exist two
  \owords that label branches originating from~$n_k$ and that have different
  \realvalues.
  One of them must have a value distinct from~$\realval{a_{k+1}a_{k+2}\cdots}$;
  we denote this \oword by~$\ow{v}$.
  We moreover write~${\ow{v}_k=a_1\xmd a_2\cdots a_{k}\xmd \ow{v}}\mes$
  which then satisfies the following.
  \begin{gather}
    \lequation{adh-spq-no-isol-point-i}
    \ow{v}_k\in \ibehav{\Tpqp}
    \\
    \lequation{adh-spq-no-isol-point-ii}
    \realval{\ow{v}_k}\neq \realval{\ow{w}}
    \\
    \lequation{adh-spq-no-isol-point-iii}
    \text{$\ow{v}_k$ and~$\ow{w}$ have the same prefix of length~$k$}
  \end{gather}
  \rtheorem{ro-L-Tpqp=adh-Spq} yields
  that~$\realval{\ibehav{\Tpqp}} = \adh{\Spq}$
  and \requation{adh-spq-no-isol-point-i} that~$\realval{\ow{v}_k}$
  belongs to~$\adh{\Spq}$.
  From \requation*{adh-spq-no-isol-point-ii},~$\realval{\ow{v}_k}$
  indeed belongs to~$\big(\adh{\Spq}\setminus\{x\}\big)$.
  From \requation*{adh-spq-no-isol-point-iii}, the
  sequence~$(\ow{v}_k)_{k\in\N}$ tends to~$\ow{w}$.
  Finally, since~$\realval{}$ is
  continuous,~$\left(\realval{\ow{v}_k}\right)_{k\in\N}$ is a sequence
  of~$\adh{\Spq}\setminus\{x\}$
  which tends to~$\realval{\ow{w}}=x$.
\end{proof}

It remains to show that~$\adh{\Spq}$ is of measure zero.
Let us first recall the classical proof that the Ternary Cantor
set~$K_{3}$ has measure~$0$.
The set~$K_{3}$ is obtained from the interval~$I_{0}=[0,1]$ by
successive refinements.
At step~$n$, $I_{n}$ is a finite union of intervals~$J_{n,j}$,
every~$J_{n,j}$ is divided in three intervals of equal length,
and~$I_{n+1}$ is obtained by subtracting from each~$J_{n,j}$ the
(open) middle interval.
The measure of~$I_{n}$, that is, the sum of the lengths of the
disjoint~$J_{n,j}$ is~$\left(\frac{2}{3}\right)^{n}$.
The~$I_{n}$ form an infinite decreasing sequence of sets,
$\msp K_{3} = \bigcap_{n\in\N}I_{n}\msp$ and its measure is the limit
of the sequence~$\left(\frac{2}{3}\right)^{n}$, hence~$0$.
The proof of part~(b) of \rtheorem{span} follows the same scheme,
loaded with some technicalities.

\begin{lemma}\llemma{n-unre-m}
We assume~$p> 2\xmd q-1$.
Let~$i$ be an integer such that~$\floor{z}^i\geq 2\xmd q$.
Then, for every integer~$n$, there exists an integer~$m$
and a path~$\msp n\pathx{}[\Tpq]m\msp$ in~$\Tpq$ of length~$i$ that does not
exists in~$\Tpqp$.
\end{lemma}

\begin{proof}
  \rproperty{big-tpq-floor-z-succ} states that every state has at least~$\floor{z}$
  outgoing transitions in~$\Tpq$.
  Hence every state is the origin of at least~$\floor{z}^i$ distinct paths of length~$i$.
  Let~$n$ be a state and~$S$ the set of the states reachable from~$n$ in~$i$ steps.
  The cardinal of~$S$ is greater than~$2 \xmd q$ (previous paragraph) and~$S$
  is an integer interval.
  Hence~$S \mod p$ visits at least~$2\xmd q$ different residue classes modulo~$p$.
  Since the function mapping the residue classes of a state~$s$ and the label
  of the unique incoming transition of~$s$ in~$\Tpq$.
  The incoming transitions of the states of~$S$ are labelled by at least~$2\xmd q$ distinct
  letters.
  At least one of these letters does not belong to~$\Dpq$ (since it is of cardinal~$2\xmd q -1$);
  we denote by~$a$ this letter and by~$m$ a state of~$S$ the incoming transition of which is labelled by~$a$.
  The last transition of the path from~$n$ to~$m$ in~$\Tpq$ is deleted in~$\Tpqp$.
\end{proof}
For every finite word~$u$ in~$\pref{\Wpq}=0^{*}\Lpq$, we denote by~$\Zu$ the set
of the \owords that are accepted by~$\Tpq$ and that start
with~$u$:~$\mes\Zu =u \xmd (u^{-1}\xmd  \Wpq)$.
It is related to the sets~$\Vn$ (\rnotation{defi-vn}) by the following:
\begin{equation*}
  \forall u\in \pref{\Wpq}\quantsp \Zu = u \xmd \Vn \quad\quad \text{where} \quad n=\val{u} \eqpnt
\end{equation*}
Moreover, we denote by~$\Iu$ the set of the \realvalues of these
words:~$I_u=\realval{\Zu}$.
It then follows from the previous equation and \rproposition{ro-vn-inte} that
\begin{equation}\lequation{iu-expl}
  \forall u\in \pref{\Wpq}\quantsp I_u =
  \left[\realval{u\xmd \minword{n}},\,\realval{u\xmd\maxword{n}}\right]
  \quad\quad \text{where}\quad n=\val{u}\eqpnt
\end{equation}
When the base~$\pq$ is large,~$I_u$ is never reduced to a single element
since~$(\realval{u\minword{n}}-\realval{u\minword{n}})$ is equal
to $(\frac{p}{q})^{-\wlen{u}}\spann$,
a positive real from \rlemma{big-span-posi}.
Note also the following properties satisfied by these intervals:
\begin{gather}\lequation{iu-pref->subs}
  \forall u,v\in \pref{\Wpq} \quantsp u \text{ is a prefix of }v
  \implies \Iu\subseteq\Iu[v]\eqpnt \\[1ex]
  \lequation{iu-two-points->pref}
  \forall u,v\in \pref{\Wpq}
  \quantsp \Iu\cap\Iu[v] \text{ is non-trivial} \implies
    \left\{\begin{array}{ll}
      \text{either } &u\text{ is a prefix of }v \\
      \text{or} & v\text{ is a prefix of }u
    \end{array}\right.\\[1ex]
  \lequation{spli-Iu}
  \forall u\in\pref{\Wpq}\quantsp
  \Iu = \bigcup_{\substack{a\in\Ap\\
                           u\xmd a\in{\pref{W_{\hspace*{-0.5pt}\scriptstyle z}}}}}
        \Iu[ua] \eqpnt\eee
\end{gather}

\begin{figure}[p]\centering
% \noindent\hspace*{-1pt}\begin{minipage}[t]{0.75\linewidth}
  \subfigure[$\Wst$]{
  \includegraphics[height=0.9\textheight]{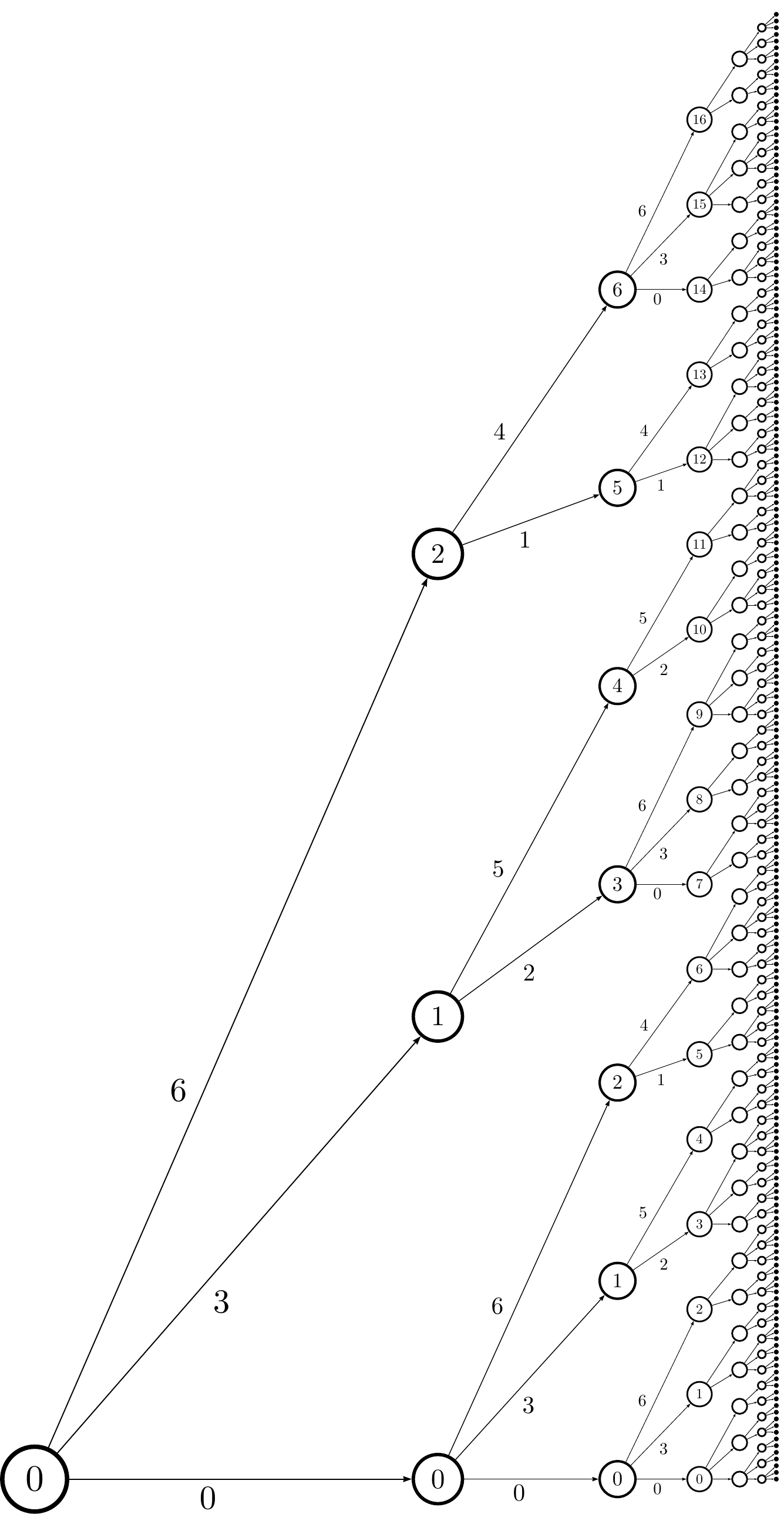}
%   \captionof{figure}{}
  \lfigure{w73}
  }
% \end{minipage}%
% \begin{minipage}[t]{0.25\linewidth}
  \subfigure[From left to right, ${\Iu[\epsilon]}$, $\itrans({\Iu[\epsilon]}),\ldots, \itrans^5({\Iu[\epsilon]})$]{
   \lfigure{ref-i73}%
    \includegraphics[height=0.9\textheight]{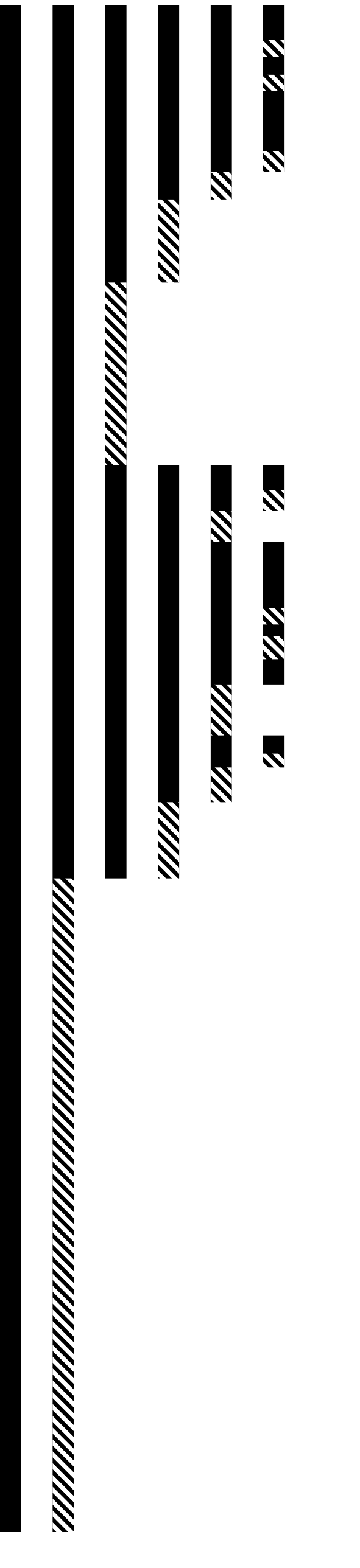}
%   \captionof{figure}{
  }
%   }

  \caption{Construction of~$\adh{\mathbf{Span}_{\frac{7}{3}}}$ by successive interval deletions}
\end{figure}

We denote by~$\II$ the set of all the intervals~$I_u$,
\begin{equation}
  \II = \big\{ I_u~|~u\in\pref{\Wpq}\big\} \quad,
\end{equation}
and by~$\itrans$ the function~$\pII\rightarrow\pII$
defined as follows.
\begin{equation}\lequation{defi-itra}
  \forall \pS\in\powerset(\II)\quantsp
  \itrans(\pS) = \Big\{~ \Iu[u\xmd d]~~\Big|~~\Iu \in \pS
                                           ~,~~
                                           ud\in\pref{\Wpq}
                                           ~~\text{and}~~
                                           d\in\Bpq
                                           ~\Big\}
\end{equation}

% \begin{subequations}\lequation{defi-itra}
% \begin{align}\lequation{defi-itra-i}
%
%
%
%   \forall u\in\pref{\Wpq} \quad \itrans\big(\{\Iu\}\big) ={}&
%                   \big\{ \Iu[u\xmd d] ~\big|~ d\in\Bpq,~ u\xmd d\in\Wpq \big\}  \\
%   \lequation{defi-itra-ii}
%   \forall \pS\in\powerset(\II)\quad \itrans(\pS) ={}&
%                                             \bigcup_{I\in \pS} \itrans(I)
% \end{align}
% \end{subequations}
%
% Note that by definition of~$\Iu$

%
In \requation*{defi-itra}, the variable~$d$ is taken in~$\Bpq$
whereas in~\requation*{spli-Iu} the variable~$a$ is taken in~$\Ap$.
When~$z$ is a large base,~$\Bpq$ is strictly included~$\Ap$, hence~$\itrans$ is a
refinement function:
\begin{equation*}
  \forall \pS\in\pII \quantsp \left( \Cup_{I\in \itrans(\pS)} I \right)\subseteq \Cup_{I\in \pS} I \eqpnt
\end{equation*}
\rfigure{ref-i73} shows the successive applications
of function~$\itrans$ to~$\Iu[\epsilon]$ in the large base~$z=\st$.
Hashed segments contain the points that are removed by the last application
of~$\itrans$.
\begin{lemma}\llemma{split-conv}
  We assume that~$p>2\xmd q -1$.
  Let~$\pS_0=\{I_\epsilon\}$ and for every integer~$j$,~$\pS_{j+1}=\itrans(\pS_j)$.
  Moreover, for every integer~$j$ we write~$U_j=\big(\Cup_{I\in\pS_j} I \big)$.
  Then, it holds
  \begin{equation*}
    \bigcap_{j\geq 0} U_j =  \adh{\Spq} \eqpnt
  \end{equation*}
\end{lemma}
\begin{proof}
  Right inclusion.
  Let~$x$ be in~$\adh{\Spq}$ and~$\ow{w}$ a word in~$\adh{\spanwords}$
  such that~$\realval{\ow{w}}=x$.
  From \rtheorem{span-auto},~$w\in\Wpqp$.
  We fix an integer~$j$ and denote by~$u$ the prefix of length~$j$ of~$\ow{w}$,
  hence~$u$ belongs to~$\Bpqs$.
  Inductively applying~\requation{defi-itra} yields that
  \begin{multline*}
    \pS_j = \itrans^{\,j}(\pS_0)
          = \{ I_v ~|~ v\in X_j\}~, \\
          \text{with }~
             X_j = \big\{ v\in\Bpqs\cap\pref{\Wpq} ~\big|~  \wlen{v} = j \big\}\quad.
  \end{multline*}
  It may be verified that~$u$ belongs to~$X_j$, hence~$I_u$ to~$\pS_j$.
  Since by definition~${\Iu=\realval{\Zu}}$ and~$\mes\ow{w}\in\Zu$, the
  number~$\mes x=\realval{\ow{w}}$ belongs to~$I_u$ hence to~$U_j$.
  Since~$j$ was taken arbitrarily, it follows that~$x$ belongs to~$\bigcap_{j\geq 0} U_j$.
  Hence, it holds
  \begin{equation*}
    \bigcap_{j\geq 0} U_j \supseteq\adh{\Spq}\eqpnt
  \end{equation*}

  \medskip

  Left inclusion.
  Let~$x$ be a real number of~$\bigcap_{j\geq 0} U_j$.
  For every integer~$j$, the number~$x$ belongs to~$U_j$, hence to some
  interval~$I_{u_j}$ of~$\pS_j$, where~$u_j\in\Bpqs$ and~$\wlen{u_j}=j$.
  Therefore, there exists an \oword~$\ow{w}_j$ in~$\Wpq$ that starts
  with~$u_j$ and that evaluates a.r.p.\@ to~$x$.
  In particular, note that the first~$j$
  letters of~$\ow{w}_j$ belong to~$\Bpq$.
  The topology on~$\Apo$ implies that every infinite sequence
  has a convergent sub-sequence.
  We denote by~$\ow{w}$ the limit of an arbitrary convergent sub-sequence of~$(\ow{w}_j)_{j\in\N}$.
  Since~$\Wpq$ is
  closed,~$\ow{w}$ is a word of~$\Wpq$.
  Since for every integer~$j$, the first~$j$
  letters of~$\ow{w}_j$ belong to~$\Bpq$,~$\ow{w}$ also belongs to~$\Bpqo$.
  Since~$\realval{}$ is continuous,~$\realval{\ow{w}}=x$ and then,~$x$ belongs to~$\realval{\Bpqo\cap\Wpq}$.
  Finally, \rproposition{adh-Spq-ro-Wpq-cap-Bpqo} yields that
%   \rproposition{adh-Spq-ro-Wpq-cap-Bpqo} yields that ~$\realval{\Bpqo\cap\Wpq}=\adh{\Spq}$,
%   hence that~$x\in\adh{\Spq}$
%   and finally
  \begin{equation*}
    \bigcap_{j\geq 0} U_j \subseteq  \realval{\Bpqo\cap\Wpq}=\adh{\Spq} \eqpnt \qedhere
  \end{equation*}
\end{proof}

We denote by~$\ell$ the Lebesgue measure on~$\R$.
Then, from \requation*{iu-expl}, it holds
\begin{equation}
  \forall u\in\pref{\Wpq}\quad \ell(I_u)
    = \left(\pq\right)^{-\wlen{u}}\spann
      \quad\quad\text{where~$n$ is such that}\quad
      0\pathx{u}[\Tpq]n
  \eqpnt
\end{equation}
Applying \rlemma{big-span-posi} then implies the following.

\begin{lemma}\llemma{l-Iu-appr}
  If~$p>2\xmd q -1$, then for every~$u\in\pref{\Wpq}$,
  it holds
  \begin{equation*}
    0<\gammapq\xmd z^{-\wlen{u}}\leq\ell(\Iu)\leq \omegapq\xmd z^{-\wlen{u}} \eqpntvrg
  \end{equation*}
  (recall that~$\gammapq=\realval{q\xmd \minword{1}}$ and~$\quad\omegapq=\realval{\maxword{0}}$).
\end{lemma}
By abuse of notation, we use~$\ell$ on elements of~$\pII$ with the following meaning
\begin{equation}
  \forall \pS\in\pII\quad \ell(\pS) = \ell\left(\,\bigcup_{I\in \pS} I\right)\eqpnt
\end{equation}

\begin{lemma}\llemma{spli-alph}
  If~$p>2\xmd q-1$, there exists a positive integer~$i$ and a real number~$\alpha$,~${0<\alpha<1}$,
  such that for every word~$u\in\pref{\Wpq}$,
  \begin{equation*}
    \ell\left(\itrans^{\,i}(\Iu)\right) \leq \alpha \xmd \ell (\Iu)
    \eqpnt
  \end{equation*}
\end{lemma}
\begin{proof}
  We choose~$i$ as in \rlemma{n-unre-m}.
  Let~$u$ be a word of~$\pref{\Wpq}$.
  We denote by~$k$ the length of~$u$ and by~$n$
  the state reached by the run of~$u$ in~$\Tpq$.
  Then, from \requation*{defi-itra} and the contrapositive
  of~\requation*{iu-two-points->pref},
  it holds
  \begin{multline*}
    \ell \left(\itrans^{\,i}(\Iu)\right) =
        \ell\bigg(\bigcup_{w\in X} \Iu[uw] \bigg) =
        \sum_{w\in X} \ell(\Iu[uw]) \\[0.2em]
        \text{where}\quad X=\big\{w\in{\Bpqs}~\big|~ \wlen{w}=i,~ uw\in\pref{\Wpq} \big\}
    \eqpnt
  \end{multline*}
  Then, \requation{spli-Iu} yields that
  \begin{multline*}
    \ell \left(\itrans^{\,i}(\Iu)\right)
        = \ell(I_u) - \sum_{w\in Y} \ell(\Iu[uw]) \\[0.2em]
    \text{where}\quad Y= \big\{w \in \Aps~\big|~w\mathbin{\notin}\Bpqs,~\wlen{w}=i,~uw\in\pref{\Wpq} \big\} \eqpnt
  \end{multline*}
  Now, we apply \rlemma{n-unre-m} to~$n$ : there exists a state~$m$ and a word~$v$ of
  length~$i>0$ such that the path~$n\pathx{v}m$ exists in~$\Tpq$ but does not exists in~$\Tpqp$.
  Hence,~$v$ features
  a digit that belongs to~$\Ap\setminus\Bpq$.
  It follows that~$v$ belongs to~$Y$, hence that the following holds.
  \begin{align}
    \ell \left(\itrans^{\,i}(\Iu)\right) \notag
        \leq{}& \ell(\Iu) - \ell(\Iu[uv]) \notag\\
        \leq{}& \ell(\Iu) \left(1 - \frac{\ell(\Iu[uv])}{\ell(\Iu)} \right) \notag\\
        \leq{}& \ell(\Iu) \left(1 - \frac{z^{-\wlen{uv}}\gammapq}{z^{-\wlen{u}}\omegapq}\right)
        \quad\quad\text{(Using \rlemma{l-Iu-appr})}
        \notag \\
        \leq{}& \alpha \xmd \ell(\Iu) \quad\quad\text{with}\quad \alpha=1 - z^{-i}\frac{\gammapq}{\omegapq}\notag
%
%         \leq{}& \ell(\Iu) \left(1 - \frac{(\frac{p}{q})^{-k-i}\,\spann[m]}{(\frac{p}{q})^{-k}\,\spann[n]} \right) \notag\\
%         \leq{}& \ell(\Iu) \left(1 - \left(\frac{p}{q}\right)^{-i}\frac{\spann[m]}{\spann[n]} \right) \lequation{spli-alph-i}
  \end{align}
%   We choose~$\alpha$ as above.
  %
  Since~$\gammapq$ and~$\omegapq$ are positive, then~$\alpha<1$.
  Since~$i$ is positive and~$\gammapq \leq \omegapq$, then~$\alpha>0$.
%   \begin{equation*}
%     \alpha= 1- \left(\frac{p}{q}\right)^{-i} \frac{x}{y}  \quad\quad\text{where}\quad
%     \begin{array}{l}
%       \displaystyle x= \inf \big(\Spq\big) \\
%       \displaystyle y= \sup \big(\Spq\big)
%     \end{array}
%   \end{equation*}
%   Since~$z$ is a large base, every span is greater than a positive real number~$\eta$
%   (\rlemma{big-span-posi}) and~$x,y$ are positive,
%   hence~$\alpha<1$.
%   %
%   Moreover, since~$i$ is positive and since~$x\leq y$,~$\alpha$ is positive.
%   %
%   Then, the whole statement follows from~$\requation*{spli-alph-i}$.
\end{proof}

\begin{proposition}\lproposition{wpqp-null-meas}
  If~$p> 2\xmd q-1$, then~$\adh{\Spq}$ is of measure zero.
\end{proposition}
\begin{proof}%[Proof of \rproposition{wpqp-null-meas}]
  Let~$(U_j)_{j\in\N}$ be the sequence defined in \rlemma{split-conv}.
  Let~$i,\alpha$ be the two parameters from \rlemma{spli-alph}.
  Applying the later yields
  \begin{equation}\lequation{wpqp-null-meas-i}
    \forall k\in\N \quantsp \ell(U_{k\xmd i}) < \alpha^{k}\,\ell(U_0)\eqpnt
  \end{equation}
  Since the sequence~$\big(U_{j}\big)_{j\in\N}$ is decreasing by inclusion,
  the sequence~$\big(\ell(U_{j})\big)_{j\in\N}$ is decreasing.
  Then, \requation*{wpqp-null-meas-i} implies that the later sequence tends
  to~$0$ when~$j$ tends to infinity.
  Finally, the set~$\adh{\Spq}$, which is the limit of the sequence~$\big(U_{j}\big)_{j\in\N}$
  (\rlemma{split-conv}),
  is of measure zero.
\end{proof}

%%%%%%%%%%%%%%%%%%%%%%%%%%%%%%%%%%%%%%%%%%%%%%%%%%%%%%%%%%%%%%%%%%%%%%%%
%%%%%%%%%%%%%%%%%%%%%%%%%%%%%%%%%%%%%%%%%%%%%%%%%%%%%%%%%%%%%%%%%%%%%%%%
% \subsection{ Hausdorff dimension of span-sets in large bases~$(p>2\xmd q -1)$}
% \lsection{Hau-dim-spa}

%
\paragraph{Hausdorff dimension}
% When it comes to sets of measure zero, the notion of Hausdorff dimension allows
% to get more precise informations of their content (\cf~\cite{Falc14}).
% %
% For instance, it is known that the Hausdorff dimension of the Ternary Cantor set
% is~$\frac{\ln 2}{\ln 3 }$\,; and a construction keeping~$k$ parts among~$n$
% usually results in sets of Hausdoff dimension~$\frac{\ln k}{\ln n}$\,.
% %
% In the case of~$\adh{\Spq}$, we keep ``in average''~$2\xmd q -1$ parts among~$p$,
% and one would expect the Hausdorff dimension to be~$\frac{\ln 2\xmd q -1}{\ln p}$\,.
% %
% However we will show that it is strictly smaller.
%
One can go further in the comparison between the Cantor sets and the
span-sets, and investigate their \emph{Hausdorff dimension} which
give more accurate information on their topological structure
(\cf\citealp{Falc14}).
It is known that the Hausdorff dimension of the Ternary Cantor set~$K_{3}$
is~$\msp\frac{\ln 2}{\ln 3}\msp$.
Generalization of the construction that yields~$K_{3}$, in which~$k$
parts out of~$n$ are kept usually results in sets with Hausdoff
dimension~$\msp\frac{\ln k}{\ln n}\msp$.
In the case of~$\adh{\Spq}$, we keep ``in average''~$2\xmd q-1$ parts
out of~$p$, and one could expect a Hausdorff dimension
of~$\msp\frac{\ln(2\xmd q-1)}{\ln p}\msp$.
We show below that this dimension is indeed strictly smaller.

Given a set~$F$, the \emph{$d$-dimensional Hausdorff measure} of~$F$ is defined by
\begin{equation*}
  \hausmeas[d]{F} = \lim_{\epsilon \rightarrow 0}~\inf \left\{~ \sum_i{} {r_i}^{\!d} ~~\middle|~~
    \begin{array}{@{}l}
    \text{there is a countable cover of }F\text{ by balls } B_0,B_1,\ldots\\
    \text{such that for every }i\text{, }B_i\text{ has radius }r_i\text{ and }r_i\mathbin{<}\epsilon.
    \end{array}
  \right\}\eqpnt
\end{equation*}
Then, the Hausdorff dimension of~$F$ is defined by:
\begin{equation*}
  \hausdim{F} = \inf\left\{ d\mathbin{\geq}0 ~\middle|~ \hausmeas{F} \mathbin{=} 0\right\} \eqpnt
\end{equation*}

\begin{proposition}\lproposition{haus}
  If~$p> 2q-1$, then~$\msp\displaystyle{\frac{\ln 2}{\ln p -\ln q}}\msp$
  is an upper bound for
  the Hausdorff dimension of~$\adh{\Spq}$.
%   is smaller than or equal
%   to~$\msp\displaystyle{\frac{\ln 2}{\ln z}}\msp$.
\end{proposition}

\begin{proof}
We compute indeed an upper bound for the \emph{Minkovski}, or
\emph{box-counting} dimension, which is known to be an upper bound for
the Hausdorff dimension.
  Let~$r$ be a positive real number.
  We denote by~$N(r)$ the minimal number of interval of length~$r$
  required to cover~$\adh{\Spq}$.
  Let~$d$ be a positive real number.
  The remainder of the proof consists in majoring~$N(r)\xmd r^{d}$.
  In the process of deleting edges from~$\Tpq$ to build~$\Tpqp$,
  there are at most two surviving edges coming out from every node.
  Hence, at the depth~$i$ of~$\Tpqp$, there are at most~$2^i$ nodes accessible
  from the root.
  We fix~$i$ as follows:
  \begin{gather}
    i = \bceil{ \frac{\ln \omegapq - \ln r}{\ln z} }      \lequation{haus-i}
%     \\
%     \lequation{haus-r}
%     z^{i-1} \leq  \frac{\omegapq}{r} < z^i
  \end{gather}
  Note in particular that from \rlemma{l-Iu-appr}, it holds:
%   of length~$i-1$:
  \begin{equation*} \lequation{haus-l-Ui}
    \forall u\in \pref{\Wpq} \quantvrg \wlen{u}=i\quantsp \ell(\Iu)< \omegapq z^{-i} < r \eqpnt
  \end{equation*}
  Hence, one interval of length~$r$ is enough to cover~$I_u$ and then
  \begin{equation*}
    N(r) < 2^i \leq 2^{\left(\textstyle \frac{\ln \omegapq - \ln r}{\ln z}+1\right)} = \eta\xmd
    \big(r\big)^{\textstyle-\frac{\ln 2}{\ln z}}
    \quad\quad\text{with}\quad\eta = 2^{\textstyle\frac{\omegapq+\ln z}{\ln z}} \\
  \end{equation*}
  Hence,~$r^{d} N(r)$ is smaller than a constant times~${r\rule{0pt}{2ex}}^{\textstyle\big(d-\frac{\ln 2}{\ln z}\big)}$.
  If moreover~$d>\frac{\ln 2}{\ln z}$, then~$N(r) \xmd r^d$ tends to~$0$
  when~$r$ tends to~$0$.
  Since for every real~$r$, we may cover~$\adh{\Spq}$ with~$N(r)$ intervals of length~$r$,
  it holds
  \begin{equation*}
    \forall d>\frac{\ln 2}{\ln z}\quantsp \hausmeas[d]{\adh{\Spq}} \leq \lim_{r\rightarrow 0} N(r) \xmd r^d = 0 \eqpnt
    \qedhere
  \end{equation*}
\end{proof}
In all cases different from~$z=\frac{5}{2}$, the
bound~$\frac{\ln2}{\ln z}$ is better (smaller)
than the bound~$\msp\frac{\ln(2\xmd q-1)}{\ln p}\msp$ that was
inspired by the example of Cantor sets.
This can be seen by means of some classical (though sometimes
tedious) computations.
The case~$z=\frac{5}{2}$ is dealt with in a very similar way.
In this case, every node in~$\Tpq$ possesses at most~$3$ surviving
paths of length~$2$ that remains in~$\Tpqp$.
This yields a bound~$\frac{\ln3}{2\xmd\ln\frac{5}{2}}$ which is easily
checked to be smaller than the corresponding
bound~$\msp\frac{\ln3}{\ln5}\msp$.
%

%
% Some tedious calculation shows that if~$z\neq \frac{5}{2}$, it holds
% \begin{equation*}
%   \frac{\ln 2}{\ln z} < \frac{\ln (2\xmd q-1)}{\ln p} \eqpnt
% \end{equation*}
% From \rproposition{haus} then follows that the Hausdorff dimension of~$\adh{\Spq}$
% is strictly smaller than~$\frac{p}{2\xmd q-1}$.
% %
% For the special case~$z=\frac{5}{2}$, with a reasoning similar to the proof of \rproposition{haus},
% one can show that the Hausdorff dimension is at most~$\frac{\ln 3}{2\xmd \ln z}$ (which
% can be verified to be strictly smaller
% than~$\frac{\ln 2\xmd q-1}{\ln p}$).
% %s
% Indeed, in this special case, every node possesses at most 3 outgoing paths
% of length~$2$ that survive the deletion.
% %

% to the one of proof Using a reasoning similar allows to show that in the
%case~$z=\frac{5}{2}$

% \begin{lemma}
%   If~$p> 2q-1$, and~$z\neq \frac{5}{2}$, then~$\frac{\ln 2}{\ln z}<\frac{\ln(2q-1)}{\ln(p)}$.
% \end{lemma}
% \begin{proof} FIXME
%
% \end{proof}

% \input{extra_lemma}

%%%%%%%%%%%%%%%%%%%%%%%%%%%%%%%%%%%%%%%%%%%%%%%%%%%%%%%%%%%%%%%%%%%%%%%%%%%%%%%%%%%%%%%%%%%%%%%%%%%
%%%%%%%%%%%%%%%%%%%%%%%%%%%%%%%%%%%%%%%%%%%%%%%%%%%%%%%%%%%%%%%%%%%%%%%%%%%%%%%%%%%%%%%%%%%%%%%%%%%
%%%%%%%%%%%%%%%%%%%%%%%%%%%%%%%%%%%%%%%%%%%%%%%%%%%%%%%%%%%%%%%%%%%%%%%%%%%%%%%%%%%%%%%%%%%%%%%%%%%

\section{Conclusion }

We have seen with \rtheorem{dpq=xi} that the function~$\Der$, which
transforms a bottom word of~$\Tpq$ into another one, is realised by a
transducer which is so to speak built upon~$\Tpq$ itself.
To tell the truth, we had in mind a stronger property when we began
this work.

All bottom words of~$\Tpq$ are distinct.
But we conjecture that they all share something in common, that they
are all of the `same kind'.
Two infinite words would be considered very naturally to be of the
same kind if they can be mapped one to the other by a finite state
machine.
It is obviously the case for~$\minword{n}$ and~$\minword{m}$ if one
is a suffix of the other, that is, if~$m$ is a node that is reached
from~$n$ by its bottom word.
We conjectured it is the case for every pair of integers~$n$ and~$m$
but were not able to prove it.
We thus leave it as an open problem:

\begin{problem}\lproblem{dpq=xi}
Prove, or disprove, the following statement:

Let~$p,q$ be two coprime integers such that~$p>q>1$ and~$z=\pq$.
For every integer~$n$, there exists a \emph{finite} letter-to-letter
and cosequential transducer~$\En$ (which depends also on~$z$ of
course) such that
$\msp\En(\minword{n}) = \minword{n+1}\msp$.
\end{problem}

Another problem that is left open by this work is the computation of
the Hausdorff dimension of the set~$\adh{\Spq}$ in the cases
where~$p>2q-1$, along the line of \rproposition{haus}.
We have seen that in this cases the set~$\adh{\Spq}$ may be described
in a way comparable to the construction of the classical ternary
Cantor set.
As a result, both sets have similar topological properties (closed, bounded,
empty interior, no isolated point, Lebesgue-measure zero).
This comparison hence suggests that the Hausdorff dimension
of~$\adh{\Spq}$ could be~$\frac{\ln(2\xmd q-1)}{\ln p}$.
We showed an upper bound that is strictly smaller than this last
value.
The exact computation of the Hausdorff dimension seems to be more
difficult and is the subject of ongoing work by the authors.
The first attempts lead to the following conjecture.

\begin{conjecture}If~$p>2\xmd q-1$, then the Hausdorff dimension of~$\adh{\Spq}$ is equal
  to~$\msp\frac{\ln (2\xmd q -1)-\ln q}{\ln p - \ln q}\msp$.
\end{conjecture}

\smallskip

%%%%%%%%%%%%%%%%%%%%%%%%%%%%%%%%%%%%%%%%%%%%%%%%%%%%%%%%%%%%%%%%%%%%%%%%%%%%%%%%%%%%%%%%%%%%%%%%%%%
%%%%%%%%%%%%%%%%%%%%%%%%%%%%%%%%%%%%%%%%%%%%%%%%%%%%%%%%%%%%%%%%%%%%%%%%%%%%%%%%%%%%%%%%%%%%%%%%%%%
%%%%%%%%%%%%%%%%%%%%%%%%%%%%%%%%%%%%%%%%%%%%%%%%%%%%%%%%%%%%%%%%%%%%%%%%%%%%%%%%%%%%%%%%%%%%%%%%%%%

\acknowledgments
The authors are very grateful to the unknown referee who suggested
them to study the Hausdorff dimension of the span-sets and hinted
at the bound from which they began to work.

The second author gratefully acknowledges the support of a Marie
Sk{\l}odowska--Curie post-doctoral fellowship, co-funded by the European Union and
the University of Li\`ege (Belgium), while he was completing this work
during the academic year~2016/2017.

\nocite{BertRigo10-b}
\bibliographystyle{abbrvnat}
\bibliography{Alexandrie-abbrevs.bib,Alexandrie-AC.bib,Alexandrie-DF.bib,Alexandrie-GL.bib,Alexandrie-MR.bib,Alexandrie-SZ.bib,CANT.bib,extra.bib}

% \clearpage
% \input{notation}
\end{document}